\documentclass[%
reprint,
superscriptaddress,
 amsmath,amssymb,
 aps,
floatfix,
]{revtex4-1}

\usepackage{graphicx}% Include figure files
\usepackage{dcolumn}% Align table columns on decimal point
\usepackage{bm}% bold math
\usepackage{dsfont}% double-struck symbols
\usepackage{amsthm}% provides proof environment
\usepackage{siunitx} %units
\DeclareSIUnit{\calorie}{cal}
\usepackage[version=4]{mhchem}

\usepackage{booktabs}
\AtBeginDocument{ %fix revtex + booktabs conflict
\heavyrulewidth=.08em
\lightrulewidth=.05em
\cmidrulewidth=.03em
\belowrulesep=.65ex
\belowbottomsep=0pt
\aboverulesep=.4ex
\abovetopsep=0pt
\cmidrulesep=\doublerulesep
\cmidrulekern=.5em
\defaultaddspace=.5em
}

\usepackage{float}
\usepackage{tabularx}
\usepackage[capitalise]{cleveref}
\crefname{figure}{Fig.}{Fig.}

\newtheorem{theorem}{Theorem}

\theoremstyle{definition}

\theoremstyle{theorem}

\theoremstyle{definition}

\usepackage{color}

\def\ket#1{\left\vert #1 \right\rangle}
\def\bra#1{\left\langle #1 \right\vert}

\newcommand{\be}{\begin{equation}}
\newcommand{\ee}{\end{equation}}
\newcommand{\bp}{\begin{pmatrix}}
\newcommand{\ep}{\end{pmatrix}}
\newcommand{\ben}{\begin{enumerate}}
\newcommand{\een}{\end{enumerate}}

\begin{document}

\title{Measurement reduction in variational quantum algorithms}% Force line breaks with \\
\author{Andrew Zhao}
\affiliation{Center for Quantum Information and Control, Department of Physics and Astronomy, University of New Mexico, Albuquerque, New Mexico 87106, USA}%

\author{Andrew Tranter}
\affiliation{Department of Physics and Astronomy, Tufts University, Medford, Massachusetts 02155, USA
}%
\author{William M. Kirby}
\affiliation{Department of Physics and Astronomy, Tufts University, Medford, Massachusetts 02155, USA
}%
\author{Shu Fay Ung}
\affiliation{California Institute of Technology, Pasadena, California 91125, USA}
\author{Akimasa Miyake}
\affiliation{Center for Quantum Information and Control, Department of Physics and Astronomy, University of New Mexico, Albuquerque, New Mexico 87106, USA}%
\author{Peter J. Love}
\email{peter.love@tufts.edu}
 \altaffiliation[Also at ]{Brookhaven National Laboratory.}
\affiliation{Department of Physics and Astronomy, Tufts University, Medford, Massachusetts 02155, USA
}%

\begin{abstract}
Variational quantum algorithms are promising applications of noisy intermediate-scale quantum (NISQ) computers. These algorithms consist of a number of separate prepare-and-measure experiments that estimate terms in a Hamiltonian. The number of terms can become overwhelmingly large for problems at the scale of NISQ hardware that may soon be available. We
%approach this problem from the perspective of contextuality, and
use unitary partitioning (developed independently by Izmaylov \emph{et al.}~[J.~Chem.~Theory Comput.~\textbf{16}, 190 (2020)]) to define variational quantum eigensolver procedures in which additional unitary operations are appended to the ansatz preparation to reduce the number of terms. This approach may be scaled to use all coherent resources available after ansatz preparation. We also study the use of asymmetric qubitization to implement the additional coherent operations with lower circuit depth.  %We investigate this technique for lattice Hamiltonians, random Pauli Hamiltonians, and electronic structure Hamiltonians. 
Using this technique, we find a constant factor speedup for lattice and random Pauli Hamiltonians. For electronic structure Hamiltonians, we prove that linear term reduction with respect to the number of orbitals, which has been previously observed in numerical studies, is always achievable.  For systems represented on 10--30 qubits, we find that there is a reduction in the number of terms by approximately an order of magnitude. Applied to the plane-wave dual basis representation of fermionic Hamiltonians, however, unitary partitioning offers only a constant factor reduction.  Finally, we show that noncontextual Hamiltonians may be reduced to effective commuting Hamiltonians using unitary partitioning.
\end{abstract}

\pacs{Valid PACS appear here}
\maketitle

\section{\label{intro}Introduction}

Quantum simulation is a promising application of future quantum computers~\cite{Feynman1982,lloyd1996universal,abrams1997simulation,abrams1999quantum}. Applications in materials science, chemistry, and high-energy physics offer the prospect of significant advantages for simulation of quantum systems~\cite{wu2002polynomial,aspuru2005simulated,preskill1}. Calculations on quantum computers that would challenge the classical state of the art require large-scale, error-corrected quantum computers~\cite{babbush2018encoding}. However, quantum hardware is entering the noisy intermediate-scale quantum (NISQ) era~\cite{preskill2018quantum}, in which the machines are still too small to implement error correction but are already too large to simulate classically~\cite{arute2019quantum}. It is natural to ask whether NISQ computers can perform useful tasks in addition to demonstrations of quantum supremacy~\cite{arute2019quantum,boixo2016characterizing,harrow2017quantum}.

The variational quantum eigensolver (VQE) was developed to enable quantum estimation of ground state energies on noisy small-scale quantum computers~\cite{peruzzo2014variational}. VQE was developed as a method for quantum simulation of electronic structure~\cite{peruzzo2014variational} and concurrently as a simulation method for quantum field theory by cavity QED~\cite{barrett2013simulating}. Contemporaneously, the quantum approximate optimization algorithm (QAOA) was developed as a variational approach to approximate solutions of classical optimization problems~\cite{farhi2014quantum}.
VQE has been widely implemented experimentally due to its simplicity and suitability for NISQ devices~\cite{peruzzo2014variational,wang2015quantum,omalley16a,kandala2017hardware,PhysRevX.8.031022,dumitrescu18a}.

VQE consists of preparation of a variational ansatz state by a low-depth parameterized quantum circuit, followed by estimation of the expectation values of the terms in the Hamiltonian, obtained by measuring each separately. This process is repeated until the statistical error on the expectation value of each term is less than some desired precision threshold. Thus, in VQE the long coherent evolutions of phase estimation are replaced by many independent and short coherent evolutions. However, the necessary number of independent measurements may become overwhelmingly large for problem sizes of ${\sim}50$ qubits, which may soon be accessible. Recently, there has been much activity in addressing this measurement problem, via numerous approaches~\cite{babbush2018low,rubin2018application,wangAcceleratedVariationalQuantum2019,verteletskyi2019measurement,jena2019pauli,izmaylov19a,yen2019measuring,huggins2019efficient,gokhale2019MinimizingStatePrep,bonetNearlyOptimalMeas2019,crawford2019efficient,gokhale2019n,torlai2019precise}. In the present paper, we consider the use of extra coherent resources to reduce the number of separate Pauli terms whose expectation values must be estimated. We refer to this process as \emph{term reduction}.  Our methods are closely related to those introduced in~\cite{izmaylov19a,bonetNearlyOptimalMeas2019}, which we discuss later.

We consider throughout a $k$-local Pauli Hamiltonian on $n$ qubits:
\begin{equation}\label{ham}
 H = \sum_{j=1}^m \alpha_j  P_j,   
\end{equation}
where the $m$ terms $ P_j \in \{ I,X,Y,Z \}^{\otimes n} $ are $k$-local Pauli operators, i.e., tensor products of the Pauli matrices and the $2\times2$ identity containing at most $k$ nonidentity tensor factors. This $k$-locality does not refer to any geometrical locality of the layout of the physical qubits. 

The Hamiltonian $H$ for $1\leq k\leq n$ and $1\leq m\leq 4^{n}$ can represent any qubit observable. Interesting cases occur for $k$ a small constant ($2\leq k\leq 4$)~\cite{farhi2014quantum} and for $k$ scaling logarithmically with $n$~\cite{bravyi02a,love12}. Jordan--Wigner mappings of fermions to qubits generate Hamiltonians with $k\leq n$, albeit of a restricted form and in which $m$ is still a polynomial in $n$~\cite{somma02a}. Techniques to map interesting physical Hamiltonians to Pauli Hamiltonians show that the Hamiltonian $H$ is expressive enough to represent problems in physics and chemistry ranging from condensed-matter models to molecular electronic structure to quantum field theory. Restricting to \cref{ham} is therefore not a significant limitation on the applicability of our results to the simulation of quantum systems.

Assuming measurements are to be performed in the $z$ basis on individual qubits, to simulate the terms of \cref{ham} it is necessary to map each $ P_j$ to a measurement in the computational basis (given by the tensor product of the $z$ bases for each qubit). If our NISQ device has all-to-all pairwise connectivity (as is the case for ion trap NISQ devices) then we require $k-1$ CNOT gates and up to $k$ single-qubit Clifford operations to reduce our measurement of a $k$-local Pauli operator $ P_j$ to a $z$-basis measurement~\cite{nielsen2002quantum}. If our NISQ computer has only nearest-neighbor connectivity on the line we may require an additional $O(n)$ CNOT gates to swap the qubits into an adjacent set.

Any completely commuting set of Pauli operators $S_C$ may be mapped to a set of Pauli words over $Z$ and the identity by mapping the common eigenbasis of $S_C$ to the computational basis~\cite{nielsen2002quantum}. Previous works have studied this as a method for reducing the number of measurements;~the resulting technique requires an additional $ O(n^2) $ gates, with numerical evidence for an $ O(n) $ measurement count reduction~\cite{gokhale2019MinimizingStatePrep,yen2019measuring}.
Because the eigenbasis of $S_C$ is a set of stabilizer states (with stabilizers given by elements of $S_C$ up to a sign), this map is a Clifford operation. 
Clifford operators are known to lack transformation contextuality~\cite{love17a}, i.e., they are describable by positive maps on Wigner functions.

Furthermore, Clifford operations map single Pauli operators to single Pauli operators, which means that if we desire to reduce the number of terms in the Pauli Hamiltonian Eq.~\eqref{ham}, our map must possess some non-Clifford structure.
Hence it must in general possess transformation contextuality.

We describe two methods for term reduction based on such transformations. The first technique, unitary partitioning, was previously and independently obtained in~\cite{izmaylov19a,bonetNearlyOptimalMeas2019}. Our second technique provides a more efficient realization of the required transformations at the cost of some ancilla state preparation using asymmetric qubitization---an extension of the linear combination of unitaries model~\cite{LCU2012}---introduced in~\cite{babbush2019SYK}. We present these two methods in \cref{termreduction}. In \cref{examples} we evaluate the method for several classes of Hamiltonians. \cref{sec4} is devoted to analyzing electronic structure Hamiltonians in depth. We confirm and extend the previous numerical results of~\cite{izmaylov19a} observing that a linear term reduction with respect to the number of orbitals is possible. We prove that this linear reduction can always be achieved. We also show in \cref{dual} that unitary partitioning offers a constant factor reduction in the number of terms of a fermionic Hamitonian expressed in the plane wave dual basis defined in~\cite{babbush2018low}. Then, in Section \ref{noncontextual} we show that noncontextual Hamiltonians, defined in~\cite{kirby19a} (also studied in~\cite{raussendorf2019phase}), are reducible to commuting Hamiltonians under unitary partitioning. We close the paper with discussion and directions for future work.

\section{Term reduction for Pauli Hamiltonians}
\label{termreduction}

Given a Hamiltonian of the form \cref{ham}, we wish to reduce the number of distinct expectation values to estimate in a VQE experiment using the coherent operations of the quantum computer.
Suppose that our ansatz $|\psi_A\rangle$ is prepared by a quantum circuit $ U$ from the state $|\psi_0\rangle\equiv\ket{0}^{\otimes n}$ so that
\be
\ket{\psi_A}= U\ket{\psi_0}.
\ee
Then our experiment estimates the expectation values
\be
\langle P_j\rangle =\bra{\psi_0} U^\dagger P_j U\ket{\psi_0}.
\ee

Suppose instead we rewrite our Hamiltonian in terms of a different set of Pauli operators $\{ Q_l\}_{l=1}^{m_c}$ and unitary operations $\{ R_l\}_{l=1}^{m_c}$ as follows:
\be
 H =\sum_{j=1}^m \alpha_j P_j = \sum_{l=1}^{m_c}\gamma_l
 R_l^\dagger Q_l R_l.
\ee
Such decompositions give the correct variational estimate:
\begin{eqnarray}
\bra{\psi_A}  H\ket{\psi_A}&=&\sum_{j=1}^m \alpha_j\bra{\psi_A} P_j\ket{\psi_A}\\
&=& \sum_{l=1}^{m_c}\gamma_l \bra{\psi_A} R_l^\dagger Q_l R_l \ket{\psi_A}.
\end{eqnarray}
Each term labeled by $l$ is estimated by a separate prepare and measure ansatz which appends a different unitary $R_l$ to the ansatz preparation. The unitary rotations $ R_l$ therefore represent the additional coherent resources required to reduce the number of separate expectations to be obtained.

Unlike the approach of~\cite{izmaylov19a}, we do not estimate the unitary operators $ R_l^\dagger Q_l R_l$ themselves. Instead, we propose to perform a set of $m_c$ experiments in which the coherent operations $ R_l$ are appended to $ U$, so that the expectation values are obtained by measuring $ Q_l$ in the resultant state. In this case, the $ R_l$ may be made as simple or as complex as the coherent resources available after the state preparation circuit allow. Term reduction therefore allows the use of VQE for larger systems by optimally using the increasing amount of coherent resources available in new devices.

\subsection{Unitary partitioning}
\label{unitarypartitioning}

We will apply rotations in the adjoint representation of $\mathfrak{su}(2^n)$ with the goal of reducing the number of Pauli terms in the Hamiltonian.
For classical algorithms the number of such terms is not a relevant variable, as one must represent all the nonzero terms of the Hamiltonian in some way. There are some general constraints on the form of terms arising from a Pauli matrix by an adjoint unitary action. We now consider what resources the $ R_l$ operations require and give constructions that achieve term reduction. These ideas were previously presented in~\cite{izmaylov19a}. 

We may write
\be
 R_l^\dagger Q_l R_l =\sum_j \beta_{lj} P_{f(l,j)},
\ee
where $f$ is a relabeling of generalized Pauli matrices.
Any unitary rotation of a generalized Pauli matrix is self-inverse, so $( R_l^\dagger Q_l R_l)^2={\openone}$, which implies
\be
\label{constraints}
\sum_j \beta_{lj}^2 = 1~~{\rm and}~~\sum_{j<k}\beta_{lj}\beta_{lk}\{ P_{f(l,j)}, P_{f(l,k)}\}=0.
\ee
The first constraint can be satisfied for any subset of terms by scaling the coefficients $\beta_{lj}$ by appropriately defining $\gamma_l$. The second constraint is the defining property of subsets of terms which can be combined into a single term by unitary rotation.
For the technique discussed in this section, we divide the terms of the Hamiltonian into sets in which the operators pairwise anticommute; we call such sets \emph{completely anticommuting sets}.
The second constraint in Eq.~\eqref{constraints} is trivially satisfied within each such set.
We then rescale these terms to satisfy the first constraint and seek unitary operators that map each set to a single Pauli operator.

The \emph{compatibility graph} associated to a set of Pauli operators is an undirected graph whose vertices are the operators in the set, and in which a pair of vertices is connected if the associated operators commute. Completely anticommuting sets of Pauli operators are independent sets of the compatibility graph. A partition of the operators into completely anticommuting sets is provided by a coloring of the vertices of the graph  such that no two vertices connected by an edge have the same color. The number of sets is determined by the number of colors. Graph coloring is a well-known NP-complete problem; however, we only require the number of colors to be less than the number of vertices for our method to provide a reduction in the number of terms. A detailed study of the use of various heuristics for graph coloring for the compatibility graphs of Hamiltonians was performed in~\cite{izmaylov19a}.

We now construct the rotation $ R$ that maps a completely anticommuting set to a single Pauli operator by conjugation.
Let $S$ be a set of Pauli operators appearing in the Hamiltonian such that $\{ P_j, P_k\}=0$ $\forall  P_j \neq  P_k\in S$. It will also be useful to define $s=|S|$. The set of terms corresponding to $S$ in the Hamiltonian is then written
\be
 H_S = \sum_{ P_j\in S} \beta_j  P_j.
\ee
We will assume for now that the coefficients satisfy
\be
\sum_j\beta_j^2=1.
\ee
We define the following Hermitian, self-inverse operators:
\be
\mathcal{X}_{sk} = i  P_s P_k, \quad 1\leq k\leq s-1. \label{eqn:xOperator}
\ee
It is straightforward to verify that $ \mathcal{X}_{sk}$ commutes with all $ P_j\in S$ for $j\neq s$, $j\neq k$, and that it anticommutes with  $ P_k$ and $ P_s$. 

We define the adjoint rotation generated by $ \mathcal{X}_{sk}$:
\be
 R_{sk}=\exp\left(-i\frac{\theta_{sk}}{2} \mathcal{X}_{sk}\right), \label{eqn:adjointRotation}
\ee
whose action on the terms in $H_S$ is given by
\be
\begin{split}
 R_{sk}  P_k R_{sk}^\dagger &=\cos\theta_{sk} P_k+\sin\theta_{sk} P_s,\\
 R_{sk}  P_s R_{sk}^\dagger &=-\sin\theta_{sk} P_k+\cos\theta_{sk} P_s.%\\
% R_{sk}^\dagger  P_l R_{sk}&= P_k~~l\neq k,~~l\neq s.\\ % Is this last line correct?
\end{split}
\ee
That is, $ R_{sk}$ is an adjoint rotation acting in the space spanned by $ P_s$ and $ P_k$.

If we act on $ H_S$ with $ R_{sk}$, we obtain
\be
\begin{split}
 R_{sk}  H_S  R_{sk}^\dagger =&~(\beta_k\cos\theta_{sk} - \beta_s\sin\theta_{sk}) P_k\\
&+(\beta_k\sin\theta_{sk} + \beta_s\cos\theta_{sk}) P_s\\
&+\sum_{ P_j\in S\setminus \{ P_k, P_s\}} \beta_j P_j.
\end{split}
\ee
Choosing $\beta_k\cos\theta_{sk}=\beta_s\sin\theta_{sk}$ therefore gives a rotation of the Hamiltonian with the $P_k$ term removed and with the norm of the term $ P_s$ increased from $\beta_s$ to $\sqrt{\beta_s^2+\beta_k^2}$.
Defining the operator
\be
\label{axis}
 R_S =  R_{s(s-1)}(\theta_{s(s-1)})\cdots R_{s2}(\theta_{s2}) R_{s1}(\theta_{s1}),
\ee
where the angles $\theta_{sk}$ satisfy
\begin{equation}
    \beta_1\cos\theta_{s1}=\beta_s\sin\theta_{s1},
\end{equation}
and, for $k > 1$,
\be
\beta_k\cos\theta_{sk}=\sqrt{\left(\beta_s^2 + \sum_{j=1}^{k-1}\beta_j^2\right)} \sin\theta_{sk}, \label{eqn:theta}
\ee
therefore gives
\be
\begin{split}
 R_{S}  H_S  R_{S}^\dagger &=  P_s,
\end{split}
\ee
where we used the fact that $\sum_{j=1}^s\beta_j^2=1$.  Care must be taken when choosing $\theta_{sk}$ so as to obtain the positive root.

Our decomposition strategy is therefore the following:
\be
\label{hamiltoniandecomp}
 H=\sum_{j=1}^m \alpha_j P_j=\sum_{l=1}^{m_c}\gamma_{l} H_{S_l},
\ee
where
\be
 H_{S_l} = \sum_{ P_j \in S_l} \beta_{lj}  P_j
\ee
has support on a set $S_l$ of self-inverse operators for which $\{ P_j, P_k\}=0$ $\forall j\neq k$ and $\sum_{j}\beta_{lj}^2=1$. Each ${H}_{S_l}$ can be obtained from a single Pauli operator by a unitary rotation as in \cref{axis}, so we can rewrite \cref{hamiltoniandecomp} as
\be
\label{rotatedhamiltonian}
 H = \sum_{l=1}^{m_c}\gamma_l R_{S_l}^\dagger P_{s_l} R_{S_l},
\ee
where the $ R_{S_l}$ operators are given for each set of pairwise anticommuting operators by \cref{axis}. 

For each $ H_{S}$ we must therefore append to our ansatz preparation the set of $s-1$ operators $ R_{sk}$ (recall that $s=|S|$). For an $l$-local Hamiltonian, each of these requires $O(l)$ CNOT and single-qubit rotations to implement. Hence one exchanges $s$ separate Pauli expectation value estimations for a single expectation value estimation, at the cost of $O(sl)$ additional coherent operations. Note that directly appending these transformations to the ansatz preparation results in a factor of 2 reduction in the required coherent resources as compared to~\cite{izmaylov19a}, where both $R$ and $R^\dagger$ must be implemented as controlled operations. 

The decomposition given above and in~\cite{izmaylov19a} is the most direct implementation of the transformation of the Hamiltonian.  Improvement can be made through the use of ancilla qubits and more coherent resources, as we now show in Section~\ref{rotations}.

\subsection{Low-depth implementation of the rotations}
\label{rotations}

In Section~\ref{unitarypartitioning} and in Ref.~\cite{izmaylov19a}, an ordered sequence of rotations is used to write a completely anticommuting set of Pauli operators as a single term. Here we will show how to use a single rotation to perform the same reduction, and show how to implement this rotation using the methods based on linear combinations of unitaries (LCU)~\cite{LCU2012}.

We define a set of operators $ H_k$ for $1\leq k\leq n$ such that $ H_1= P_1$, $ H_n=\sin\phi_{n-1}  H_{n-1}+ \cos\phi_{n-1}  P_n$. Each $ H_n$ is self-inverse, and we consider rotations of $ H_n$ around an axis that is Hilbert--Schmidt orthogonal to both $ H_{n-1}$ and $ P_n$. 
The operator defining this axis is:
\be\label{axis_b}
 \mathcal{X} = \frac{i}{2}\left[ H_{n-1}, P_n\right]. 
\ee
The operator $ \mathcal{X}$ is self-inverse, anticommutes with $ H_n$, and so $[\mathcal{X}, H_n]=2 \mathcal{X}  H_n$. Furthermore, we may show that
\begin{equation}
 \mathcal{X} H_n = i(-\sin\phi_{n-1} P_n+\cos\phi_{n-1} H_{n-1}).
\end{equation}
The operator $ \mathcal{X}$ generates the rotation 
\begin{equation}\label{req}
 R=\exp(-i\alpha  \mathcal{X}/2)=\cos(\alpha/2){\openone}-i\sin(\alpha/2)  \mathcal{X}.
\end{equation} 
The adjoint action of $ R$ on $ H_n$ is given by
\be
 R  H_n  R^\dagger = \sin(\phi_{n-1}-\alpha) H_{n-1}+\cos(\phi_{n-1}-\alpha) P_n.
\ee
Choosing $\alpha=\phi_{n-1}$ therefore gives $ R  H_n  R^\dagger =  P_n$. This is a simple constructive demonstration that any self-inverse operator supported on a set of pairwise anticommuting operators $S$ can be mapped to a single Pauli operator. (The details of these calculations can be found in Appendix~\ref{app:comp_X}.)

The terms in the operator $ \mathcal{X}$ all pairwise anticommute, and $ \mathcal{X}$ squares to the identity. This yields the expression for $ R$ given in \cref{req}. As a linear combination of Pauli operators, which are unitary, this naturally suggests implementation of $ R$ using the LCU method~\cite{LCU2012}. These methods can be combined with qubitization and quantum signal processing to reduce the required gate count~\cite{low16a,low17,poulin18a,babbush2018encoding}. However, $ \mathcal{X}$ has coefficients that are $\ell_2$-normalized, whereas the standard LCU methods naturally treat Hamiltonians with $\ell_1$-normalized coefficients. Fortunately, this issue was already addressed in Ref.~\cite{babbush2019SYK}, in which an asymmetric LCU (ALCU) method was introduced. We propose the ALCU method for the implementation of $ R$. Because $ R$ is equivalent to evolution under the Hamiltonian $ \mathcal{X}$, the cost of asymmetric qubitization scales as the square root of the number of terms in $ \mathcal{X}$, and hence the use of this method offers a quadratic speedup in asymptotic scaling compared to the methods of Section~\ref{unitarypartitioning} and Ref.~\cite{izmaylov19a}. 

ALCU requires $O(\log s)$ additional qubits ($s$ being the maximum size of any of the anticommuting sets) and more complex gate operations than the method of Section~\ref{unitarypartitioning} and~\cite{izmaylov19a}. However, the use of these methods in the context of VQE provides a motivation to implement more sophisticated quantum algorithms on NISQ devices. It should be noted that implementation of ALCU for this purpose is much simpler than its use for direct simulation of time evolution under the original Hamiltonian. This is because the number of terms in $ \mathcal{X}$ is only equal to the number of terms in an anticommuting set. As we discuss in detail below, this can be made smaller in order to take advantage of any additional coherent resources available after state preparation.

\subsection{Commuting terms}

Requiring that the sets of terms to be combined anticommute, as in Sections~\ref{unitarypartitioning} and \ref{rotations}, is sufficient but not necessary to perform term reduction. If there is additional structure on the coefficients of the Hamiltonian, the second constraint in \cref{constraints} may be satisfied without the individual terms all vanishing. Here we consider the possibility that for some $l$,
\be
\sum_{j<k}\beta_{lj}\beta_{lk}\{ P_j, P_k\}=0,
\ee
while the individual terms are nonzero (note that we have simplified the labeling of the Pauli terms). Because generalized Pauli matrices have the property that they either commute or anticommute, we can restrict attention to the subset of the operators that commute. We then require that
\be
\sum_{j<k}\beta_{lj}\beta_{lk}\{ P_j, P_k\}=2\sum_{S(l,j,k)} \beta_{lj}\beta_{lk}  P_j P_k=0,
\ee
where $S(l,j,k)$ is the set of indices satisfying $j<k$ and $[ P_j, P_k]=0$. Each term here is nonzero, so the condition must be enforced by cancellation of pairs, i.e., due to relations of the form
\be
\beta_{lj}\beta_{lk}  P_j P_k + \beta_{ls}\beta_{lr}  P_s P_r= 0.
\ee
This can only be true if $|\beta_{lj}\beta_{lk}|=|\beta_{ls}\beta_{lr}|$, and so this possibility of term reduction depends on the details of the coefficients more sensitively than simply requiring all terms to anticommute in a particular subset. 

Supposing that the conditions on pairs of coefficients are satisfied, we also require that
\be
 P_j P_k \pm  P_s P_r =0
\ee
(for $\beta_{lj}\beta_{lk}=\pm \beta_{ls}\beta_{lr}$).
Suppose the pairs $(j,k)$ and $(s,r)$ have one operator in common, $j=s$. Then our requirement is $ P_k=\pm  P_r$, meaning that $(j,k)$ and $(s,r)$ are the same pair. Hence the pairs $(j,k)$ and $(s,r)$ must be completely distinct. This implies that $ P_j P_k= P_t$ and $\pm  P_s P_r= P_t$. This is perfectly possible:~for example, if $ P_k=IX$, $ P_j=XI$, $ P_r=ZZ$, and $ P_s=YY$, then $ P_k P_j=XX$ and $ P_r P_s=-XX$. We leave further investigation of this possibility for term reduction to future work.

\subsection{Total measurement cost estimates}

Achieving precision $ \epsilon $ in the estimate of the expectation value $ \langle {H} \rangle $ requires a statistically significant sample of qubit measurements for each Pauli term in $ {H} $. Naively, this requires approximately $ |\alpha_j|^2/\epsilon^2 $ measurements for the $ j $th term, where $ \alpha_j $ is its associated weight. However, it was proposed in~\cite{PhysRevA.92.042303}, and formally proven in~\cite{rubin2018application}, that the optimal number of measurements per term is
\be
\label{singletermmeasurements}
M_j = \frac{|\alpha_j|\sigma_j}{\epsilon^2} \left( \sum_{k=1}^m |\alpha_k| \sigma_k \right),
\ee
where $ \sigma_j^2 = \langle {P}_j^2 \rangle - \langle {P}_j \rangle^2 $ is the operator variance of the $ j $th term. Using $ \sigma_j^2 \leq 1 $ for all self-inverse operators, the upper bound for the total number of measurements to estimate the full Hamiltonian is~\cite{rubin2018application}
\be
M = \sum_{j=1}^m M_j = \left( \frac{1}{\epsilon} \sum_{j=1}^m |\alpha_j|\sigma_j \right)^2 \leq \frac{\Lambda^2}{\epsilon^2},
\ee
where $ \Lambda = \sum_{j=1}^m |\alpha_j| $ is the $ \ell_1 $-norm of the Hamiltonian weights.

Using the standard inequalities
\be
\label{lpinequality}
\frac{1}{\sqrt{d}}\|{x}\|_1 \leq \|{x}\|_2 \leq \|{x}\|_1
\ee
for any $ x \in \mathbb{R}^d $, where $ \| \cdot \|_p $ denotes the $ \ell_p $-norm, we may establish bounds for the value of $ \Lambda^2 $ after transforming the Hamiltonian via unitary partitioning. We reuse the notation of \cref{hamiltoniandecomp,rotatedhamiltonian}, so that
\be
{H} = \sum_{j=1}^m \alpha_j P_j
\ee
is the Hamiltonian as given, and
\be
{H} = \sum_{l=1}^{m_c}\gamma_l R_{S_l}^\dagger P_l R_{S_l}
\ee
is its form after unitary partitioning. Note that $ R_{S_l}^\dagger P_l R_{S_l} $ is self-inverse, so the variances remain bounded by 1. Since the coefficients associated with each anticommuting set $ S_l $ must be $ \ell_2 $-normalized, we have
\be
\gamma_l^2 = \sum_{k\in S_l} \alpha_k^2.
\ee
By abuse of notation, here we use $ S_l $ to denote the index set on which its elements are supported.

Let $ \Lambda $ be the $ \ell_1 $-norm of the weights $ \{ \alpha_j \}_{j=1}^m $ as before, and $ \Lambda_c $ be the $ \ell_1 $-norm of $ \{ \gamma_l \}_{l=1}^{m_c} $. Then, using the right-hand inequality of \cref{lpinequality}, we obtain
\be
\label{Lambdaupper}
\begin{split}
\Lambda_c = \sum_{l=1}^{m_c} |\gamma_l| &= \sum_{l=1}^{m_c} \sqrt{ \sum_{k\in S_l} \alpha_k^2 } \\
&\leq \sum_{l=1}^{m_c} \sum_{k\in S_l}  |\alpha_k| \\
&= \sum_{j=1}^{m} |\alpha_j| = \Lambda.
\end{split}
\ee
Thus $ \Lambda_c \leq \Lambda $, and in fact this bound is saturated only if no partitioning is performed at all.

Applying the left-hand inequality of \cref{lpinequality} to the first line of \cref{Lambdaupper} yields
\be
\label{Lambdalower1}
\sum_{l=1}^{m_c} \Bigg( \frac{1}{\sqrt{|S_l|}} \sum_{k\in S_l}  |\alpha_k| \Bigg) \leq \Lambda_c.
\ee
Let $ s_{\mathrm{max}} = \max_{l} |S_l| $ be the size of the largest set in the partition. Then
\be
\label{Lambdalower2}
\frac{1}{\sqrt{s_{\mathrm{max}}}}  \sum_{l=1}^{m_c} \sum_{k\in S_l}  |\alpha_k| = \frac{\Lambda}{\sqrt{s_{\mathrm{max}}}} \leq \Lambda_c.
\ee
Bounding the set sizes by $ s_{\mathrm{max}} $ is fairly tight if they are all roughly equal, which is both desirable (since the gate complexity scales with the set size) and always possible (one may take a large set and simply divide it into smaller ones, which remain fully anticommuting). Roughly speaking, the number of measurements $ M_c $ may be thought of as being lower bounded by $ M/s_{\mathrm{max}} $, although this is not the whole story, since $ \Lambda $ (resp.~$ \Lambda_c $) is itself an upper bound estimate for $ M $ (resp.~$ M_c $). Equation~\eqref{Lambdalower2} gives only an approximate sense for the maximum amount of measurement reduction possible by unitary partitioning when taking into account the statistical repetitions.

It is worth noting that this lower bound is saturated when $ |\alpha_j| = |\alpha_k| $ $ \forall j,k $. In fact, a weaker condition saturates the tighter bound of \cref{Lambdalower1}. There we require only that $ |\alpha_j| = |\alpha_k| $ $ \forall j,k \in S_l $ for each $ l $---that is, the coefficient magnitudes are uniform within each set. Supposing that this approximately holds, and again that all $ |S_l| $ are roughly the same, yields $ \Lambda_c \approx \Lambda/\sqrt{s_{\mathrm{max}}} $.

Thus partitioning with additional constraints respecting these coefficient conditions may result in more measurement reduction, without requiring any additional coherent rotations. The partitioning algorithm would then require significantly more classical computational resources, as this is now a \emph{weighted} graph coloring problem, but in principle these ideas may be implemented straightforwardly. For the analysis in the following section, we focus only on the number of unique Hamiltonian terms before and after partitioning as a rough estimate for the amount of measurement reduction achieved by our method.

\section{Preliminary Applications}
\label{examples}

\subsection{Transverse-field Ising model in one dimension}

To give a simple realization of these ideas we consider the transverse-field Ising model (TIM) on a one-dimensional lattice with $L$ sites and periodic boundary conditions:
\be
H = \sum_{j=1}^L (Z_{j+1}Z_j + x X_j).
\ee
No pair of $Z$ terms and no pair of $X$ terms can be in the same anticommuting set, so we choose pairs of anticommuting operators composed of $Z_{j+1}Z_j$ and $X_{j+1}$.
We then write:
\begin{equation}
Z_{j+1}Z_j + x X_{j+1}I_j = \sqrt{1+x^2}\left(\frac{Z_{j+1}Z_j + xX_{j+1}I_j}{\sqrt{1+x^2}} \right) .
%&=&\sqrt{1+x^2}\biggl(\cos \theta ZZ + \sin\theta XI\biggr)
\end{equation}
From \cref{eqn:xOperator,eqn:adjointRotation} we define the operator
\begin{align}
    R_j &= \exp\left(\frac{i\theta}{2}Y_{j+1}Z_j\right) \\
    &= {\rm CNOT}_{(j+1,j)} \times \exp\left(\frac{i\theta}{2}Y_{j+1}I_j\right) ,
\end{align}
where $\theta$ is given from \cref{eqn:theta}:
\begin{equation}
    \frac{x}{\sqrt{1+x^2}} \cos{\theta} = \frac{1}{\sqrt{1+x^2}} \sin{\theta}.
\end{equation}
%and use:
%\begin{eqnarray}
%\cos \theta ZZ + \sin\theta XI &=&UZIU^\dagger\\
%%\cos \theta ZZ + \sin\theta IX &=&VIZV^\dagger
%\end{eqnarray}
%where
%\begin{eqnarray}
%U&=&{\rm CNOT}\times\exp(i\frac{\theta}{2}YI)\\
%%V&=&{\rm CNOT}\times\exp(i\frac{\theta}{2}IY)\\
%\end{eqnarray}
Our final Hamiltonian decomposition is then:
\be
% H = \sum_{i=1}^{L/2}\biggl[U_{2i} Z_{2i}U_{2i}^\dagger + U_{2i+1} Z_{2i+1}U_{2i+1}^\dagger\biggr]
 H = \sum_{j=1}^{L}\left[R_{j}^\dagger Z_{j}Z_{j+1}R_{j}\right].
%\sum_{i=1}^{L/2}\biggl[U_{2i} Z_{2i}U_{2i}^\dagger + V_{2i+1} Z_{2i+1}V_{2i+1}^\dagger\biggr]
\ee
Whereas our initial Hamiltonian had $2L$ terms, our final Hamiltonian has $L$ terms. %While this is only a constant factor in the reduction of terms, it demonstrates that term reduction is possible by this method. 

\subsection{TIM on arbitrary graphs}

If we consider transverse Ising Hamiltonians defined on arbitrary graphs, the analysis does not change substantially. The maximum size of a totally anticommuting set is still 2, independent of the graph, because once a single local $ X_i$ is included in the set, one can include only one $ Z_i Z_j$ term in the set. Hence, the number of terms in a transverse Ising Hamiltonian on a general graph with vertex set $V$ and edges $E$ can be reduced from $|E|+|V|$ to $|E|$. This cannot change the asymptotic scaling of the number of terms as a function of the number of vertices. In particular, for regular graphs with degree $q$ the number of edges is $|V|q/2$ and the number of terms in the transverse Ising model Hamiltonian is $|V|(1+q/2)$, which can be reduced to $|V|q/2$, a constant factor improvement of $q/(q+2)$. Note that this case includes lattice models. The relative lack of performance here is due to the presence of little anticommutative structure in the operators of the transverse-field Ising model.

\subsection{Compatibility graphs of random Hamiltonians}

Randomly choosing Pauli terms from the complete set of $n$-qubit Pauli observables corresponds to selecting a subset of the vertices of the full compatibility graph of all Pauli observables. The resulting compatibility graphs can only be subgraphs of this graph, which has a finite geometric structure considered in~\cite{planat2008pauli}. Therefore randomly sampling Pauli terms, resulting in an edge in the compatibility graph with given probability, say $p$, does not result in Erd\H{o}s--R\'{e}nyi random graphs given by populating edges with probability $p$. The constraint that the graphs arising be subgraphs of the full compatibility graph of all Pauli operators causes this deviation.

However, for large numbers of qubits, fixed locality of operators, and a number of Pauli terms scaling polynomially with the number of qubits, the probability that a randomly sampled pair of Pauli operators commutes should approach $1$ with increasing $n$. In this limit the compatibility graph will be closely approximated by a polynomially sized complete subgraph of the exponentially large compatibility graph of all Pauli operators on $n$ qubits, with a few edges missing. Asymptotically, we expect that the number of colors required will tend to the chromatic number of the complete graph, which is equal to the number of vertices.

As we we shall see in Sections~\ref{r2} and~\ref{rk}, for any fixed random $k$-local Hamiltonian we may write the probability that a randomly sampled pair of terms commute as
\be\label{lim}
p_c\simeq 1-\frac{{\rm const}}{n}.
\ee
The chromatic number of almost all such graphs will be proportional to $n$, and hence we expect at most a constant factor reduction in the number of terms~\cite{bollobas1988chromatic}. The problem of finding commuting cliques of related graphs was discussed in~\cite{jena2019pauli}. Here we study the problem from the context of finding anticommuting sets for unitary partitioning.

\subsection{Random $2$-local Pauli Hamiltonians}
\label{r2}

Consider a $2$-local Pauli Hamiltonian defined on an Erd\H{o}s--R\'{e}nyi random interaction graph with $n$ vertices and $|E|$ edges. A term in the Hamiltonian corresponds to an edge in the set $E$ and a sample drawn uniformly at random from $\{X,Y,Z\}^{\otimes 2}$. We choose Hamiltonians with only one term per edge. Two terms corresponding to edges $e_1$ and $e_2$ from such a Hamiltonian anticommute if
\ben
\item{$e_1\neq e_2$,}
\item{$|e_1\bigcap e_2|=1$.}
\een

What is the probability that $|e_1\bigcap e_2|=1$? There are $n-2$ vertices connected to each vertex of $e_1$ that form edges with $|e_1\bigcap e_2|=1$. There are therefore $2(n-2)$ of the $n(n-1)/2$ possible edges that give $|e_1\bigcap e_2|=1$ for any given $e_1$. The probability of such an incidence is therefore $p_e = 4(n-2)/[n(n-1)]$. 

What is the probability that two terms intersect on one qubit and do not commute? There are nine operators that can be associated with an edge. Examination of this set gives a probability of $2/3$ that tensor factors incident on the same vertex disagree. 

Given a pair of edges from the interaction graph, i.e., a pair of terms in the Hamiltonian, the probability that the associated operators anticommute is therefore
\be\label{2l}
p_a =\frac{8}{3n}\frac{n-2}{n-1}.
\ee

We now analyze the coloring of an Erd\H{o}s--R\'{e}nyi random graph in which edges are populated independently with probability $p$~\cite{erds1960evolution}. As noted above, the compatibility graphs of random Pauli Hamiltonians cannot be Erd\H{o}s--R\'{e}nyi but in the limit of large numbers of qubits we expect these results to be asymptotically correct. Our procedure for defining a random 2-local Pauli Hamiltonian has given us a probability $1-p_a$ that an edge is present in the compatibility graph, because Pauli operators either commute or anticommute. 

Almost every random graph with $m$ vertices drawn from an ensemble where the probability of an edge between any pair of vertices is $1-p_a$ has chromatic number~\cite{bollobas1988chromatic}
\be
\chi=\left(\frac{1}{2}+o(1)\right)\log \frac{1}{p_a}\frac{m}{\log m}.
\ee
This immediately enables us to characterize the performance of our method on random $2$-local Hamiltonians. Suppose the number of terms rises as a power $\tau$ of the number of qubits $m=O(n^\tau)$.
Then the fractional improvement $m_c/m$ in the number of terms in the Hamiltonian will be
\be
\frac{m_c}{m}=\left(\frac{1}{2}+o(1)\right)\log \frac{3n(n-1)}{8(n-2)}\frac{1}{\tau\log n}.
\ee
This implies that we should expect a reduction in the number of the terms in the Hamiltonian by a constant factor of about $2\tau$.

\subsection{Random $k$-local Hamiltonians}
\label{rk}

To choose a random interaction hypergraph of a $k$-local Hamiltonian we choose $m$ independent $k$-tuples of qubit labels between $1$ and $n$. We then uniformly randomly assign one of the $3^k$ Pauli operators of weight $k$ to that $k$-tuple. Let $S_1$ and $S_2$ be two sets of $k$ qubits. Given tuple $S_1$ there are
\begin{equation}
N_I={n-k\choose k-I}{k\choose I}    
\end{equation}
tuples $S_2$ with $I\leq|S_1\cup S_2|$, where $0\leq I\leq k$. Summing over $I$ recovers all $k$-tuples, by the Chu--Vandermonde identity. The probability of tuples $S_1$ and $S_2$ intersecting on $I$ qubits is therefore
\begin{equation}
p_I =  {n\choose k}^{-1}{n-k\choose k-I}{k\choose I}. 
\end{equation}

Given that the tuples $S_1$ and $S_2$ intersect on $I$ qubits, what is the probability that they commute? Let the Pauli factors of $S_1$ and $S_2$ be identical on a subset of their intersection of size $\sigma$ and otherwise every pair of tensor factors in the intersection disagree. The total number of pairs of Pauli operators on the intersection is $9^I$. The number of Pauli operators identical on $\sigma$ qubits is
\begin{equation}
t_{I,\sigma}=3^\sigma{I\choose \sigma}3^{I-\sigma}2^{I-\sigma},   
\end{equation}
which is obtained by multiplying the $3^\sigma$ Pauli operators common to the subset of $\sigma$ qubits by the number of subsets of size $\sigma$ and the number of distinct assignments to pairs of tensor factors in the complement of the subset of size $\sigma$. The total number of Pauli operators is then given by
\begin{equation}
9^I = \sum_{\sigma=0}^I\frac{6^I}{2^\sigma}{I\choose \sigma}.
\end{equation}
In order that a pair of operators commutes the size of the complement of the identical set must be even. That is,
\begin{equation}
p_c^{(I)} = \left(\frac{2}{3}\right)^I \sum_{I-\sigma~{\rm even}}\frac{1}{2^\sigma}{I\choose \sigma} =\frac{1}{2}\left(1+\frac{1}{3^I}\right).  
\end{equation}

The overall probability that a pair of tuples commutes is therefore
\begin{equation}\label{pcgen}
\begin{split}
p_c &=\sum_{I} p_Ip_c^{(I)}=
\sum_{I}\frac{p_I}{2}\left(1+\frac{1}{3^{I}}\right).
\end{split}
\end{equation}
For $k=2$ we recover \cref{2l}. For $k=3$ we obtain
\be
p_c=1-\frac{1}{n(n-1)(n-2)}\left(3n^2-13n-\frac{134}{3}\right).
\ee
Higher values of $k$ can be obtained from \cref{pcgen}. The expression in \cref{pcgen} justifies the use of coloring bounds for Erd\H{o}s--R\'{e}nyi random graphs for large numbers of qubits when the expression of \cref{pcgen} limits to \cref{lim}.

\section{Electronic Structure Hamiltonians}
\label{sec4}

Quantum chemistry simulations are expected to be an important use of variational quantum algorithms~\cite{olsonQuantumInformationComputation2017}. The goal is to find the eigenvalues and eigenvectors of the molecular electronic Hamiltonian
\begin{equation}\label{eq:H_e}
	H = \sum_{p,q} h_{pq} a_p^\dagger a_q + \frac{1}{2}\sum_{p,q,r,s} h_{pqrs} a_p^\dagger a_q^\dagger a_r a_s,
\end{equation}
where $ a_p^\dagger $ and $ a_p $ are fermionic creation and annihilation operators acting on the space spanned by molecular spin orbitals $ \chi_p $. For computational purposes, this basis set is truncated to the first $ N $ orbitals. The fermionic operators satisfy the canonical anticommutation relations
\begin{equation}\label{eq:fermion_CAR}
	\begin{split}
		\{ a_p^\dagger,a_q^\dagger \} &= \{ a_p,a_q \} = 0,\\
		\{ a_p,a_q^\dagger \} &= \delta_{pq}{\openone}.
	\end{split}
\end{equation}
The weights $h_{pq}$ and $h_{pqrs}$ are defined as
\begin{align}
	h_{pq} &= \delta_{\sigma_p\sigma_q} \int d^3r\, \chi_p^*(\mathbf{r}) \left( -\frac{\nabla^2}{2} - \sum_I \frac{\zeta_I}{|\mathbf{r}-\mathbf{R}_I|} \right) \chi_q(\mathbf{r}),\\
	h_{pqrs} &= \delta_{\sigma_p\sigma_s} \delta_{\sigma_q\sigma_r} \int d^3r_1 d^3r_2\, \frac{\chi_p^*(\mathbf{r}_1)\chi_q^*(\mathbf{r}_2)\chi_r(\mathbf{r}_2)\chi_s(\mathbf{r}_1)}{|\mathbf{r}_1-\mathbf{r}_2|},
\end{align}
where $ \mathbf{r} $ denotes the electronic spatial coordinates, $ \sigma_p \in \{ \uparrow,\downarrow \} $ is the spin value of the $ p $th orbital, and $ \{ \mathbf{R}_I \}_I $ and $ \{ \zeta_I \}_I $ are the molecule's classical nuclear positions and their associated charges, respectively. These spatial integrals can be efficiently pre-computed on a classical computer. For use in a quantum algorithm, the Hamiltonian is then transformed to a weighted sum of Pauli strings using a fermion-to-qubit encoding, such as the Jordan--Wigner~\cite{jordan28a}, Bravyi--Kitaev~\cite{bravyi02a,love12,tranter15a}, or other similar~\cite{setia17a} mappings. For the former two encodings, the number $n$ of qubits is the same as the number $N$ of molecular spin orbitals.  The expectation value of each Pauli string is measured independently.  The power of this approach stems from the ability to prepare ansatz states that cannot be efficiently constructed on a classical computer;~these are typically derived from a unitary coupled cluster ansatz~\cite{mccleanTheoryVariationalHybrid2016,romeroStrategiesQuantumComputing2018,leeGeneralizedUnitaryCoupled2019}. This allows for efficient computation of high-precision eigenvalues, which has importance when considering calculations that require such precision, such as reaction kinetics and dynamics.

Implementation of this procedure for chemical systems at the desired accuracy is challenging. For chemistry, the required precision is typically considered to be a constant $\SI{1}{\kilo\calorie\per\mol}$, or $1.6$~mHa. This level of precision is roughly commensurate with that obtained by experimental techniques in thermochemistry.  Recall from \cref{singletermmeasurements} that the number of independent measurements that must be performed to estimate the expectation value of a single term with weight $ h $ to precision $ \epsilon $ is $O(\Lambda|h|/\epsilon^2)$.  For chemical accuracy, this means that each term requires on the order of hundreds of thousands of independent measurements, each of which requires a separate ansatz preparation stage.  This must be repeated for each step of the variational optimisation, for each of the $O(N^4)$ terms in the molecular Hamiltonian (noting that using the Jordan--Wigner transformation requires up to 16 Pauli strings for each term).  As such, this quantum chemistry problem has recently garnered much interest with regard to reducing VQE measurement costs~\cite{wangAcceleratedVariationalQuantum2019,huggins2019efficient,babbush2018low,bonetNearlyOptimalMeas2019,gokhale2019MinimizingStatePrep,yen2019measuring,izmaylov19a}. The term reduction strategy discussed in \cref{termreduction} appears a promising way to reduce the overall resources required by utilising available coherent computational resources subsequent to ansatz preparation.

In the absence of restrictions on the length of circuits that can be performed coherently, the term reduction strategy reduces the number of expectation values that must be independently estimated, going from the number of Hamiltonian terms to the number of fully anticommuting sets of terms.  The main task is therefore to partition the Hamiltonian into such sets.  The effectiveness of this term reduction strategy can be quantified by examining the number of fully anticommuting sets for a given Hamiltonian with respect to both the number of orbitals and the total number terms in the unmodified Hamiltonian. In Section~\ref{subsec:linearReductionInTerms}, we show that it is always possible to reduce the number of terms from $ O(N^4) $ to at most $ O(N^3) $ for any electronic structure Hamiltonian. In Section~\ref{subsec:pauliLevelColouring}, we perform numerical studies using specific molecules and compare the results to our analytic construction. We also consider how the constraint of circuit size affects one's ability to construct such partitions. Finally, in \cref{dual} we examine such Hamiltonians in the plane-wave dual basis introduced in~\cite{babbush2018low} and observe a constant factor reduction of terms by unitary partitioning.
	
\subsection{Majorana operators}
The approach we take here will be agnostic to the choice of qubit encoding. However, in order to partition the terms into completely anticommuting sets, it will be convenient to express them using Majorana operators. This is because they place all the fermionic operators on an equal footing, are Hermitian and unitary, and obey a single anticommutation relation. Here, we briefly review the properties of these operators essential for our analysis. The single-mode Majorana operators are defined from the fermionic modes as
\begin{equation}
\begin{split}
	\gamma_{2p} &= a_p + a_p^\dagger,\\
	\gamma_{2p+1} &= -i(a_p - a_p^\dagger).
\end{split}
\end{equation}
In this formalism, the anticommutation relations of Eq.~\eqref{eq:fermion_CAR} become
\begin{equation}\label{eq:majorana_CAR}
	\{ \gamma_j, \gamma_k \} = 2\delta_{jk}{\openone}.
\end{equation}
These $ 2N $ single-mode operators generate a basis (up to phase factors) for the full algebra of Majorana operators via arbitrary products, i.e.,
\begin{equation}
	\gamma_A = \prod_{j \in A} \gamma_j,
\end{equation}
where $ A \subseteq \{ 0,\ldots,2N-1 \} $ is the support of $ \gamma_A $. From Eq.~\eqref{eq:majorana_CAR}, it is straightforward to show that the anticommutator between two arbitrary Majorana operators $ \gamma_A $ and $ \gamma_B $ is determined by their individual supports and their overlap:
\begin{equation}\label{eq:majorana_gen_AR}
	\{ \gamma_A, \gamma_B \} = \left[ 1 + (-1)^{|A||B| + |A \cap B|} \right] \gamma_A \gamma_B.
\end{equation}
This relation provides a clear picture of how to construct fully anticommuting sets of fermionic operators. Since the electronic Hamiltonian contains only terms of quadratic and quartic order, we restrict our attention to even-parity products. In this setting, we only need to examine the overlap of the Majorana operators' supports:~if $ |A \cap B| $ is odd (i.e., the two operators share an odd number of single-mode indices), then they anticommute.

\subsection{Linear reduction in terms}
\label{subsec:linearReductionInTerms}
Since there are no spin interaction terms in our Hamiltonian, we can always choose molecular orbital basis functions $ \chi_p $ which are real-valued. With this, it follows that $ h_{pq}, h_{pqrs} \in \mathbb{R} $, and in particular, we have the permutational symmetries
\begin{align}
	h_{pq} &= h_{qp}, \label{eq:hpq_sym}\\
	h_{pqrs} = h_{sqrp} &= h_{prqs} = h_{srqp}. \label{eq:hpqrs_sym_1}
\end{align}
Furthermore, the canonical anticommutation relations give $ a_p^\dagger a_q^\dagger a_r a_s = a_q^\dagger a_p^\dagger a_s a_r $, which implies that
\begin{equation}\label{eq:hpqrs_sym_2}
	h_{pqrs} = h_{qpsr},
\end{equation}
for a total of eight permutational symmetries in the two-body integrals. Using these symmetries and the generalized anticommutation relation, Eq.~\eqref{eq:majorana_gen_AR}, one can rewrite the Hamiltonian using Majorana operators as
\begin{equation}\label{eq:H_e_majorana}
\begin{split}
	H &= \tilde{h} {\openone} + \sum_{p,q} \tilde{h}_{pq} i\gamma_{2p}\gamma_{2q+1} \\
	&\quad + \frac{1}{2} \sum_{\substack{p,q,r,s\\p\neq q,r\neq s}} \tilde{h}_{pqrs} \gamma_{2p}\gamma_{2q}\gamma_{2r+1}\gamma_{2s+1}.
\end{split}
\end{equation}
We refer the reader to Appendix~\ref{sec:H_e_majorana_appendix} for the details of this derivation. The redefined weights $ \tilde{h} $, $ \tilde{h}_{pq} $, and $ \tilde{h}_{pqrs} $ are given in Eq.~\eqref{eq:new_coeff}. For our present analysis, the only relevant detail here is that each term features an equal number of even and odd indices in its support. In principle, any such combination of terms may appear in the Hamiltonian. In this form, it becomes clear that there are up to $ N^2 $ quadratic terms and $ \binom{N}{2}^2 $ quartic terms.

Furthermore, since the single-mode Majorana operators are Hermitian, there is a one-to-one correspondence between Majorana operators and the respective Pauli strings obtained after a fermion-to-qubit transformation (for encodings that preserve the number of orbitals as the number of qubits). For instance, in the Jordan--Wigner encoding, we have
\begin{equation}
\begin{split}
    \gamma_{2p} = X_p Z_{p-1} \cdots Z_0, \\
    \gamma_{2p+1} = Y_p Z_{p-1} \cdots Z_0. \\
\end{split}
\end{equation}
%which means that, say, a single quartic term becomes a single Pauli string (suppose $ p > q > r > s $ for this particular example):
%\begin{equation}
%    \gamma_{2p}\gamma_{2q}\gamma_{2r+1}\gamma_{2s+1} = X_p Z_{p-1} \cdots Z_{q+1} Y_q Y_r Z_{r-1} \cdots Z_{s+1} X_s.
%\end{equation}
Since the single-mode Majorana operators simply become Pauli strings, arbitrary products of them remain single Pauli strings. In contrast, if one were to deal with the fermionic operators directly, a single $ a_p^\dagger a_q^\dagger a_r a_s $ term would generate a linear combination of up to 16 unique Pauli strings. By writing the Hamiltonian in terms of Majorana operators, we have not circumvented this overhead, but rather, we have explicitly incorporated it into our term counting, while remaining encoding agnostic. In particular, many cancellations and simplifications may occur between the transformed terms, yielding the expression given above in Eq.~\eqref{eq:H_e_majorana}. Also note that any anticommuting partition in the Majorana formalism remains valid after a qubit transformation, since the anticommutation relations are preserved.

Recall from Eq.~\eqref{eq:majorana_gen_AR} that we had determined that every pair of terms anticommutes if and only if their supports intersect an odd number of times. This fact, along with the specific form of the terms appearing in Eq.~\eqref{eq:H_e_majorana}, is crucial for showing that it is always possible to partition this Hamiltonian into at most $ O(N^3) $ completely anticommuting sets.

We note that very recent results have made similar findings. In~\cite{bonetNearlyOptimalMeas2019}, it was observed that at least $ \Omega(N^3) $ sets would be necessary to divide the set of \emph{all} quartic Majorana operators, rather than the specific terms appearing in electronic structure Hamiltonians. Meanwhile, in~\cite{gokhale2019n}, an algorithm was presented which partitions electronic structure terms into $ O(N^3) $ completely \emph{commuting} sets. The analysis presented there specifies the Jordan--Wigner encoding, but does not assume any of the permutational symmetries in the $ h_{pq}, h_{pqrs} $ coefficients.

We now prove our claim by providing an explicit construction of such a partition.

\begin{theorem}\label{thm:cubicPartitionTheorem}
Let
\begin{equation}
	\mathcal{M} = \{ \gamma_{2p}\gamma_{2q}\gamma_{2r+1}\gamma_{2s+1} \mid p < q \;{\rm and }\; r < s \}
\end{equation}
be the set of all possible quartic Majorana operators appearing in the electronic structure Hamiltonian. For each triple $ (q,r,s) \in \{ 0,\ldots,N-1 \}^3 $ satisfying $ r < s $, define
\begin{equation}
	S_{(q,r,s)} = \{ \gamma_{2p}\gamma_{2q}\gamma_{2r+1}\gamma_{2s+1} \mid p < q \}.
\end{equation}
These sets $ S_{(q,r,s)} $ are completely anticommuting, and they form a partition of $ \mathcal{M} $. Furthermore, there are $ O(N^3) $ such sets.
\end{theorem}

\begin{proof}
By construction, all elements of $ S_{(q,r,s)} $ share support on exactly three indices, hence they all pairwise anticommute, per Eq.~\eqref{eq:majorana_gen_AR}. It is also straightforward to see that these sets form an exact cover of $ \mathcal{M} $:
\begin{align}
	\left| S_{(q,r,s)} \cap S_{(q',r',s')} \right| &= q \, \delta_{qq'} \delta_{rr'} \delta_{ss'}, \label{eq:disjointness}\\
	\bigcup_{\substack{q,r,s\\r<s}} S_{(q,r,s)} &= \mathcal{M}. \label{eq:covering}
\end{align}
There are $ \binom{N}{2} $ values that the pair $ (r,s) $ can take and $ N-1 $ values that $ q $ can take ($ q=0 $ yields the empty set, which we ignore). A slight optimization arises from the observation that the union $ S_{(1,r,s)} \cup S_{(2,r,s)} $ remains a completely anticommuting set. Hence there are a total of $ \binom{N}{2}(N-2) = O(N^3) $ such sets.
\end{proof}

We refer the reader to Appendix~\ref{sec:majorana_proof_details} for further details of the above proof. Although there are only $ O(N^2) $ quadratic terms, hence not affecting the asymptotic scaling of Theorem~\ref{thm:cubicPartitionTheorem}, they can in fact be included in the above construction with no additional overhead. Intuitively, since there are at most $ N^2 $ such operators which need to be placed into $ O(N^3) $ sets, one has a great deal of freedom in how to allocate them. As one example, consider the set
\begin{equation}
	T_p = \{ i\gamma_{2p}\gamma_{2q+1} \mid 0 \leq q \leq N-1 \}
\end{equation}
for some fixed $ p $. Then all the elements of $ T_p $ anticommute with all of some $ S_{(p,r,s)} $, except for those with $ q=r $ or $ q=s $. The new completely anticommuting set then becomes
\begin{equation}\label{eq:Sjkm_quad_1}
	S_{(p,r,s)} \cup T_p \setminus \{ i\gamma_{2p}\gamma_{2r+1}, i\gamma_{2p}\gamma_{2s+1} \},
\end{equation}
and those two excluded operators can be placed with any other $ S_{(p,r',s')} $, where all of $ r,r',s $, and $ s' $ are different:
\begin{equation}\label{eq:Sjkm_quad_2}
	S_{(p,r',s')} \cup \{ i\gamma_{2p}\gamma_{2r+1}, i\gamma_{2p}\gamma_{2s+1} \}.
\end{equation}
Since there are $ N $ such sets $ T_p $, this procedure combines all possible $ N^2 $ quadratic operators with only $ 2N $ of the preexisting sets of quartic operators.

We emphasize that the partition presented here is not an optimal solution to the problem. Rather, it demonstrates that even in the worst case one can always achieve term reduction by at least a factor of $ O(N) $. For a practical demonstration, we now move to numerical studies of specific molecular Hamiltonians.

\subsection{Pauli-level colouring and numerics}\label{subsec:pauliLevelColouring}

The above analysis demonstrates a reduction in difficulty of VQE by considering the number of fully anticommuting sets of terms in the electronic Hamiltonian.  Equivalently, we may consider fully anticommuting sets of terms at the level of Pauli strings, i.e.~subsequent to transforming the electronic Hamiltonian with, for example, the Jordan--Wigner or Bravyi--Kitaev mappings.  This approach could hold advantage by allowing the combination of duplicate strings and allowing the combination of anticommuting Pauli subterms between different fermionic terms.  However, once the fermion-to-qubit mapping is applied, the natural symmetries of the spatial molecular orbital integrals are embedded into a complex structure.  Moreover, the anticommutativity structure of the resulting Pauli terms is difficult to predict.  As such, we turn to numerical methods.

The key metric here is the number of fully anticommuting sets in the Pauli Hamiltonian.  As discussed in \cref{unitarypartitioning}, this is equivalent to a colouring of the compatibility graph---the graph composed of nodes corresponding to terms, with edges drawn where terms commute.  Optimal graph colouring is an NP-hard problem~\cite{gareySimplifiedNPcompleteProblems1974a},  but many approximate algorithms exist~\cite{kosowskiClassicalColoringGraphs2004}.  While minimising the number of sets is advantageous for reducing the number of measurements needed, an approximate solution is sufficient, and diminishing returns are obtained from improving the quality of the approximation.

In order to assess whether this strategy is viable for molecular Hamiltonians, we generated colouring schemes for $65$ Hamiltonians (previously used in Refs.~\cite{tranterComparisonBravyiKitaev2018,tranterOrderingTrotterizationImpact2019} and described in Appendix~\ref{appendix:systems}).  Geometry specifications were obtained from the NIST CCBDB database~\cite{johnsoniiiNISTComputationalChemistry2016}.  Molecular orbital integrals in the Hartree--Fock basis were gathered using the Psi4 package~\cite{parrishPsi4OpenSourceElectronic2017} and OpenFermion~\cite{mccleanOpenFermionElectronicStructure2017}.  Our code was then used to generate Jordan--Wigner and Bravyi--Kitaev Hamiltonians, which were divided into anticommuting subsets using the NetworkX Python package~\cite{SciPyProceedings_11} and the greedy independent sets strategy~\cite{kosowskiClassicalColoringGraphs2004}.  As our focus was on quantifying whether the term reduction technique is viable, alternative colouring strategies were not considered;~such an analysis was performed in~\cite{izmaylov19a}.  Our colouring strategy here is relatively computationally expensive, limiting our analysis to a maximum of 36 spin orbitals, with only three systems involving 30 or more.  While our code is unoptimised and can likely be improved upon, this does indicate that it would be difficult to extend this approach to larger systems.  The Majorana-based scheme of Section~\ref{subsec:linearReductionInTerms} was also used to partition the Hamiltonians.  In contrast to the greedy colouring strategy, this  does not require extensive classical computational resources.

\begin{figure}
    \centering
    \includegraphics[width=\columnwidth]{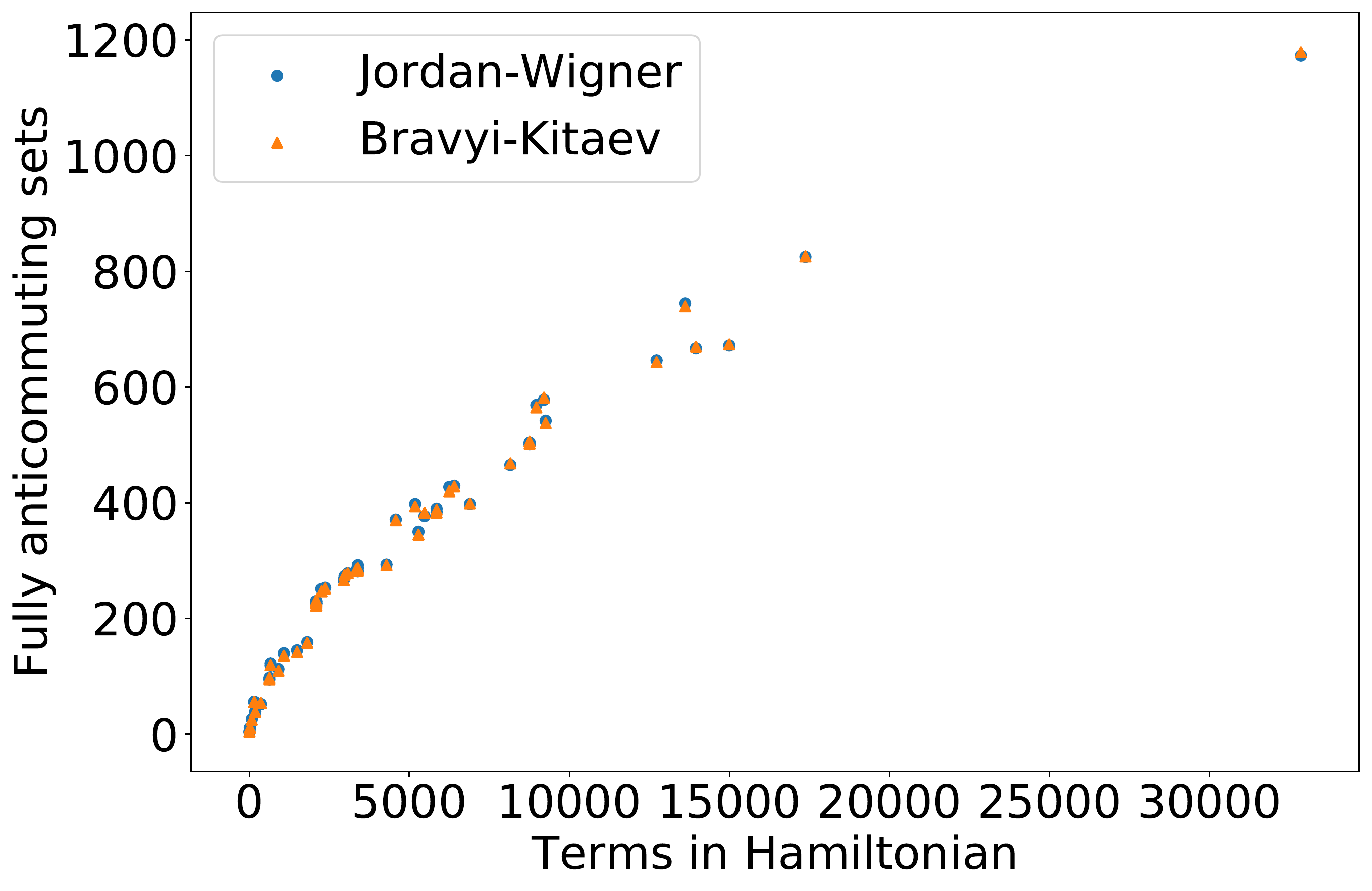}
    \caption{Number of fully anticommuting sets for electronic structure Hamiltonians versus the number of terms in the full Pauli Hamiltonian, using the greedy independent sets strategy.  The number of fully anticommuting sets is at least an order of magnitude less than the number of terms. The Jordan--Wigner and Bravyi--Kitaev mappings perform almost equivalently.}
    \label{fig:numSetsVersusTerms}
\end{figure}

Figure~\ref{fig:numSetsVersusTerms} shows the number of fully anticommuting sets obtained versus the number of terms in the Hamiltonian.  The number of fully anticommuting sets is approximately an order of magnitude less than the number of terms. The choice of Jordan--Wigner and Bravyi--Kitaev mapping does not appear to meaningfully affect the number of fully anticommuting sets found, as the anticommutativity structure is dependent on the underlying molecular Hamiltonian.  Encouragingly, the agreement demonstrated here by Figure~\ref{fig:numSetsVersusTerms} suggests that the greedy independent set strategy is finding close-to-optimal colourings.

The results for both partitioning schemes against the number of spin orbitals are depicted in Figure~\ref{fig:numSets}.  Both the numerical implementation of the Majorana-based construction and the greedy colouring scheme prove to be consistently effective.  Beyond the smallest Hamiltonians, a roughly linear trend between the number of sets found and the number of Hamiltonian terms is observed, demonstrating that the asymptotic improvement discussed in \cref{subsec:linearReductionInTerms} can be achieved when using numerical approaches to colouring Pauli Hamiltonians.  The numerical Majorana results, and the greedy colouring strategy, consistently outperform the analytic upper bound, as expected.  This may be attributed primarily to the sparsity in the $h_{pq}$ and $h_{pqrs}$ weights, due to geometric molecular symmetries and the locality of the basis functions.  The ratio of the number of terms to the number of sets also appears to increase linearly with the number of spin orbitals (albeit with high variance), in agreement with the scaling properties discussed in \cref{subsec:linearReductionInTerms}.

\begin{figure*}
    \centering
    \includegraphics[width=\textwidth]{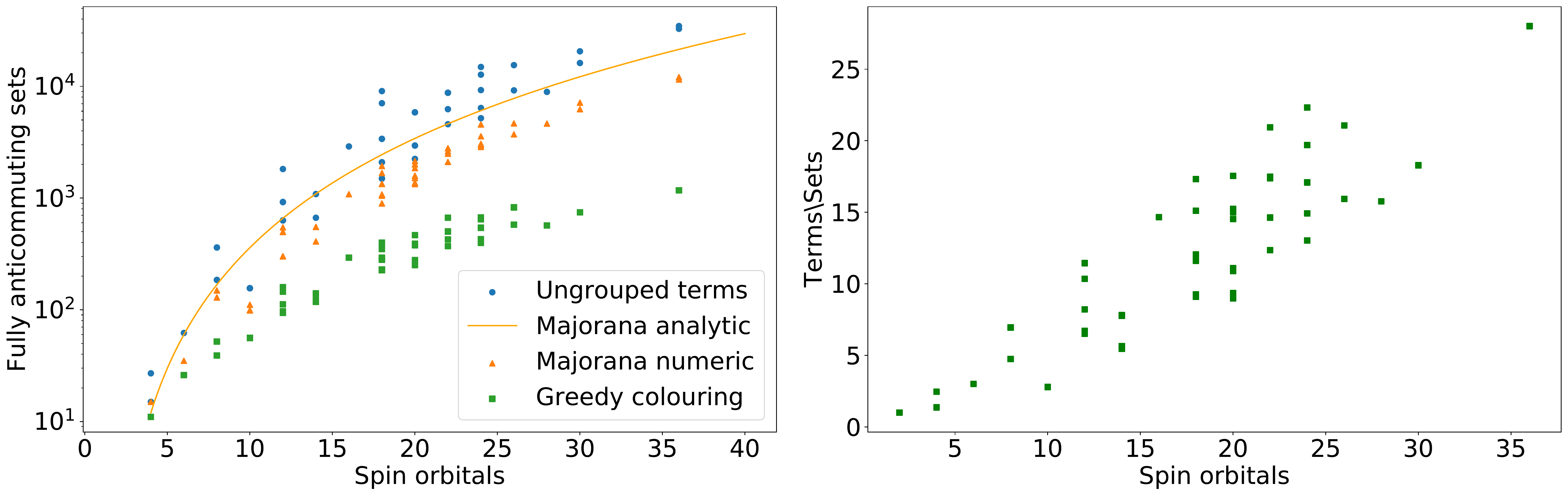}
    \caption{Number of fully anticommuting sets for electronic structure Hamiltonians versus the number of spin orbitals, using the Jordan--Wigner mapping.   Left:~Including all partitioning schemes. The ``Majorana analytic'' curve is the $ \binom{N}{2}(N-2) $ upper bound obtained from \cref{thm:cubicPartitionTheorem} for generic Hamiltonians of Eq.~\eqref{eq:H_e_majorana}. The ``Majorana numeric'' data points correspond to the partitions described in \cref{subsec:linearReductionInTerms} without further optimisation. This upper bound is loose, due to sparsity in the molecular Hamiltonians versus the set of all possible terms.  Right:~Ratio of the number of terms to the number of anticommuting sets, for systems with more than 5 spin orbitals.  A roughly linear trend is observed, in agreement with the analytic scaling discussed in \cref{subsec:linearReductionInTerms}.}
    \label{fig:numSets}
\end{figure*}

The greedy colouring scheme yields roughly a factor of $10$ improvement over the numerical Majorana scheme, suggesting that it may be of substantial use in NISQ VQE experiments.  However, it should be emphasised that the substantial classical computing resources required may inhibit its use for systems with more spin orbitals.  The Majorana-based scheme demonstrates the same term reduction scaling, but with substantially reduced classical overhead.

Although these results are promising, they do not consider the difficulty of performing the additional coherent operations required for the term recombination procedure. In principle, our analytic construction of anticommuting sets in Section~\ref{subsec:linearReductionInTerms} requires only $ O(N) $ depth circuits under the Jordan--Wigner mapping. This can be shown using well-known gate-compiling techniques~\cite{whitfieldSimulation2011,hastingsImprovingQuantumAlgorithms2015}.  Figure~\ref{fig:circuits} shows that the length of the circuits grows slowly in comparison to the amount of terms in the Hamiltonian.  However, near-term quantum devices are likely to be heavily constrained in the number of operations that can be performed coherently.  As such, it is likely that it will not be possible to combine entire sets of anticommuting terms.  Crucially, however, the term recombination procedure can be applied to \emph{subsets} of the fully anticommuting sets.  Provided the available coherent resources can be quantified prior to execution of the circuits, subsets of terms can be found to maximally use such resources to reduce the overall number of measurements required.  This yields a hardware-dependent tunable parameter---for example, the number of gates that can be implemented coherently subsequent to ansatz preparation---introduced at compile time.  This parameter allows for optimal use of the quantum resources provided by a given hardware option.

In order to assess the implications of varying such a parameter, we generated circuits corresponding to the implementation of the term reduction procedure for each Hamiltonian, introducing a maximum post-ansatz preparation gate count parameter.  For simplicity, these circuits used the standard method of implementing exponentiated Pauli strings given in \cref{unitarypartitioning}, rather than the ALCU circuits of \cref{rotations}.  Where circuits exceeded this length, the corresponding anticommuting set was split in half and new circuits were generated.  This binary splitting process was iterated until sufficiently short circuits were found.  Adjacent self-inverse gates were cancelled, moving through commuting gates where necessary~\cite{hastingsImprovingQuantumAlgorithms2015}.  For verification purposes, we calculated the expectation values with the true ground state of the Hamiltonians predicted by the circuits for systems with less than ten qubits.  As the results presented in Figure~\ref{fig:numSets} suggest that there is little difference between Jordan--Wigner and Bravyi--Kitaev circuits, we consider only Jordan--Wigner circuits.

\begin{figure*}
    \centering
    \includegraphics[width=\textwidth]{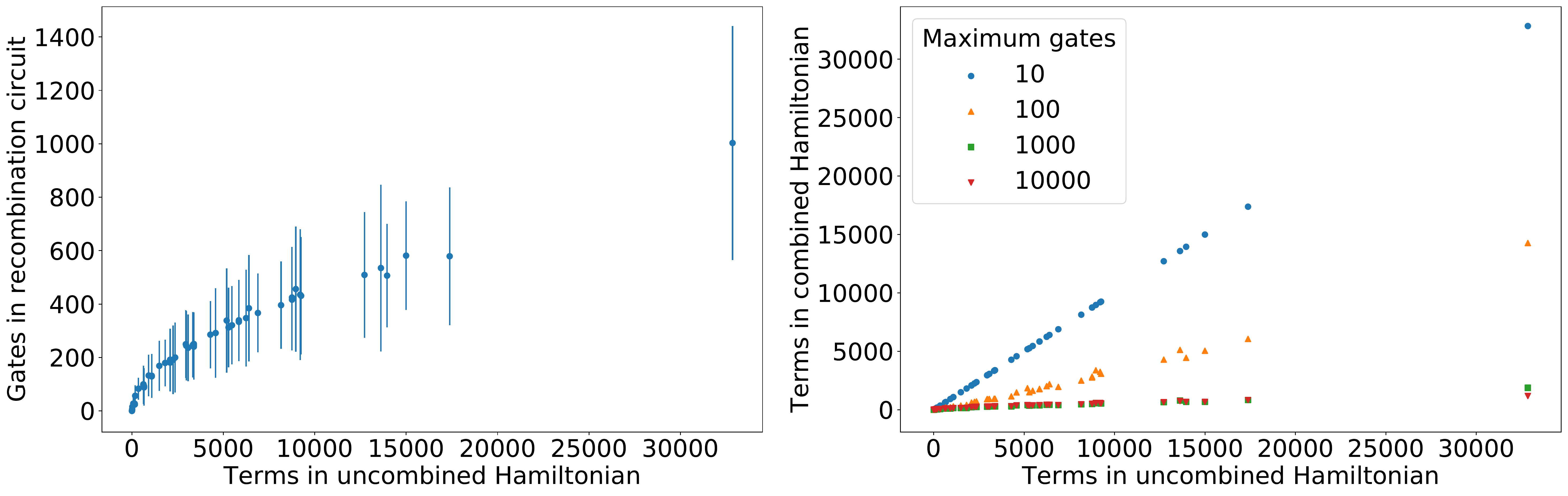}
    \caption{Resource requirements for full and partial term reduction using the greedy algorithm for partitioning. Left:~Average post-ansatz gates required for full term reduction.  Whiskers denote one standard deviation in the length of the circuit required for each anticommuting partition in the Hamiltonian.  The growth in circuit length is dramatically slower than the growth in the number of Hamiltonian terms, but displays high variance between anticommuting sets. Right:~Reduction in the number of required independent expectation values, given restrictions on maximum individual circuit length.  With highly restricted circuit lengths, term recombination is largely impossible.  However, roughly 1000 additional gates at most are sufficient to perform near-maximal term reduction for the molecules considered here (up to 36 spin orbitals), which is in agreement with the figure to the left.}
    \label{fig:circuits}
\end{figure*}

Figure~\ref{fig:circuits} shows the results of this process.  Using a maximum circuit length of $10\,000$ gates subsequent to ansatz preparation allows all anticommuting sets, in all Hamiltonians, to be combined.  Allowing only $10$ gates removes any possibility of term recombination.  Encouragingly, allowing $100$ gates does not dramatically impede term recombination.  Even for the longest circuit considered, using $100$ gates allows for a reduction in terms by a factor of over 2.  Allowing $1000$ postansatz gates similarly performs as well as full anticommuting set recombination in all systems apart from the bromine atom;~in this instance, the difference between the $1000$- and $10\,000$- gate decompositions is minor. 

Our choice of allowable circuit length here is intended to be illustrative of the practicality of the term recombination procedure.  In a true simulation, the maximum post-ansatz gates parameter should be set to a value that is empirically determined by the ability of the hardware and should not be restricted to an integer power of 10.  Given the relatively low gate counts required for substantial improvement with regard to the number of terms, the results here strongly suggest that this approach is an effective way of reducing the overall runtime of variational quantum algorithms for electronic structure.

\subsection{The plane-wave dual basis}\label{dual}

The use of a plane-wave basis is well established for condensed-matter systems. The plane-wave and plane-wave dual basis was recently used in the context of quantum simulation of quantum chemistry to express the Hamiltonian with a number of terms scaling quadratically with the number of basis functions~\cite{babbush2018low}. While suitable for periodic systems, the plane-wave dual basis requires a constant factor of additional spin orbitals to achieve the same accuracy as Gaussian-type orbitals for nonperiodic systems such as molecules. Thus the choice of basis set depends highly on the system under consideration, especially for near-term applications.

The qubit Hamiltonian obtained from the Jordan--Wigner transformation is (see Eq.~(9) in~\cite{babbush2018low}):
\begin{equation}
\begin{split}
\label{jwhampwb}
H & =  \sum_{\substack{p, \sigma \\ \nu \neq 0}} \! \left( \! \frac{\pi}{\Omega \, k_\nu^2} \! - \! \frac{k_\nu^2}{4 \, N} \! + \! \frac{2\pi}{\Omega} \! \sum_{j}\zeta_j \frac{\cos\left[k_\nu \! \cdot \! \left(R_j \!- \! r_p\right)\right]}{k_\nu^2} \! \right) \! Z_{p,\sigma}\\
&+ \frac{\pi}{2\,\Omega } \sum_{\substack{(p, \sigma) \neq (q, \sigma') \\ \nu \neq 0}} \frac{\cos \left[k_\nu \cdot r_{p-q}\right]}{k_\nu^2} Z_{p,\sigma} Z_{q,\sigma'}\\
& + \frac{1}{4\, N} \sum_{\substack{p \neq q \\ \nu, \sigma}} k_\nu^2 \cos \left[k_\nu \cdot r_{q - p} \right] X_{p,\sigma} Z_{p + 1,\sigma} \cdots Z_{q - 1,\sigma} X_{q,\sigma}\\
& + \frac{1}{4\, N} \sum_{\substack{p \neq q \\ \nu, \sigma}} k_\nu^2 \cos \left[k_\nu \cdot r_{q - p} \right] Y_{p,\sigma} Z_{p + 1,\sigma} \cdots Z_{q - 1,\sigma} Y_{q,\sigma} \\
&+ \sum_{\nu \neq 0} \left(\frac{k_\nu^2}{2}- \frac{\pi \, N}{\Omega \, k_\nu^2} \right) {\openone}.
\end{split}
\end{equation}
The labels $p$ run over $N$ basis functions, and so by inspection we see that the number of terms in the Hamiltonian is $O(N^2)$. Also by inspection, we can identify a set of $N^2$ commuting operators $Z_{p,\sigma}Z_{q,\sigma'}$. Thus, we can immediately conclude that unitary partitioning cannot reduce the asymptotic number of terms in this Hamiltonian.

However, we may use unitary partitioning to reduce the number of terms by a constant factor.
We can identify sets of anticommuting terms from \cref{jwhampwb} as follows. Define the sets
\be
\begin{split}
A_p &=\{Z_p\}\cup\{X_{p-1}X_p\}\cup\{Y_{p}Y_{p+1}\}\\
&\cup\{Y_pZ_{[p+1,p+l+1]}Y_{p+l+2} \mid 0\leq l\leq N-p-3\}\\
&\cup\{X_lZ_{[l+1,p-1]}X_p \mid 0\leq l\leq p-2\}.
\end{split}
\ee
There are $N$ operators in each set $A_p$, all of which pairwise anticommute. All sets $A_p$ are distinct and so unitary partitioning can reduce each set $A_p$ to a single term. This results in a fractional reduction in the number of terms of $(2N+1)/(4N-1)$, giving an asymptotic reduction in the number of terms by a factor of 2. 

\section{Noncontextual Hamiltonians}
\label{noncontextual}

In Ref.~\cite{kirby19a}, \emph{contextuality} of a Pauli Hamiltonian is defined as the condition under which it is impossible to consistently assign values to the Pauli terms in the Hamiltonian.
Contextuality, if present, is a manifestation of nonclassicality of the Hamiltonian.
Contextuality of a Hamiltonian is determined by the following criterion on the set $\mathcal{S}$ of Pauli terms~\cite{kirby19a}:~first, let $\mathcal{Z}\subseteq\mathcal{S}$ be the set of terms that commute with all other terms, and let $\mathcal{T}\equiv\mathcal{S}\setminus\mathcal{Z}$.
Then $\mathcal{S}$ is noncontextual if and only if commutation is an equivalence relation on $\mathcal{T}$.
In other words, if and only if $\mathcal{S}$ is noncontextual, $\mathcal{T}$ partitions into a union of disjoint cliques $C_1,C_2,\ldots,C_N$ such that operators in different cliques anticommute, while operators in the same clique commute (so in the graph-theoretic sense these are cliques in the compatibility graph).

We now show that, using the term reduction technique presented above, we can map any noncontextual Hamiltonian to a commuting Hamiltonian.
First, as shown in~\cite{kirby19a}, we can check that the Hamiltonian is noncontextual in $O(|\mathcal{S}|^3)$ time.
Given that the Hamiltonian is noncontextual, we know that it has the structure described above:~we can find the cliques $C_i$ as well as $\mathcal{Z}$ in $O(|\mathcal{S}|^2)$ time.

To map these terms to a commuting set, find a largest clique, and without loss of generality let it be $C_1$.
Then construct a set $D_1$ by selecting exactly one element from each of the $C_i$ ($D_1$ is a minimal hitting set on the $C_i$).
Similarly, construct $D_2$ by selecting exactly one element from each of the $C_i$ after removing the elements in $D_1$, and so forth, until we have covered $\mathcal{T}$ with disjoint sets $D_1,D_2,\ldots,D_M$, where $M=|C_1|$.
Letting $C_{ij}$ denote the $j$th element of $C_i$, we may visualize the $D_j$ as
\be
\begin{array}{c|cccc}
    ~&D_1&D_2&\cdots&D_M\\
    \hline
    C_1&C_{11}&C_{12}&\cdots&C_{1M}\\
    C_2&C_{21}&C_{22}&\cdots&\cdots\\
    \vdots&\vdots&\vdots&\ddots&\vdots\\
    C_N&C_{N1}&\cdots&\cdots&\cdots
\end{array}\quad.
\ee
Since $C_1$ is a largest clique, it is guaranteed to have nonempty intersection with all of the sets $D_j$.
Elements of different cliques anticommute, so each of the $D_j$ is completely anticommuting.
Therefore, we can use the techniques described in \cref{unitarypartitioning,rotations} to construct Pauli rotations $ R_{D_j}$ that map the operators in each $D_j$ to the operator $C_{1j}$, the single operator in $D_j\cap C_1$.

The resulting Hamiltonian has terms $\mathcal{Z}\cup C_1$, which commute, since by definition the operators in $\mathcal{Z}$ commute with all operators in $\mathcal{S}$, and the operators in $C_1$ also commute with each other.
Thus any noncontextual Hamiltonian may be mapped to a set of commuting terms that form an effective commuting Hamiltonian, using as a resource only the ability to append the additional Pauli rotations $ R_{D_j}$ to the state preparations as in \cref{rotatedhamiltonian}.
It is important to note, however, that the commuting Hamiltonian is not unitarily equivalent to the noncontextual Hamiltonian, since the rotations required to map each set $D_j$ to a single operator vary with $j$.

\section{Conclusions}

In this paper we have discussed unitary partitioning---the technique for using anticommuting sets of Hamiltonian terms to reduce the number of measurements needed when performing variational quantum algorithms.  Applying this technique to transverse Ising models and random Hamiltonians resulted in a constant factor improvement in the number of independent expectation value estimations required.  However, applying the technique to electronic structure Hamiltonians yielded greater reduction, scaling linearly with the number of qubits.

The dramatic growth in the number of independent expectation values that must be determined is a key problem in the use of variational quantum algorithms for quantum chemistry in the NISQ era. Due to the nature of the plane-wave dual basis representation, in which one defines a basis that yields only $O(N^2)$ nonzero electronic Hamiltonian weights, we observed only a constant factor reduction in terms with unitary partitioning. However, using generic molecular orbital basis sets, we were able to obtain a reduction that scales linearly. We proved this result in Section~\ref{subsec:linearReductionInTerms} and confirmed its practicality by numerics in Section~\ref{subsec:pauliLevelColouring}.

We report two strategies for partitioning the electronic structure Hamiltonian into fully anticommuting subsets.  The first of these, based on expressing the fermionic Hamiltonian using Majorana operators, demonstrates the favourable scaling properties, and can be rapidly performed for even large numbers of spin orbitals.  Conversely, using a greedy colouring scheme is relatively expensive with regard to classical computational resources, but demonstrates an order-of-magnitude reduction, even for relatively small systems (less than $30$ qubits).
% This directly corresponds to a similar reduction in overall computation time.
The latter scheme is likely to be useful in NISQ applications where systems are small and greedy solutions can be feasibly computed.  The former yields the same scaling, and is not restricted by the cost of colouring algorithms, but suffers from a constant factor overhead in the number of fully anticommuting sets, compared to the greedy colouring method. The availability of postansatz coherent resources, and the relative difficulty of the classical partitioning step, may determine which scheme is favoured.

Finally, in \cref{noncontextual} we studied the class of noncontextual Hamiltonians, as defined in~\cite{kirby19a}. The presence of contextuality in a quantum system provides a barrier to a classical description of the system. Here, we have shown that any noncontextual Hamiltonian (which lacks this separation from a classical Hamiltonian) may be transformed into a Hamiltonian of fully commuting terms, using only the rotations developed in \cref{unitarypartitioning}. This result helps us further understand the connection between noncontextual Hamiltonians and commuting Hamiltonians, and it adds support to the notion that VQE experiments should focus on contextual Hamiltonians~\cite{kirby19a}.

Our analysis of circuits for implementing the unitary partitioning procedure indicates that relatively modest additional coherent resources are required, compared to those typically needed for ansatz preparation.  Crucially, this optimisation is tunable, allowing for optimal use of coherent resources by hardware-dependent parameterisation at compile time.  It is also likely that unitary partitioning is compatible with other aspects of VQE optimisation. For instance, while we have remained agnostic to the choice of the parametrised ansatz for this study, the form of the unitaries required to perform term reduction matches those of popular ansatz choices, such as the unitary coupled cluster and related methods~\cite{romeroStrategiesQuantumComputing2018,ryabinkin2018qubit,leeGeneralizedUnitaryCoupled2019}. Thus with proper circuit compilation, one may significantly reduce the effective number of postansatz operations in practice, instead incorporating their rotation angles into the appropriate ansatz parameters. For these reasons, we believe that unitary partitioning could substantially aid in the use of variational quantum algorithms for studying classically intractable systems.

\section*{Acknowledgements}
The authors would like to thank Alexis Ralli for productive discussions.
This work was supported by the National Science Foundation STAQ project (Grant No.~PHY-1818914).
W.M.K.~additionally acknowledges support from the National Science Foundation (Grant No.~DGE-1842474).

\bibliography{references}

\begin{thebibliography}{67}%
\makeatletter
\providecommand \@ifxundefined [1]{%
 \@ifx{#1\undefined}
}%
\providecommand \@ifnum [1]{%
 \ifnum #1\expandafter \@firstoftwo
 \else \expandafter \@secondoftwo
 \fi
}%
\providecommand \@ifx [1]{%
 \ifx #1\expandafter \@firstoftwo
 \else \expandafter \@secondoftwo
 \fi
}%
\providecommand \natexlab [1]{#1}%
\providecommand \enquote  [1]{``#1''}%
\providecommand \bibnamefont  [1]{#1}%
\providecommand \bibfnamefont [1]{#1}%
\providecommand \citenamefont [1]{#1}%
\providecommand \href@noop [0]{\@secondoftwo}%
\providecommand \href [0]{\begingroup \@sanitize@url \@href}%
\providecommand \@href[1]{\@@startlink{#1}\@@href}%
\providecommand \@@href[1]{\endgroup#1\@@endlink}%
\providecommand \@sanitize@url [0]{\catcode `\\12\catcode `\$12\catcode
  `\&12\catcode `\#12\catcode `\^12\catcode `\_12\catcode `\%12\relax}%
\providecommand \@@startlink[1]{}%
\providecommand \@@endlink[0]{}%
\providecommand \url  [0]{\begingroup\@sanitize@url \@url }%
\providecommand \@url [1]{\endgroup\@href {#1}{\urlprefix }}%
\providecommand \urlprefix  [0]{URL }%
\providecommand \Eprint [0]{\href }%
\providecommand \doibase [0]{http://dx.doi.org/}%
\providecommand \selectlanguage [0]{\@gobble}%
\providecommand \bibinfo  [0]{\@secondoftwo}%
\providecommand \bibfield  [0]{\@secondoftwo}%
\providecommand \translation [1]{[#1]}%
\providecommand \BibitemOpen [0]{}%
\providecommand \bibitemStop [0]{}%
\providecommand \bibitemNoStop [0]{.\EOS\space}%
\providecommand \EOS [0]{\spacefactor3000\relax}%
\providecommand \BibitemShut  [1]{\csname bibitem#1\endcsname}%
\let\auto@bib@innerbib\@empty
%</preamble>
\bibitem [{\citenamefont {Feynman}(1982)}]{Feynman1982}%
  \BibitemOpen
  \bibfield  {author} {\bibinfo {author} {\bibfnamefont {R.~P.}\ \bibnamefont
  {Feynman}},\ }\href {\doibase 10.1007/BF02650179} {\bibfield  {journal}
  {\bibinfo  {journal} {International Journal of Theoretical Physics}\ }\textbf
  {\bibinfo {volume} {21}},\ \bibinfo {pages} {467} (\bibinfo {year}
  {1982})}\BibitemShut {NoStop}%
\bibitem [{\citenamefont {Lloyd}(1996)}]{lloyd1996universal}%
  \BibitemOpen
  \bibfield  {author} {\bibinfo {author} {\bibfnamefont {S.}~\bibnamefont
  {Lloyd}},\ }\href@noop {} {\bibfield  {journal} {\bibinfo  {journal}
  {Science}\ }\textbf {\bibinfo {volume} {273}},\ \bibinfo {pages} {1073}
  (\bibinfo {year} {1996})}\BibitemShut {NoStop}%
\bibitem [{\citenamefont {Abrams}\ and\ \citenamefont
  {Lloyd}(1997)}]{abrams1997simulation}%
  \BibitemOpen
  \bibfield  {author} {\bibinfo {author} {\bibfnamefont {D.~S.}\ \bibnamefont
  {Abrams}}\ and\ \bibinfo {author} {\bibfnamefont {S.}~\bibnamefont {Lloyd}},\
  }\href@noop {} {\bibfield  {journal} {\bibinfo  {journal} {Physical Review
  Letters}\ }\textbf {\bibinfo {volume} {79}},\ \bibinfo {pages} {2586}
  (\bibinfo {year} {1997})}\BibitemShut {NoStop}%
\bibitem [{\citenamefont {Abrams}\ and\ \citenamefont
  {Lloyd}(1999)}]{abrams1999quantum}%
  \BibitemOpen
  \bibfield  {author} {\bibinfo {author} {\bibfnamefont {D.~S.}\ \bibnamefont
  {Abrams}}\ and\ \bibinfo {author} {\bibfnamefont {S.}~\bibnamefont {Lloyd}},\
  }\href@noop {} {\bibfield  {journal} {\bibinfo  {journal} {Physical Review
  Letters}\ }\textbf {\bibinfo {volume} {83}},\ \bibinfo {pages} {5162}
  (\bibinfo {year} {1999})}\BibitemShut {NoStop}%
\bibitem [{\citenamefont {Wu}\ \emph {et~al.}(2002)\citenamefont {Wu},
  \citenamefont {Byrd},\ and\ \citenamefont {Lidar}}]{wu2002polynomial}%
  \BibitemOpen
  \bibfield  {author} {\bibinfo {author} {\bibfnamefont {L.-A.}\ \bibnamefont
  {Wu}}, \bibinfo {author} {\bibfnamefont {M.}~\bibnamefont {Byrd}}, \ and\
  \bibinfo {author} {\bibfnamefont {D.}~\bibnamefont {Lidar}},\ }\href@noop {}
  {\bibfield  {journal} {\bibinfo  {journal} {Physical Review Letters}\
  }\textbf {\bibinfo {volume} {89}},\ \bibinfo {pages} {057904} (\bibinfo
  {year} {2002})}\BibitemShut {NoStop}%
\bibitem [{\citenamefont {Aspuru-Guzik}\ \emph {et~al.}(2005)\citenamefont
  {Aspuru-Guzik}, \citenamefont {Dutoi}, \citenamefont {Love},\ and\
  \citenamefont {Head-Gordon}}]{aspuru2005simulated}%
  \BibitemOpen
  \bibfield  {author} {\bibinfo {author} {\bibfnamefont {A.}~\bibnamefont
  {Aspuru-Guzik}}, \bibinfo {author} {\bibfnamefont {A.~D.}\ \bibnamefont
  {Dutoi}}, \bibinfo {author} {\bibfnamefont {P.~J.}\ \bibnamefont {Love}}, \
  and\ \bibinfo {author} {\bibfnamefont {M.}~\bibnamefont {Head-Gordon}},\
  }\href@noop {} {\bibfield  {journal} {\bibinfo  {journal} {Science}\ }\textbf
  {\bibinfo {volume} {309}},\ \bibinfo {pages} {1704} (\bibinfo {year}
  {2005})}\BibitemShut {NoStop}%
\bibitem [{\citenamefont {Jordan}\ \emph {et~al.}(2012)\citenamefont {Jordan},
  \citenamefont {Lee},\ and\ \citenamefont {Preskill}}]{preskill1}%
  \BibitemOpen
  \bibfield  {author} {\bibinfo {author} {\bibfnamefont {S.~P.}\ \bibnamefont
  {Jordan}}, \bibinfo {author} {\bibfnamefont {K.~S.~M.}\ \bibnamefont {Lee}},
  \ and\ \bibinfo {author} {\bibfnamefont {J.}~\bibnamefont {Preskill}},\
  }\href {\doibase 10.1126/science.1217069} {\bibfield  {journal} {\bibinfo
  {journal} {Science}\ }\textbf {\bibinfo {volume} {336}},\ \bibinfo {pages}
  {1130} (\bibinfo {year} {2012})},\ \Eprint {http://arxiv.org/abs/1111.3633}
  {arXiv:1111.3633 [quant-ph]} \BibitemShut {NoStop}%
%%CITATION = ARXIV:1111.3633;%%
\bibitem [{\citenamefont {Babbush}\ \emph
  {et~al.}(2018{\natexlab{a}})\citenamefont {Babbush}, \citenamefont {Gidney},
  \citenamefont {Berry}, \citenamefont {Wiebe}, \citenamefont {McClean},
  \citenamefont {Paler}, \citenamefont {Fowler},\ and\ \citenamefont
  {Neven}}]{babbush2018encoding}%
  \BibitemOpen
  \bibfield  {author} {\bibinfo {author} {\bibfnamefont {R.}~\bibnamefont
  {Babbush}}, \bibinfo {author} {\bibfnamefont {C.}~\bibnamefont {Gidney}},
  \bibinfo {author} {\bibfnamefont {D.~W.}\ \bibnamefont {Berry}}, \bibinfo
  {author} {\bibfnamefont {N.}~\bibnamefont {Wiebe}}, \bibinfo {author}
  {\bibfnamefont {J.}~\bibnamefont {McClean}}, \bibinfo {author} {\bibfnamefont
  {A.}~\bibnamefont {Paler}}, \bibinfo {author} {\bibfnamefont
  {A.}~\bibnamefont {Fowler}}, \ and\ \bibinfo {author} {\bibfnamefont
  {H.}~\bibnamefont {Neven}},\ }\href {\doibase 10.1103/PhysRevX.8.041015}
  {\bibfield  {journal} {\bibinfo  {journal} {Phys. Rev. X}\ }\textbf {\bibinfo
  {volume} {8}},\ \bibinfo {pages} {041015} (\bibinfo {year}
  {2018}{\natexlab{a}})}\BibitemShut {NoStop}%
\bibitem [{\citenamefont {Preskill}(2018)}]{preskill2018quantum}%
  \BibitemOpen
  \bibfield  {author} {\bibinfo {author} {\bibfnamefont {J.}~\bibnamefont
  {Preskill}},\ }\href@noop {} {\bibfield  {journal} {\bibinfo  {journal}
  {Quantum}\ }\textbf {\bibinfo {volume} {2}},\ \bibinfo {pages} {79} (\bibinfo
  {year} {2018})}\BibitemShut {NoStop}%
\bibitem [{\citenamefont {Arute}\ \emph {et~al.}(2019)\citenamefont {Arute}
  \emph {et~al.}}]{arute2019quantum}%
  \BibitemOpen
  \bibfield  {author} {\bibinfo {author} {\bibfnamefont {F.}~\bibnamefont
  {Arute}} \emph {et~al.},\ }\href@noop {} {\bibfield  {journal} {\bibinfo
  {journal} {Nature}\ }\textbf {\bibinfo {volume} {574}},\ \bibinfo {pages}
  {505} (\bibinfo {year} {2019})}\BibitemShut {NoStop}%
\bibitem [{\citenamefont {Boixo}\ \emph {et~al.}(2018)\citenamefont {Boixo},
  \citenamefont {Isakov}, \citenamefont {Smelyanskiy}, \citenamefont {Babbush},
  \citenamefont {Ding}, \citenamefont {Jiang}, \citenamefont {Bremner},
  \citenamefont {Martinis},\ and\ \citenamefont
  {Neven}}]{boixo2016characterizing}%
  \BibitemOpen
  \bibfield  {author} {\bibinfo {author} {\bibfnamefont {S.}~\bibnamefont
  {Boixo}}, \bibinfo {author} {\bibfnamefont {S.~V.}\ \bibnamefont {Isakov}},
  \bibinfo {author} {\bibfnamefont {V.~N.}\ \bibnamefont {Smelyanskiy}},
  \bibinfo {author} {\bibfnamefont {R.}~\bibnamefont {Babbush}}, \bibinfo
  {author} {\bibfnamefont {N.}~\bibnamefont {Ding}}, \bibinfo {author}
  {\bibfnamefont {Z.}~\bibnamefont {Jiang}}, \bibinfo {author} {\bibfnamefont
  {M.~J.}\ \bibnamefont {Bremner}}, \bibinfo {author} {\bibfnamefont {J.~M.}\
  \bibnamefont {Martinis}}, \ and\ \bibinfo {author} {\bibfnamefont
  {H.}~\bibnamefont {Neven}},\ }\href@noop {} {\bibfield  {journal} {\bibinfo
  {journal} {Nature Physics}\ }\textbf {\bibinfo {volume} {14}},\ \bibinfo
  {pages} {595} (\bibinfo {year} {2018})}\BibitemShut {NoStop}%
\bibitem [{\citenamefont {Harrow}\ and\ \citenamefont
  {Montanaro}(2017)}]{harrow2017quantum}%
  \BibitemOpen
  \bibfield  {author} {\bibinfo {author} {\bibfnamefont {A.~W.}\ \bibnamefont
  {Harrow}}\ and\ \bibinfo {author} {\bibfnamefont {A.}~\bibnamefont
  {Montanaro}},\ }\href@noop {} {\bibfield  {journal} {\bibinfo  {journal}
  {Nature}\ }\textbf {\bibinfo {volume} {549}},\ \bibinfo {pages} {203}
  (\bibinfo {year} {2017})}\BibitemShut {NoStop}%
\bibitem [{\citenamefont {Peruzzo}\ \emph {et~al.}(2014)\citenamefont
  {Peruzzo}, \citenamefont {McClean}, \citenamefont {Shadbolt}, \citenamefont
  {Yung}, \citenamefont {Zhou}, \citenamefont {Love}, \citenamefont
  {Aspuru-Guzik},\ and\ \citenamefont {O'brien}}]{peruzzo2014variational}%
  \BibitemOpen
  \bibfield  {author} {\bibinfo {author} {\bibfnamefont {A.}~\bibnamefont
  {Peruzzo}}, \bibinfo {author} {\bibfnamefont {J.}~\bibnamefont {McClean}},
  \bibinfo {author} {\bibfnamefont {P.}~\bibnamefont {Shadbolt}}, \bibinfo
  {author} {\bibfnamefont {M.-H.}\ \bibnamefont {Yung}}, \bibinfo {author}
  {\bibfnamefont {X.-Q.}\ \bibnamefont {Zhou}}, \bibinfo {author}
  {\bibfnamefont {P.~J.}\ \bibnamefont {Love}}, \bibinfo {author}
  {\bibfnamefont {A.}~\bibnamefont {Aspuru-Guzik}}, \ and\ \bibinfo {author}
  {\bibfnamefont {J.~L.}\ \bibnamefont {O'brien}},\ }\href@noop {} {\bibfield
  {journal} {\bibinfo  {journal} {Nature communications}\ }\textbf {\bibinfo
  {volume} {5}},\ \bibinfo {pages} {4213} (\bibinfo {year} {2014})}\BibitemShut
  {NoStop}%
\bibitem [{\citenamefont {Barrett}\ \emph {et~al.}(2013)\citenamefont
  {Barrett}, \citenamefont {Hammerer}, \citenamefont {Harrison}, \citenamefont
  {Northup},\ and\ \citenamefont {Osborne}}]{barrett2013simulating}%
  \BibitemOpen
  \bibfield  {author} {\bibinfo {author} {\bibfnamefont {S.}~\bibnamefont
  {Barrett}}, \bibinfo {author} {\bibfnamefont {K.}~\bibnamefont {Hammerer}},
  \bibinfo {author} {\bibfnamefont {S.}~\bibnamefont {Harrison}}, \bibinfo
  {author} {\bibfnamefont {T.~E.}\ \bibnamefont {Northup}}, \ and\ \bibinfo
  {author} {\bibfnamefont {T.~J.}\ \bibnamefont {Osborne}},\ }\href@noop {}
  {\bibfield  {journal} {\bibinfo  {journal} {Physical review letters}\
  }\textbf {\bibinfo {volume} {110}},\ \bibinfo {pages} {090501} (\bibinfo
  {year} {2013})}\BibitemShut {NoStop}%
\bibitem [{\citenamefont {Farhi}\ \emph {et~al.}(2014)\citenamefont {Farhi},
  \citenamefont {Goldstone},\ and\ \citenamefont {Gutmann}}]{farhi2014quantum}%
  \BibitemOpen
  \bibfield  {author} {\bibinfo {author} {\bibfnamefont {E.}~\bibnamefont
  {Farhi}}, \bibinfo {author} {\bibfnamefont {J.}~\bibnamefont {Goldstone}}, \
  and\ \bibinfo {author} {\bibfnamefont {S.}~\bibnamefont {Gutmann}},\
  }\href@noop {} {\bibfield  {journal} {\bibinfo  {journal} {arXiv preprint
  arXiv:1411.4028}\ } (\bibinfo {year} {2014})}\BibitemShut {NoStop}%
\bibitem [{\citenamefont {Wang}\ \emph {et~al.}(2015)\citenamefont {Wang},
  \citenamefont {Dolde}, \citenamefont {Biamonte}, \citenamefont {Babbush},
  \citenamefont {Bergholm}, \citenamefont {Yang}, \citenamefont {Jakobi},
  \citenamefont {Neumann}, \citenamefont {Aspuru-Guzik}, \citenamefont
  {Whitfield},\ and\ \citenamefont {Wrachtrup}}]{wang2015quantum}%
  \BibitemOpen
  \bibfield  {author} {\bibinfo {author} {\bibfnamefont {Y.}~\bibnamefont
  {Wang}}, \bibinfo {author} {\bibfnamefont {F.}~\bibnamefont {Dolde}},
  \bibinfo {author} {\bibfnamefont {J.}~\bibnamefont {Biamonte}}, \bibinfo
  {author} {\bibfnamefont {R.}~\bibnamefont {Babbush}}, \bibinfo {author}
  {\bibfnamefont {V.}~\bibnamefont {Bergholm}}, \bibinfo {author}
  {\bibfnamefont {S.}~\bibnamefont {Yang}}, \bibinfo {author} {\bibfnamefont
  {I.}~\bibnamefont {Jakobi}}, \bibinfo {author} {\bibfnamefont
  {P.}~\bibnamefont {Neumann}}, \bibinfo {author} {\bibfnamefont
  {A.}~\bibnamefont {Aspuru-Guzik}}, \bibinfo {author} {\bibfnamefont {J.~D.}\
  \bibnamefont {Whitfield}}, \ and\ \bibinfo {author} {\bibfnamefont
  {J.}~\bibnamefont {Wrachtrup}},\ }\href@noop {} {\bibfield  {journal}
  {\bibinfo  {journal} {ACS Nano}\ }\textbf {\bibinfo {volume} {9}},\ \bibinfo
  {pages} {7769} (\bibinfo {year} {2015})}\BibitemShut {NoStop}%
\bibitem [{\citenamefont {O'Malley}\ \emph {et~al.}(2016)\citenamefont
  {O'Malley}, \citenamefont {Babbush}, \citenamefont {Kivlichan}, \citenamefont
  {Romero}, \citenamefont {McClean}, \citenamefont {Barends}, \citenamefont
  {Kelly}, \citenamefont {Roushan}, \citenamefont {Tranter}, \citenamefont
  {Ding}, \citenamefont {Campbell}, \citenamefont {Chen}, \citenamefont {Chen},
  \citenamefont {Chiaro}, \citenamefont {Dunsworth}, \citenamefont {Fowler},
  \citenamefont {Jeffrey}, \citenamefont {Lucero}, \citenamefont {Megrant},
  \citenamefont {Mutus}, \citenamefont {Neeley}, \citenamefont {Neill},
  \citenamefont {Quintana}, \citenamefont {Sank}, \citenamefont {Vainsencher},
  \citenamefont {Wenner}, \citenamefont {White}, \citenamefont {Coveney},
  \citenamefont {Love}, \citenamefont {Neven}, \citenamefont {Aspuru-Guzik},\
  and\ \citenamefont {Martinis}}]{omalley16a}%
  \BibitemOpen
  \bibfield  {author} {\bibinfo {author} {\bibfnamefont {P.~J.~J.}\
  \bibnamefont {O'Malley}}, \bibinfo {author} {\bibfnamefont {R.}~\bibnamefont
  {Babbush}}, \bibinfo {author} {\bibfnamefont {I.~D.}\ \bibnamefont
  {Kivlichan}}, \bibinfo {author} {\bibfnamefont {J.}~\bibnamefont {Romero}},
  \bibinfo {author} {\bibfnamefont {J.~R.}\ \bibnamefont {McClean}}, \bibinfo
  {author} {\bibfnamefont {R.}~\bibnamefont {Barends}}, \bibinfo {author}
  {\bibfnamefont {J.}~\bibnamefont {Kelly}}, \bibinfo {author} {\bibfnamefont
  {P.}~\bibnamefont {Roushan}}, \bibinfo {author} {\bibfnamefont
  {A.}~\bibnamefont {Tranter}}, \bibinfo {author} {\bibfnamefont
  {N.}~\bibnamefont {Ding}}, \bibinfo {author} {\bibfnamefont {B.}~\bibnamefont
  {Campbell}}, \bibinfo {author} {\bibfnamefont {Y.}~\bibnamefont {Chen}},
  \bibinfo {author} {\bibfnamefont {Z.}~\bibnamefont {Chen}}, \bibinfo {author}
  {\bibfnamefont {B.}~\bibnamefont {Chiaro}}, \bibinfo {author} {\bibfnamefont
  {A.}~\bibnamefont {Dunsworth}}, \bibinfo {author} {\bibfnamefont {A.~G.}\
  \bibnamefont {Fowler}}, \bibinfo {author} {\bibfnamefont {E.}~\bibnamefont
  {Jeffrey}}, \bibinfo {author} {\bibfnamefont {E.}~\bibnamefont {Lucero}},
  \bibinfo {author} {\bibfnamefont {A.}~\bibnamefont {Megrant}}, \bibinfo
  {author} {\bibfnamefont {J.~Y.}\ \bibnamefont {Mutus}}, \bibinfo {author}
  {\bibfnamefont {M.}~\bibnamefont {Neeley}}, \bibinfo {author} {\bibfnamefont
  {C.}~\bibnamefont {Neill}}, \bibinfo {author} {\bibfnamefont
  {C.}~\bibnamefont {Quintana}}, \bibinfo {author} {\bibfnamefont
  {D.}~\bibnamefont {Sank}}, \bibinfo {author} {\bibfnamefont {A.}~\bibnamefont
  {Vainsencher}}, \bibinfo {author} {\bibfnamefont {J.}~\bibnamefont {Wenner}},
  \bibinfo {author} {\bibfnamefont {T.~C.}\ \bibnamefont {White}}, \bibinfo
  {author} {\bibfnamefont {P.~V.}\ \bibnamefont {Coveney}}, \bibinfo {author}
  {\bibfnamefont {P.~J.}\ \bibnamefont {Love}}, \bibinfo {author}
  {\bibfnamefont {H.}~\bibnamefont {Neven}}, \bibinfo {author} {\bibfnamefont
  {A.}~\bibnamefont {Aspuru-Guzik}}, \ and\ \bibinfo {author} {\bibfnamefont
  {J.~M.}\ \bibnamefont {Martinis}},\ }\href {\doibase
  10.1103/PhysRevX.6.031007} {\bibfield  {journal} {\bibinfo  {journal} {Phys.
  Rev. X}\ }\textbf {\bibinfo {volume} {6}},\ \bibinfo {pages} {031007}
  (\bibinfo {year} {2016})}\BibitemShut {NoStop}%
\bibitem [{\citenamefont {Kandala}\ \emph {et~al.}(2017)\citenamefont
  {Kandala}, \citenamefont {Mezzacapo}, \citenamefont {Temme}, \citenamefont
  {Takita}, \citenamefont {Brink}, \citenamefont {Chow},\ and\ \citenamefont
  {Gambetta}}]{kandala2017hardware}%
  \BibitemOpen
  \bibfield  {author} {\bibinfo {author} {\bibfnamefont {A.}~\bibnamefont
  {Kandala}}, \bibinfo {author} {\bibfnamefont {A.}~\bibnamefont {Mezzacapo}},
  \bibinfo {author} {\bibfnamefont {K.}~\bibnamefont {Temme}}, \bibinfo
  {author} {\bibfnamefont {M.}~\bibnamefont {Takita}}, \bibinfo {author}
  {\bibfnamefont {M.}~\bibnamefont {Brink}}, \bibinfo {author} {\bibfnamefont
  {J.~M.}\ \bibnamefont {Chow}}, \ and\ \bibinfo {author} {\bibfnamefont
  {J.~M.}\ \bibnamefont {Gambetta}},\ }\href@noop {} {\bibfield  {journal}
  {\bibinfo  {journal} {Nature}\ }\textbf {\bibinfo {volume} {549}},\ \bibinfo
  {pages} {242} (\bibinfo {year} {2017})}\BibitemShut {NoStop}%
\bibitem [{\citenamefont {Hempel}\ \emph {et~al.}(2018)\citenamefont {Hempel},
  \citenamefont {Maier}, \citenamefont {Romero}, \citenamefont {McClean},
  \citenamefont {Monz}, \citenamefont {Shen}, \citenamefont {Jurcevic},
  \citenamefont {Lanyon}, \citenamefont {Love}, \citenamefont {Babbush},
  \citenamefont {Aspuru-Guzik}, \citenamefont {Blatt},\ and\ \citenamefont
  {Roos}}]{PhysRevX.8.031022}%
  \BibitemOpen
  \bibfield  {author} {\bibinfo {author} {\bibfnamefont {C.}~\bibnamefont
  {Hempel}}, \bibinfo {author} {\bibfnamefont {C.}~\bibnamefont {Maier}},
  \bibinfo {author} {\bibfnamefont {J.}~\bibnamefont {Romero}}, \bibinfo
  {author} {\bibfnamefont {J.}~\bibnamefont {McClean}}, \bibinfo {author}
  {\bibfnamefont {T.}~\bibnamefont {Monz}}, \bibinfo {author} {\bibfnamefont
  {H.}~\bibnamefont {Shen}}, \bibinfo {author} {\bibfnamefont {P.}~\bibnamefont
  {Jurcevic}}, \bibinfo {author} {\bibfnamefont {B.~P.}\ \bibnamefont
  {Lanyon}}, \bibinfo {author} {\bibfnamefont {P.}~\bibnamefont {Love}},
  \bibinfo {author} {\bibfnamefont {R.}~\bibnamefont {Babbush}}, \bibinfo
  {author} {\bibfnamefont {A.}~\bibnamefont {Aspuru-Guzik}}, \bibinfo {author}
  {\bibfnamefont {R.}~\bibnamefont {Blatt}}, \ and\ \bibinfo {author}
  {\bibfnamefont {C.~F.}\ \bibnamefont {Roos}},\ }\href {\doibase
  10.1103/PhysRevX.8.031022} {\bibfield  {journal} {\bibinfo  {journal} {Phys.
  Rev. X}\ }\textbf {\bibinfo {volume} {8}},\ \bibinfo {pages} {031022}
  (\bibinfo {year} {2018})}\BibitemShut {NoStop}%
\bibitem [{\citenamefont {Dumitrescu}\ \emph {et~al.}(2018)\citenamefont
  {Dumitrescu}, \citenamefont {McCaskey}, \citenamefont {Hagen}, \citenamefont
  {Jansen}, \citenamefont {Morris}, \citenamefont {Papenbrock}, \citenamefont
  {Pooser}, \citenamefont {Dean},\ and\ \citenamefont
  {Lougovski}}]{dumitrescu18a}%
  \BibitemOpen
  \bibfield  {author} {\bibinfo {author} {\bibfnamefont {E.~F.}\ \bibnamefont
  {Dumitrescu}}, \bibinfo {author} {\bibfnamefont {A.~J.}\ \bibnamefont
  {McCaskey}}, \bibinfo {author} {\bibfnamefont {G.}~\bibnamefont {Hagen}},
  \bibinfo {author} {\bibfnamefont {G.~R.}\ \bibnamefont {Jansen}}, \bibinfo
  {author} {\bibfnamefont {T.~D.}\ \bibnamefont {Morris}}, \bibinfo {author}
  {\bibfnamefont {T.}~\bibnamefont {Papenbrock}}, \bibinfo {author}
  {\bibfnamefont {R.~C.}\ \bibnamefont {Pooser}}, \bibinfo {author}
  {\bibfnamefont {D.~J.}\ \bibnamefont {Dean}}, \ and\ \bibinfo {author}
  {\bibfnamefont {P.}~\bibnamefont {Lougovski}},\ }\href@noop {} {\bibfield
  {journal} {\bibinfo  {journal} {Phys. Rev. Lett.}\ }\textbf {\bibinfo
  {volume} {120}},\ \bibinfo {pages} {210501} (\bibinfo {year}
  {2018})}\BibitemShut {NoStop}%
\bibitem [{\citenamefont {Babbush}\ \emph
  {et~al.}(2018{\natexlab{b}})\citenamefont {Babbush}, \citenamefont {Wiebe},
  \citenamefont {McClean}, \citenamefont {McClain}, \citenamefont {Neven},\
  and\ \citenamefont {Chan}}]{babbush2018low}%
  \BibitemOpen
  \bibfield  {author} {\bibinfo {author} {\bibfnamefont {R.}~\bibnamefont
  {Babbush}}, \bibinfo {author} {\bibfnamefont {N.}~\bibnamefont {Wiebe}},
  \bibinfo {author} {\bibfnamefont {J.}~\bibnamefont {McClean}}, \bibinfo
  {author} {\bibfnamefont {J.}~\bibnamefont {McClain}}, \bibinfo {author}
  {\bibfnamefont {H.}~\bibnamefont {Neven}}, \ and\ \bibinfo {author}
  {\bibfnamefont {G.~K.-L.}\ \bibnamefont {Chan}},\ }\href@noop {} {\bibfield
  {journal} {\bibinfo  {journal} {Physical Review X}\ }\textbf {\bibinfo
  {volume} {8}},\ \bibinfo {pages} {011044} (\bibinfo {year}
  {2018}{\natexlab{b}})}\BibitemShut {NoStop}%
\bibitem [{\citenamefont {Rubin}\ \emph {et~al.}(2018)\citenamefont {Rubin},
  \citenamefont {Babbush},\ and\ \citenamefont
  {McClean}}]{rubin2018application}%
  \BibitemOpen
  \bibfield  {author} {\bibinfo {author} {\bibfnamefont {N.~C.}\ \bibnamefont
  {Rubin}}, \bibinfo {author} {\bibfnamefont {R.}~\bibnamefont {Babbush}}, \
  and\ \bibinfo {author} {\bibfnamefont {J.}~\bibnamefont {McClean}},\
  }\href@noop {} {\bibfield  {journal} {\bibinfo  {journal} {New Journal of
  Physics}\ }\textbf {\bibinfo {volume} {20}},\ \bibinfo {pages} {053020}
  (\bibinfo {year} {2018})}\BibitemShut {NoStop}%
\bibitem [{\citenamefont {Wang}\ \emph {et~al.}(2019)\citenamefont {Wang},
  \citenamefont {Higgott},\ and\ \citenamefont
  {Brierley}}]{wangAcceleratedVariationalQuantum2019}%
  \BibitemOpen
  \bibfield  {author} {\bibinfo {author} {\bibfnamefont {D.}~\bibnamefont
  {Wang}}, \bibinfo {author} {\bibfnamefont {O.}~\bibnamefont {Higgott}}, \
  and\ \bibinfo {author} {\bibfnamefont {S.}~\bibnamefont {Brierley}},\ }\href
  {\doibase 10.1103/PhysRevLett.122.140504} {\bibfield  {journal} {\bibinfo
  {journal} {Physical Review Letters}\ }\textbf {\bibinfo {volume} {122}},\
  \bibinfo {pages} {140504} (\bibinfo {year} {2019})}\BibitemShut {NoStop}%
\bibitem [{\citenamefont {Verteletskyi}\ \emph {et~al.}(2019)\citenamefont
  {Verteletskyi}, \citenamefont {Yen},\ and\ \citenamefont
  {Izmaylov}}]{verteletskyi2019measurement}%
  \BibitemOpen
  \bibfield  {author} {\bibinfo {author} {\bibfnamefont {V.}~\bibnamefont
  {Verteletskyi}}, \bibinfo {author} {\bibfnamefont {T.-C.}\ \bibnamefont
  {Yen}}, \ and\ \bibinfo {author} {\bibfnamefont {A.~F.}\ \bibnamefont
  {Izmaylov}},\ }\href@noop {} {\bibfield  {journal} {\bibinfo  {journal}
  {arXiv preprint arXiv:1907.03358}\ } (\bibinfo {year} {2019})}\BibitemShut
  {NoStop}%
\bibitem [{\citenamefont {Jena}\ \emph {et~al.}(2019)\citenamefont {Jena},
  \citenamefont {Genin},\ and\ \citenamefont {Mosca}}]{jena2019pauli}%
  \BibitemOpen
  \bibfield  {author} {\bibinfo {author} {\bibfnamefont {A.}~\bibnamefont
  {Jena}}, \bibinfo {author} {\bibfnamefont {S.}~\bibnamefont {Genin}}, \ and\
  \bibinfo {author} {\bibfnamefont {M.}~\bibnamefont {Mosca}},\ }\href@noop {}
  {\bibfield  {journal} {\bibinfo  {journal} {arXiv preprint arXiv:1907.07859}\
  } (\bibinfo {year} {2019})}\BibitemShut {NoStop}%
\bibitem [{\citenamefont {Izmaylov}\ \emph {et~al.}(2020)\citenamefont
  {Izmaylov}, \citenamefont {Yen}, \citenamefont {Lang},\ and\ \citenamefont
  {Verteletskyi}}]{izmaylov19a}%
  \BibitemOpen
  \bibfield  {author} {\bibinfo {author} {\bibfnamefont {A.~F.}\ \bibnamefont
  {Izmaylov}}, \bibinfo {author} {\bibfnamefont {T.-C.}\ \bibnamefont {Yen}},
  \bibinfo {author} {\bibfnamefont {R.~A.}\ \bibnamefont {Lang}}, \ and\
  \bibinfo {author} {\bibfnamefont {V.}~\bibnamefont {Verteletskyi}},\ }\href
  {\doibase 10.1021/acs.jctc.9b00791} {\bibfield  {journal} {\bibinfo
  {journal} {Journal of Chemical Theory and Computation}\ }\textbf {\bibinfo
  {volume} {16}},\ \bibinfo {pages} {190} (\bibinfo {year} {2020})}\BibitemShut
  {NoStop}%
\bibitem [{\citenamefont {Yen}\ \emph {et~al.}(2019)\citenamefont {Yen},
  \citenamefont {Verteletsky},\ and\ \citenamefont
  {Izmaylov}}]{yen2019measuring}%
  \BibitemOpen
  \bibfield  {author} {\bibinfo {author} {\bibfnamefont {T.-C.}\ \bibnamefont
  {Yen}}, \bibinfo {author} {\bibfnamefont {V.}~\bibnamefont {Verteletsky}}, \
  and\ \bibinfo {author} {\bibfnamefont {A.~F.}\ \bibnamefont {Izmaylov}},\
  }\href@noop {} {\bibfield  {journal} {\bibinfo  {journal} {arXiv preprint
  arXiv:1907.09386}\ } (\bibinfo {year} {2019})}\BibitemShut {NoStop}%
\bibitem [{\citenamefont {Huggins}\ \emph {et~al.}(2019)\citenamefont
  {Huggins}, \citenamefont {McClean}, \citenamefont {Rubin}, \citenamefont
  {Jiang}, \citenamefont {Wiebe}, \citenamefont {Whaley},\ and\ \citenamefont
  {Babbush}}]{huggins2019efficient}%
  \BibitemOpen
  \bibfield  {author} {\bibinfo {author} {\bibfnamefont {W.~J.}\ \bibnamefont
  {Huggins}}, \bibinfo {author} {\bibfnamefont {J.}~\bibnamefont {McClean}},
  \bibinfo {author} {\bibfnamefont {N.}~\bibnamefont {Rubin}}, \bibinfo
  {author} {\bibfnamefont {Z.}~\bibnamefont {Jiang}}, \bibinfo {author}
  {\bibfnamefont {N.}~\bibnamefont {Wiebe}}, \bibinfo {author} {\bibfnamefont
  {K.~B.}\ \bibnamefont {Whaley}}, \ and\ \bibinfo {author} {\bibfnamefont
  {R.}~\bibnamefont {Babbush}},\ }\href@noop {} {\bibfield  {journal} {\bibinfo
   {journal} {arXiv preprint arXiv:1907.13117}\ } (\bibinfo {year}
  {2019})}\BibitemShut {NoStop}%
\bibitem [{\citenamefont {Gokhale}\ \emph {et~al.}(2019)\citenamefont
  {Gokhale}, \citenamefont {Angiuli}, \citenamefont {Ding}, \citenamefont
  {Gui}, \citenamefont {Tomesh}, \citenamefont {Suchara}, \citenamefont
  {Martonosi},\ and\ \citenamefont {Chong}}]{gokhale2019MinimizingStatePrep}%
  \BibitemOpen
  \bibfield  {author} {\bibinfo {author} {\bibfnamefont {P.}~\bibnamefont
  {Gokhale}}, \bibinfo {author} {\bibfnamefont {O.}~\bibnamefont {Angiuli}},
  \bibinfo {author} {\bibfnamefont {Y.}~\bibnamefont {Ding}}, \bibinfo {author}
  {\bibfnamefont {K.}~\bibnamefont {Gui}}, \bibinfo {author} {\bibfnamefont
  {T.}~\bibnamefont {Tomesh}}, \bibinfo {author} {\bibfnamefont
  {M.}~\bibnamefont {Suchara}}, \bibinfo {author} {\bibfnamefont
  {M.}~\bibnamefont {Martonosi}}, \ and\ \bibinfo {author} {\bibfnamefont
  {F.~T.}\ \bibnamefont {Chong}},\ }\href@noop {} {\bibfield  {journal}
  {\bibinfo  {journal} {arXiv preprint arXiv:1907.13623}\ } (\bibinfo {year}
  {2019})}\BibitemShut {NoStop}%
\bibitem [{\citenamefont {Bonet-Monroig}\ \emph {et~al.}(2019)\citenamefont
  {Bonet-Monroig}, \citenamefont {Babbush},\ and\ \citenamefont
  {O'Brien}}]{bonetNearlyOptimalMeas2019}%
  \BibitemOpen
  \bibfield  {author} {\bibinfo {author} {\bibfnamefont {X.}~\bibnamefont
  {Bonet-Monroig}}, \bibinfo {author} {\bibfnamefont {R.}~\bibnamefont
  {Babbush}}, \ and\ \bibinfo {author} {\bibfnamefont {T.~E.}\ \bibnamefont
  {O'Brien}},\ }\href@noop {} {\bibfield  {journal} {\bibinfo  {journal} {arXiv
  preprint arXiv:1908.05628}\ } (\bibinfo {year} {2019})}\BibitemShut {NoStop}%
\bibitem [{\citenamefont {Crawford}\ \emph {et~al.}(2019)\citenamefont
  {Crawford}, \citenamefont {van Straaten}, \citenamefont {Wang}, \citenamefont
  {Parks}, \citenamefont {Campbell},\ and\ \citenamefont
  {Brierley}}]{crawford2019efficient}%
  \BibitemOpen
  \bibfield  {author} {\bibinfo {author} {\bibfnamefont {O.}~\bibnamefont
  {Crawford}}, \bibinfo {author} {\bibfnamefont {B.}~\bibnamefont {van
  Straaten}}, \bibinfo {author} {\bibfnamefont {D.}~\bibnamefont {Wang}},
  \bibinfo {author} {\bibfnamefont {T.}~\bibnamefont {Parks}}, \bibinfo
  {author} {\bibfnamefont {E.}~\bibnamefont {Campbell}}, \ and\ \bibinfo
  {author} {\bibfnamefont {S.}~\bibnamefont {Brierley}},\ }\href@noop {}
  {\bibfield  {journal} {\bibinfo  {journal} {arXiv preprint arXiv:1908.06942}\
  } (\bibinfo {year} {2019})}\BibitemShut {NoStop}%
\bibitem [{\citenamefont {Gokhale}\ and\ \citenamefont
  {Chong}(2019)}]{gokhale2019n}%
  \BibitemOpen
  \bibfield  {author} {\bibinfo {author} {\bibfnamefont {P.}~\bibnamefont
  {Gokhale}}\ and\ \bibinfo {author} {\bibfnamefont {F.~T.}\ \bibnamefont
  {Chong}},\ }\href@noop {} {\bibfield  {journal} {\bibinfo  {journal} {arXiv
  preprint arXiv:1908.11857}\ } (\bibinfo {year} {2019})}\BibitemShut {NoStop}%
\bibitem [{\citenamefont {Torlai}\ \emph {et~al.}(2019)\citenamefont {Torlai},
  \citenamefont {Mazzola}, \citenamefont {Carleo},\ and\ \citenamefont
  {Mezzacapo}}]{torlai2019precise}%
  \BibitemOpen
  \bibfield  {author} {\bibinfo {author} {\bibfnamefont {G.}~\bibnamefont
  {Torlai}}, \bibinfo {author} {\bibfnamefont {G.}~\bibnamefont {Mazzola}},
  \bibinfo {author} {\bibfnamefont {G.}~\bibnamefont {Carleo}}, \ and\ \bibinfo
  {author} {\bibfnamefont {A.}~\bibnamefont {Mezzacapo}},\ }\href@noop {}
  {\bibfield  {journal} {\bibinfo  {journal} {arXiv preprint arXiv:1910.07596}\
  } (\bibinfo {year} {2019})}\BibitemShut {NoStop}%
\bibitem [{\citenamefont {Bravyi}\ and\ \citenamefont
  {Kitaev}(2002)}]{bravyi02a}%
  \BibitemOpen
  \bibfield  {author} {\bibinfo {author} {\bibfnamefont {S.~B.}\ \bibnamefont
  {Bravyi}}\ and\ \bibinfo {author} {\bibfnamefont {A.~Y.}\ \bibnamefont
  {Kitaev}},\ }\href {\doibase https://doi.org/10.1006/aphy.2002.6254}
  {\bibfield  {journal} {\bibinfo  {journal} {Annals of Physics}\ }\textbf
  {\bibinfo {volume} {298}},\ \bibinfo {pages} {210 } (\bibinfo {year}
  {2002})}\BibitemShut {NoStop}%
\bibitem [{\citenamefont {Seeley}\ \emph {et~al.}(2012)\citenamefont {Seeley},
  \citenamefont {Richard},\ and\ \citenamefont {Love}}]{love12}%
  \BibitemOpen
  \bibfield  {author} {\bibinfo {author} {\bibfnamefont {J.~T.}\ \bibnamefont
  {Seeley}}, \bibinfo {author} {\bibfnamefont {M.~J.}\ \bibnamefont {Richard}},
  \ and\ \bibinfo {author} {\bibfnamefont {P.~J.}\ \bibnamefont {Love}},\
  }\href@noop {} {\bibfield  {journal} {\bibinfo  {journal} {Journal of
  Chemical Physics}\ }\textbf {\bibinfo {volume} {137}},\ \bibinfo {pages}
  {224109} (\bibinfo {year} {2012})}\BibitemShut {NoStop}%
\bibitem [{\citenamefont {Somma}\ \emph {et~al.}(2002)\citenamefont {Somma},
  \citenamefont {Ortiz}, \citenamefont {Gubernatis}, \citenamefont {Knill},\
  and\ \citenamefont {Laflamme}}]{somma02a}%
  \BibitemOpen
  \bibfield  {author} {\bibinfo {author} {\bibfnamefont {R.}~\bibnamefont
  {Somma}}, \bibinfo {author} {\bibfnamefont {G.}~\bibnamefont {Ortiz}},
  \bibinfo {author} {\bibfnamefont {J.~E.}\ \bibnamefont {Gubernatis}},
  \bibinfo {author} {\bibfnamefont {E.}~\bibnamefont {Knill}}, \ and\ \bibinfo
  {author} {\bibfnamefont {R.}~\bibnamefont {Laflamme}},\ }\href {\doibase
  10.1103/PhysRevA.65.042323} {\bibfield  {journal} {\bibinfo  {journal} {Phys.
  Rev. A}\ }\textbf {\bibinfo {volume} {65}},\ \bibinfo {pages} {042323}
  (\bibinfo {year} {2002})}\BibitemShut {NoStop}%
\bibitem [{\citenamefont {Nielsen}\ and\ \citenamefont
  {Chuang}(2002)}]{nielsen2002quantum}%
  \BibitemOpen
  \bibfield  {author} {\bibinfo {author} {\bibfnamefont {M.~A.}\ \bibnamefont
  {Nielsen}}\ and\ \bibinfo {author} {\bibfnamefont {I.}~\bibnamefont
  {Chuang}},\ }\href@noop {} {\enquote {\bibinfo {title} {Quantum computation
  and quantum information},}\ } (\bibinfo {year} {2002})\BibitemShut {NoStop}%
\bibitem [{\citenamefont {Kocia}\ and\ \citenamefont {Love}(2017)}]{love17a}%
  \BibitemOpen
  \bibfield  {author} {\bibinfo {author} {\bibfnamefont {L.}~\bibnamefont
  {Kocia}}\ and\ \bibinfo {author} {\bibfnamefont {P.}~\bibnamefont {Love}},\
  }\href@noop {} {\bibfield  {journal} {\bibinfo  {journal} {Phys. Rev. A}\
  }\textbf {\bibinfo {volume} {96}},\ \bibinfo {pages} {062134} (\bibinfo
  {year} {2017})}\BibitemShut {NoStop}%
\bibitem [{\citenamefont {Childs}\ and\ \citenamefont {Wiebe}(2012)}]{LCU2012}%
  \BibitemOpen
  \bibfield  {author} {\bibinfo {author} {\bibfnamefont {A.~M.}\ \bibnamefont
  {Childs}}\ and\ \bibinfo {author} {\bibfnamefont {N.}~\bibnamefont {Wiebe}},\
  }\href {http://dl.acm.org/citation.cfm?id=2481569.2481570} {\bibfield
  {journal} {\bibinfo  {journal} {Quantum Info. Comput.}\ }\textbf {\bibinfo
  {volume} {12}},\ \bibinfo {pages} {901} (\bibinfo {year} {2012})}\BibitemShut
  {NoStop}%
\bibitem [{\citenamefont {Babbush}\ \emph {et~al.}(2019)\citenamefont
  {Babbush}, \citenamefont {Berry},\ and\ \citenamefont
  {Neven}}]{babbush2019SYK}%
  \BibitemOpen
  \bibfield  {author} {\bibinfo {author} {\bibfnamefont {R.}~\bibnamefont
  {Babbush}}, \bibinfo {author} {\bibfnamefont {D.~W.}\ \bibnamefont {Berry}},
  \ and\ \bibinfo {author} {\bibfnamefont {H.}~\bibnamefont {Neven}},\
  }\href@noop {} {\bibfield  {journal} {\bibinfo  {journal} {Physical Review
  A}\ }\textbf {\bibinfo {volume} {99}},\ \bibinfo {pages} {040301} (\bibinfo
  {year} {2019})}\BibitemShut {NoStop}%
\bibitem [{\citenamefont {Kirby}\ and\ \citenamefont {Love}(2019)}]{kirby19a}%
  \BibitemOpen
  \bibfield  {author} {\bibinfo {author} {\bibfnamefont {W.~M.}\ \bibnamefont
  {Kirby}}\ and\ \bibinfo {author} {\bibfnamefont {P.~J.}\ \bibnamefont
  {Love}},\ }\href {\doibase 10.1103/PhysRevLett.123.200501} {\bibfield
  {journal} {\bibinfo  {journal} {Phys. Rev. Lett.}\ }\textbf {\bibinfo
  {volume} {123}},\ \bibinfo {pages} {200501} (\bibinfo {year}
  {2019})}\BibitemShut {NoStop}%
\bibitem [{\citenamefont {Raussendorf}\ \emph {et~al.}(2019)\citenamefont
  {Raussendorf}, \citenamefont {Bermejo-Vega}, \citenamefont {Tyhurst},
  \citenamefont {Okay},\ and\ \citenamefont {Zurel}}]{raussendorf2019phase}%
  \BibitemOpen
  \bibfield  {author} {\bibinfo {author} {\bibfnamefont {R.}~\bibnamefont
  {Raussendorf}}, \bibinfo {author} {\bibfnamefont {J.}~\bibnamefont
  {Bermejo-Vega}}, \bibinfo {author} {\bibfnamefont {E.}~\bibnamefont
  {Tyhurst}}, \bibinfo {author} {\bibfnamefont {C.}~\bibnamefont {Okay}}, \
  and\ \bibinfo {author} {\bibfnamefont {M.}~\bibnamefont {Zurel}},\
  }\href@noop {} {\bibfield  {journal} {\bibinfo  {journal} {arXiv preprint
  arXiv:1905.05374}\ } (\bibinfo {year} {2019})}\BibitemShut {NoStop}%
\bibitem [{\citenamefont {Low}\ and\ \citenamefont {Chuang}(2019)}]{low16a}%
  \BibitemOpen
  \bibfield  {author} {\bibinfo {author} {\bibfnamefont {G.~H.}\ \bibnamefont
  {Low}}\ and\ \bibinfo {author} {\bibfnamefont {I.~L.}\ \bibnamefont
  {Chuang}},\ }\href {\doibase 10.22331/q-2019-07-12-163} {\bibfield  {journal}
  {\bibinfo  {journal} {{Quantum}}\ }\textbf {\bibinfo {volume} {3}},\ \bibinfo
  {pages} {163} (\bibinfo {year} {2019})}\BibitemShut {NoStop}%
\bibitem [{\citenamefont {Low}\ and\ \citenamefont {Chuang}(2017)}]{low17}%
  \BibitemOpen
  \bibfield  {author} {\bibinfo {author} {\bibfnamefont {G.~H.}\ \bibnamefont
  {Low}}\ and\ \bibinfo {author} {\bibfnamefont {I.~L.}\ \bibnamefont
  {Chuang}},\ }\href {\doibase 10.1103/PhysRevLett.118.010501} {\bibfield
  {journal} {\bibinfo  {journal} {Phys. Rev. Lett.}\ }\textbf {\bibinfo
  {volume} {118}},\ \bibinfo {pages} {010501} (\bibinfo {year}
  {2017})}\BibitemShut {NoStop}%
\bibitem [{\citenamefont {Poulin}\ \emph {et~al.}(2018)\citenamefont {Poulin},
  \citenamefont {Kitaev}, \citenamefont {Steiger}, \citenamefont {Hastings},\
  and\ \citenamefont {Troyer}}]{poulin18a}%
  \BibitemOpen
  \bibfield  {author} {\bibinfo {author} {\bibfnamefont {D.}~\bibnamefont
  {Poulin}}, \bibinfo {author} {\bibfnamefont {A.}~\bibnamefont {Kitaev}},
  \bibinfo {author} {\bibfnamefont {D.~S.}\ \bibnamefont {Steiger}}, \bibinfo
  {author} {\bibfnamefont {M.~B.}\ \bibnamefont {Hastings}}, \ and\ \bibinfo
  {author} {\bibfnamefont {M.}~\bibnamefont {Troyer}},\ }\href {\doibase
  10.1103/PhysRevLett.121.010501} {\bibfield  {journal} {\bibinfo  {journal}
  {Phys. Rev. Lett.}\ }\textbf {\bibinfo {volume} {121}},\ \bibinfo {pages}
  {010501} (\bibinfo {year} {2018})}\BibitemShut {NoStop}%
\bibitem [{\citenamefont {Wecker}\ \emph {et~al.}(2015)\citenamefont {Wecker},
  \citenamefont {Hastings},\ and\ \citenamefont {Troyer}}]{PhysRevA.92.042303}%
  \BibitemOpen
  \bibfield  {author} {\bibinfo {author} {\bibfnamefont {D.}~\bibnamefont
  {Wecker}}, \bibinfo {author} {\bibfnamefont {M.~B.}\ \bibnamefont
  {Hastings}}, \ and\ \bibinfo {author} {\bibfnamefont {M.}~\bibnamefont
  {Troyer}},\ }\href {\doibase 10.1103/PhysRevA.92.042303} {\bibfield
  {journal} {\bibinfo  {journal} {Phys. Rev. A}\ }\textbf {\bibinfo {volume}
  {92}},\ \bibinfo {pages} {042303} (\bibinfo {year} {2015})}\BibitemShut
  {NoStop}%
\bibitem [{\citenamefont {Planat}\ and\ \citenamefont
  {Saniga}(2008)}]{planat2008pauli}%
  \BibitemOpen
  \bibfield  {author} {\bibinfo {author} {\bibfnamefont {M.}~\bibnamefont
  {Planat}}\ and\ \bibinfo {author} {\bibfnamefont {M.}~\bibnamefont
  {Saniga}},\ }\href@noop {} {\bibfield  {journal} {\bibinfo  {journal}
  {Quantum Information and Computation}\ }\textbf {\bibinfo {volume} {8}},\
  \bibinfo {pages} {0127} (\bibinfo {year} {2008})}\BibitemShut {NoStop}%
\bibitem [{\citenamefont {Bollob{\'a}s}(1988)}]{bollobas1988chromatic}%
  \BibitemOpen
  \bibfield  {author} {\bibinfo {author} {\bibfnamefont {B.}~\bibnamefont
  {Bollob{\'a}s}},\ }\href@noop {} {\bibfield  {journal} {\bibinfo  {journal}
  {Combinatorica}\ }\textbf {\bibinfo {volume} {8}},\ \bibinfo {pages} {49}
  (\bibinfo {year} {1988})}\BibitemShut {NoStop}%
\bibitem [{\citenamefont {Erd\H{o}s}\ and\ \citenamefont
  {R{\'e}nyi}(1960)}]{erds1960evolution}%
  \BibitemOpen
  \bibfield  {author} {\bibinfo {author} {\bibfnamefont {P.}~\bibnamefont
  {Erd\H{o}s}}\ and\ \bibinfo {author} {\bibfnamefont {A.}~\bibnamefont
  {R{\'e}nyi}},\ }\href@noop {} {\bibfield  {journal} {\bibinfo  {journal}
  {Publ. Math. Inst. Hungar. Acad. Sci}\ }\textbf {\bibinfo {volume} {5}},\
  \bibinfo {pages} {17} (\bibinfo {year} {1960})}\BibitemShut {NoStop}%
\bibitem [{\citenamefont {Olson}\ \emph {et~al.}(2017)\citenamefont {Olson},
  \citenamefont {Cao}, \citenamefont {Romero}, \citenamefont {Johnson},
  \citenamefont {Dallaire-Demers}, \citenamefont {Sawaya}, \citenamefont
  {Narang}, \citenamefont {Kivlichan}, \citenamefont {Wasielewski},\ and\
  \citenamefont {Aspuru-Guzik}}]{olsonQuantumInformationComputation2017}%
  \BibitemOpen
  \bibfield  {author} {\bibinfo {author} {\bibfnamefont {J.}~\bibnamefont
  {Olson}}, \bibinfo {author} {\bibfnamefont {Y.}~\bibnamefont {Cao}}, \bibinfo
  {author} {\bibfnamefont {J.}~\bibnamefont {Romero}}, \bibinfo {author}
  {\bibfnamefont {P.}~\bibnamefont {Johnson}}, \bibinfo {author} {\bibfnamefont
  {P.-L.}\ \bibnamefont {Dallaire-Demers}}, \bibinfo {author} {\bibfnamefont
  {N.}~\bibnamefont {Sawaya}}, \bibinfo {author} {\bibfnamefont
  {P.}~\bibnamefont {Narang}}, \bibinfo {author} {\bibfnamefont
  {I.}~\bibnamefont {Kivlichan}}, \bibinfo {author} {\bibfnamefont
  {M.}~\bibnamefont {Wasielewski}}, \ and\ \bibinfo {author} {\bibfnamefont
  {A.}~\bibnamefont {Aspuru-Guzik}},\ }\href@noop {} {\bibfield  {journal}
  {\bibinfo  {journal} {arXiv preprint arXiv:1706.05413}\ } (\bibinfo {year}
  {2017})}\BibitemShut {NoStop}%
\bibitem [{\citenamefont {Jordan}\ and\ \citenamefont
  {Wigner}(1928)}]{jordan28a}%
  \BibitemOpen
  \bibfield  {author} {\bibinfo {author} {\bibfnamefont {P.}~\bibnamefont
  {Jordan}}\ and\ \bibinfo {author} {\bibfnamefont {E.}~\bibnamefont
  {Wigner}},\ }\href@noop {} {\bibfield  {journal} {\bibinfo  {journal} {Z.
  Phys.}\ }\textbf {\bibinfo {volume} {47}},\ \bibinfo {pages} {631} (\bibinfo
  {year} {1928})}\BibitemShut {NoStop}%
\bibitem [{\citenamefont {Tranter}\ \emph {et~al.}(2015)\citenamefont
  {Tranter}, \citenamefont {Sofia}, \citenamefont {Seeley}, \citenamefont
  {Kaicher}, \citenamefont {Mcclean}, \citenamefont {Babbush}, \citenamefont
  {Coveney}, \citenamefont {Mintert}, \citenamefont {Wilhelm},\ and\
  \citenamefont {Love}}]{tranter15a}%
  \BibitemOpen
  \bibfield  {author} {\bibinfo {author} {\bibfnamefont {A.}~\bibnamefont
  {Tranter}}, \bibinfo {author} {\bibfnamefont {S.}~\bibnamefont {Sofia}},
  \bibinfo {author} {\bibfnamefont {J.}~\bibnamefont {Seeley}}, \bibinfo
  {author} {\bibfnamefont {M.}~\bibnamefont {Kaicher}}, \bibinfo {author}
  {\bibfnamefont {J.}~\bibnamefont {Mcclean}}, \bibinfo {author} {\bibfnamefont
  {R.}~\bibnamefont {Babbush}}, \bibinfo {author} {\bibfnamefont
  {P.}~\bibnamefont {Coveney}}, \bibinfo {author} {\bibfnamefont
  {F.}~\bibnamefont {Mintert}}, \bibinfo {author} {\bibfnamefont
  {F.}~\bibnamefont {Wilhelm}}, \ and\ \bibinfo {author} {\bibfnamefont
  {P.}~\bibnamefont {Love}},\ }\href {\doibase 10.1002/qua.24969} {\bibfield
  {journal} {\bibinfo  {journal} {International Journal of Quantum Chemistry}\
  }\textbf {\bibinfo {volume} {115}} (\bibinfo {year} {2015}),\
  10.1002/qua.24969}\BibitemShut {NoStop}%
\bibitem [{\citenamefont {Setia}\ and\ \citenamefont
  {Whitfield}(2018)}]{setia17a}%
  \BibitemOpen
  \bibfield  {author} {\bibinfo {author} {\bibfnamefont {K.}~\bibnamefont
  {Setia}}\ and\ \bibinfo {author} {\bibfnamefont {J.~D.}\ \bibnamefont
  {Whitfield}},\ }\href@noop {} {\bibfield  {journal} {\bibinfo  {journal} {The
  Journal of Chemical Physics}\ }\textbf {\bibinfo {volume} {148}},\ \bibinfo
  {pages} {164104} (\bibinfo {year} {2018})}\BibitemShut {NoStop}%
\bibitem [{\citenamefont {McClean}\ \emph {et~al.}(2016)\citenamefont
  {McClean}, \citenamefont {Romero}, \citenamefont {Babbush},\ and\
  \citenamefont {{Aspuru-Guzik}}}]{mccleanTheoryVariationalHybrid2016}%
  \BibitemOpen
  \bibfield  {author} {\bibinfo {author} {\bibfnamefont {J.~R.}\ \bibnamefont
  {McClean}}, \bibinfo {author} {\bibfnamefont {J.}~\bibnamefont {Romero}},
  \bibinfo {author} {\bibfnamefont {R.}~\bibnamefont {Babbush}}, \ and\
  \bibinfo {author} {\bibfnamefont {A.}~\bibnamefont {{Aspuru-Guzik}}},\ }\href
  {\doibase 10.1088/1367-2630/18/2/023023} {\bibfield  {journal} {\bibinfo
  {journal} {New Journal of Physics}\ }\textbf {\bibinfo {volume} {18}},\
  \bibinfo {pages} {023023} (\bibinfo {year} {2016})}\BibitemShut {NoStop}%
\bibitem [{\citenamefont {Romero}\ \emph {et~al.}(2018)\citenamefont {Romero},
  \citenamefont {Babbush}, \citenamefont {McClean}, \citenamefont {Hempel},
  \citenamefont {Love},\ and\ \citenamefont
  {{Aspuru-Guzik}}}]{romeroStrategiesQuantumComputing2018}%
  \BibitemOpen
  \bibfield  {author} {\bibinfo {author} {\bibfnamefont {J.}~\bibnamefont
  {Romero}}, \bibinfo {author} {\bibfnamefont {R.}~\bibnamefont {Babbush}},
  \bibinfo {author} {\bibfnamefont {J.~R.}\ \bibnamefont {McClean}}, \bibinfo
  {author} {\bibfnamefont {C.}~\bibnamefont {Hempel}}, \bibinfo {author}
  {\bibfnamefont {P.~J.}\ \bibnamefont {Love}}, \ and\ \bibinfo {author}
  {\bibfnamefont {A.}~\bibnamefont {{Aspuru-Guzik}}},\ }\href {\doibase
  10.1088/2058-9565/aad3e4} {\bibfield  {journal} {\bibinfo  {journal} {Quantum
  Science and Technology}\ }\textbf {\bibinfo {volume} {4}},\ \bibinfo {pages}
  {014008} (\bibinfo {year} {2018})}\BibitemShut {NoStop}%
\bibitem [{\citenamefont {Lee}\ \emph {et~al.}(2019)\citenamefont {Lee},
  \citenamefont {Huggins}, \citenamefont {{Head-Gordon}},\ and\ \citenamefont
  {Whaley}}]{leeGeneralizedUnitaryCoupled2019}%
  \BibitemOpen
  \bibfield  {author} {\bibinfo {author} {\bibfnamefont {J.}~\bibnamefont
  {Lee}}, \bibinfo {author} {\bibfnamefont {W.~J.}\ \bibnamefont {Huggins}},
  \bibinfo {author} {\bibfnamefont {M.}~\bibnamefont {{Head-Gordon}}}, \ and\
  \bibinfo {author} {\bibfnamefont {K.~B.}\ \bibnamefont {Whaley}},\ }\href
  {\doibase 10.1021/acs.jctc.8b01004} {\bibfield  {journal} {\bibinfo
  {journal} {Journal of Chemical Theory and Computation}\ }\textbf {\bibinfo
  {volume} {15}},\ \bibinfo {pages} {311} (\bibinfo {year} {2019})}\BibitemShut
  {NoStop}%
\bibitem [{\citenamefont {Garey}\ \emph {et~al.}(1974)\citenamefont {Garey},
  \citenamefont {Johnson},\ and\ \citenamefont
  {Stockmeyer}}]{gareySimplifiedNPcompleteProblems1974a}%
  \BibitemOpen
  \bibfield  {author} {\bibinfo {author} {\bibfnamefont {M.~R.}\ \bibnamefont
  {Garey}}, \bibinfo {author} {\bibfnamefont {D.~S.}\ \bibnamefont {Johnson}},
  \ and\ \bibinfo {author} {\bibfnamefont {L.}~\bibnamefont {Stockmeyer}},\
  }in\ \href {\doibase 10.1145/800119.803884} {\emph {\bibinfo {booktitle}
  {Proceedings of the {{Sixth Annual ACM Symposium}} on {{Theory}} of
  {{Computing}}}}},\ \bibinfo {series and number} {{{STOC}} '74}\ (\bibinfo
  {publisher} {{ACM}},\ \bibinfo {address} {{New York, NY, USA}},\ \bibinfo
  {year} {1974})\ pp.\ \bibinfo {pages} {47--63}\BibitemShut {NoStop}%
\bibitem [{\citenamefont {Kosowski}\ and\ \citenamefont
  {Manuszewski}(2004)}]{kosowskiClassicalColoringGraphs2004}%
  \BibitemOpen
  \bibfield  {author} {\bibinfo {author} {\bibfnamefont {A.}~\bibnamefont
  {Kosowski}}\ and\ \bibinfo {author} {\bibfnamefont {K.}~\bibnamefont
  {Manuszewski}},\ }in\ \href {\doibase 10.1090/conm/352/06369} {\emph
  {\bibinfo {booktitle} {Contemporary {{Mathematics}}}}},\ Vol.\ \bibinfo
  {volume} {352},\ \bibinfo {editor} {edited by\ \bibinfo {editor}
  {\bibfnamefont {M.}~\bibnamefont {Kubale}}}\ (\bibinfo  {publisher}
  {{American Mathematical Society}},\ \bibinfo {address} {{Providence, Rhode
  Island}},\ \bibinfo {year} {2004})\ pp.\ \bibinfo {pages} {1--19}\BibitemShut
  {NoStop}%
\bibitem [{\citenamefont {Tranter}\ \emph {et~al.}(2018)\citenamefont
  {Tranter}, \citenamefont {Love}, \citenamefont {Mintert},\ and\ \citenamefont
  {Coveney}}]{tranterComparisonBravyiKitaev2018}%
  \BibitemOpen
  \bibfield  {author} {\bibinfo {author} {\bibfnamefont {A.}~\bibnamefont
  {Tranter}}, \bibinfo {author} {\bibfnamefont {P.~J.}\ \bibnamefont {Love}},
  \bibinfo {author} {\bibfnamefont {F.}~\bibnamefont {Mintert}}, \ and\
  \bibinfo {author} {\bibfnamefont {P.~V.}\ \bibnamefont {Coveney}},\ }\href
  {\doibase 10.1021/acs.jctc.8b00450} {\bibfield  {journal} {\bibinfo
  {journal} {Journal of Chemical Theory and Computation}\ }\textbf {\bibinfo
  {volume} {14}},\ \bibinfo {pages} {5617} (\bibinfo {year}
  {2018})}\BibitemShut {NoStop}%
\bibitem [{\citenamefont {Tranter}\ \emph {et~al.}(2019)\citenamefont
  {Tranter}, \citenamefont {Love}, \citenamefont {Mintert}, \citenamefont
  {Wiebe},\ and\ \citenamefont
  {Coveney}}]{tranterOrderingTrotterizationImpact2019}%
  \BibitemOpen
  \bibfield  {author} {\bibinfo {author} {\bibfnamefont {A.}~\bibnamefont
  {Tranter}}, \bibinfo {author} {\bibfnamefont {P.~J.}\ \bibnamefont {Love}},
  \bibinfo {author} {\bibfnamefont {F.}~\bibnamefont {Mintert}}, \bibinfo
  {author} {\bibfnamefont {N.}~\bibnamefont {Wiebe}}, \ and\ \bibinfo {author}
  {\bibfnamefont {P.~V.}\ \bibnamefont {Coveney}},\ }\href {\doibase
  10.3390/e21121218} {\bibfield  {journal} {\bibinfo  {journal} {Entropy}\
  }\textbf {\bibinfo {volume} {21}},\ \bibinfo {pages} {1218} (\bibinfo {year}
  {2019})}\BibitemShut {NoStop}%
\bibitem [{\citenamefont
  {Johnson~III}(2016)}]{johnsoniiiNISTComputationalChemistry2016}%
  \BibitemOpen
  \bibfield  {author} {\bibinfo {author} {\bibfnamefont {R.~D.}\ \bibnamefont
  {Johnson~III}},\ }\href@noop {} {\emph {\bibinfo {title} {{{NIST
  Computational Chemistry Comparison}} and {{Benchmark Database NIST Standard
  Reference Database Number}} 101 {{Release}} 18}}}\ (\bibinfo {year}
  {2016})\BibitemShut {NoStop}%
\bibitem [{\citenamefont {Parrish}\ \emph {et~al.}(2017)\citenamefont
  {Parrish}, \citenamefont {Burns}, \citenamefont {Smith}, \citenamefont
  {Simmonett}, \citenamefont {DePrince}, \citenamefont {Hohenstein},
  \citenamefont {Bozkaya}, \citenamefont {Sokolov}, \citenamefont {Di~Remigio},
  \citenamefont {Richard}, \citenamefont {Gonthier}, \citenamefont {James},
  \citenamefont {McAlexander}, \citenamefont {Kumar}, \citenamefont {Saitow},
  \citenamefont {Wang}, \citenamefont {Pritchard}, \citenamefont {Verma},
  \citenamefont {Schaefer}, \citenamefont {Patkowski}, \citenamefont {King},
  \citenamefont {Valeev}, \citenamefont {Evangelista}, \citenamefont {Turney},
  \citenamefont {Crawford},\ and\ \citenamefont
  {Sherrill}}]{parrishPsi4OpenSourceElectronic2017}%
  \BibitemOpen
  \bibfield  {author} {\bibinfo {author} {\bibfnamefont {R.~M.}\ \bibnamefont
  {Parrish}}, \bibinfo {author} {\bibfnamefont {L.~A.}\ \bibnamefont {Burns}},
  \bibinfo {author} {\bibfnamefont {D.~G.~A.}\ \bibnamefont {Smith}}, \bibinfo
  {author} {\bibfnamefont {A.~C.}\ \bibnamefont {Simmonett}}, \bibinfo {author}
  {\bibfnamefont {A.~E.}\ \bibnamefont {DePrince}}, \bibinfo {author}
  {\bibfnamefont {E.~G.}\ \bibnamefont {Hohenstein}}, \bibinfo {author}
  {\bibfnamefont {U.}~\bibnamefont {Bozkaya}}, \bibinfo {author} {\bibfnamefont
  {A.~Y.}\ \bibnamefont {Sokolov}}, \bibinfo {author} {\bibfnamefont
  {R.}~\bibnamefont {Di~Remigio}}, \bibinfo {author} {\bibfnamefont {R.~M.}\
  \bibnamefont {Richard}}, \bibinfo {author} {\bibfnamefont {J.~F.}\
  \bibnamefont {Gonthier}}, \bibinfo {author} {\bibfnamefont {A.~M.}\
  \bibnamefont {James}}, \bibinfo {author} {\bibfnamefont {H.~R.}\ \bibnamefont
  {McAlexander}}, \bibinfo {author} {\bibfnamefont {A.}~\bibnamefont {Kumar}},
  \bibinfo {author} {\bibfnamefont {M.}~\bibnamefont {Saitow}}, \bibinfo
  {author} {\bibfnamefont {X.}~\bibnamefont {Wang}}, \bibinfo {author}
  {\bibfnamefont {B.~P.}\ \bibnamefont {Pritchard}}, \bibinfo {author}
  {\bibfnamefont {P.}~\bibnamefont {Verma}}, \bibinfo {author} {\bibfnamefont
  {H.~F.}\ \bibnamefont {Schaefer}}, \bibinfo {author} {\bibfnamefont
  {K.}~\bibnamefont {Patkowski}}, \bibinfo {author} {\bibfnamefont {R.~A.}\
  \bibnamefont {King}}, \bibinfo {author} {\bibfnamefont {E.~F.}\ \bibnamefont
  {Valeev}}, \bibinfo {author} {\bibfnamefont {F.~A.}\ \bibnamefont
  {Evangelista}}, \bibinfo {author} {\bibfnamefont {J.~M.}\ \bibnamefont
  {Turney}}, \bibinfo {author} {\bibfnamefont {T.~D.}\ \bibnamefont
  {Crawford}}, \ and\ \bibinfo {author} {\bibfnamefont {C.~D.}\ \bibnamefont
  {Sherrill}},\ }\href {\doibase 10.1021/acs.jctc.7b00174} {\bibfield
  {journal} {\bibinfo  {journal} {Journal of Chemical Theory and Computation}\
  }\textbf {\bibinfo {volume} {13}},\ \bibinfo {pages} {3185} (\bibinfo {year}
  {2017})}\BibitemShut {NoStop}%
\bibitem [{\citenamefont {McClean}\ \emph {et~al.}(2017)\citenamefont
  {McClean}, \citenamefont {Kivlichan}, \citenamefont {Sung}, \citenamefont
  {Steiger}, \citenamefont {Cao}, \citenamefont {Dai}, \citenamefont {Fried},
  \citenamefont {Gidney}, \citenamefont {Gimby}, \citenamefont {H{\"a}ner},
  \citenamefont {Hardikar}, \citenamefont {Havl{\'i}{\v c}ek}, \citenamefont
  {Huang}, \citenamefont {Jiang}, \citenamefont {Neeley}, \citenamefont
  {O'Brien}, \citenamefont {Ozfidan}, \citenamefont {Radin}, \citenamefont
  {Romero}, \citenamefont {Rubin}, \citenamefont {Sawaya}, \citenamefont
  {Setia}, \citenamefont {Sim}, \citenamefont {Steudtner}, \citenamefont {Sun},
  \citenamefont {Zhang},\ and\ \citenamefont
  {Babbush}}]{mccleanOpenFermionElectronicStructure2017}%
  \BibitemOpen
  \bibfield  {author} {\bibinfo {author} {\bibfnamefont {J.~R.}\ \bibnamefont
  {McClean}}, \bibinfo {author} {\bibfnamefont {I.~D.}\ \bibnamefont
  {Kivlichan}}, \bibinfo {author} {\bibfnamefont {K.~J.}\ \bibnamefont {Sung}},
  \bibinfo {author} {\bibfnamefont {D.~S.}\ \bibnamefont {Steiger}}, \bibinfo
  {author} {\bibfnamefont {Y.}~\bibnamefont {Cao}}, \bibinfo {author}
  {\bibfnamefont {C.}~\bibnamefont {Dai}}, \bibinfo {author} {\bibfnamefont
  {E.~S.}\ \bibnamefont {Fried}}, \bibinfo {author} {\bibfnamefont
  {C.}~\bibnamefont {Gidney}}, \bibinfo {author} {\bibfnamefont
  {B.}~\bibnamefont {Gimby}}, \bibinfo {author} {\bibfnamefont
  {T.}~\bibnamefont {H{\"a}ner}}, \bibinfo {author} {\bibfnamefont
  {T.}~\bibnamefont {Hardikar}}, \bibinfo {author} {\bibfnamefont
  {V.}~\bibnamefont {Havl{\'i}{\v c}ek}}, \bibinfo {author} {\bibfnamefont
  {C.}~\bibnamefont {Huang}}, \bibinfo {author} {\bibfnamefont
  {Z.}~\bibnamefont {Jiang}}, \bibinfo {author} {\bibfnamefont
  {M.}~\bibnamefont {Neeley}}, \bibinfo {author} {\bibfnamefont
  {T.}~\bibnamefont {O'Brien}}, \bibinfo {author} {\bibfnamefont
  {I.}~\bibnamefont {Ozfidan}}, \bibinfo {author} {\bibfnamefont {M.~D.}\
  \bibnamefont {Radin}}, \bibinfo {author} {\bibfnamefont {J.}~\bibnamefont
  {Romero}}, \bibinfo {author} {\bibfnamefont {N.}~\bibnamefont {Rubin}},
  \bibinfo {author} {\bibfnamefont {N.~P.~D.}\ \bibnamefont {Sawaya}}, \bibinfo
  {author} {\bibfnamefont {K.}~\bibnamefont {Setia}}, \bibinfo {author}
  {\bibfnamefont {S.}~\bibnamefont {Sim}}, \bibinfo {author} {\bibfnamefont
  {M.}~\bibnamefont {Steudtner}}, \bibinfo {author} {\bibfnamefont
  {W.}~\bibnamefont {Sun}}, \bibinfo {author} {\bibfnamefont {F.}~\bibnamefont
  {Zhang}}, \ and\ \bibinfo {author} {\bibfnamefont {R.}~\bibnamefont
  {Babbush}},\ }\href@noop {} {\bibfield  {journal} {\bibinfo  {journal}
  {arXiv:1710.07629 [physics, physics:quant-ph]}\ } (\bibinfo {year} {2017})},\
  \Eprint {http://arxiv.org/abs/1710.07629} {arXiv:1710.07629} \BibitemShut
  {NoStop}%
\bibitem [{\citenamefont {Hagberg}\ \emph {et~al.}(2008)\citenamefont
  {Hagberg}, \citenamefont {Schult},\ and\ \citenamefont
  {Swart}}]{SciPyProceedings_11}%
  \BibitemOpen
  \bibfield  {author} {\bibinfo {author} {\bibfnamefont {A.~A.}\ \bibnamefont
  {Hagberg}}, \bibinfo {author} {\bibfnamefont {D.~A.}\ \bibnamefont {Schult}},
  \ and\ \bibinfo {author} {\bibfnamefont {P.~J.}\ \bibnamefont {Swart}},\ }in\
  \href@noop {} {\emph {\bibinfo {booktitle} {Proceedings of the 7th Python in
  Science Conference}}},\ \bibinfo {editor} {edited by\ \bibinfo {editor}
  {\bibfnamefont {G.}~\bibnamefont {Varoquaux}}, \bibinfo {editor}
  {\bibfnamefont {T.}~\bibnamefont {Vaught}}, \ and\ \bibinfo {editor}
  {\bibfnamefont {J.}~\bibnamefont {Millman}}}\ (\bibinfo {address} {Pasadena,
  CA USA},\ \bibinfo {year} {2008})\ pp.\ \bibinfo {pages} {11 --
  15}\BibitemShut {NoStop}%
\bibitem [{\citenamefont {Whitfield}\ \emph {et~al.}(2011)\citenamefont
  {Whitfield}, \citenamefont {Biamonte},\ and\ \citenamefont
  {Aspuru-Guzik}}]{whitfieldSimulation2011}%
  \BibitemOpen
  \bibfield  {author} {\bibinfo {author} {\bibfnamefont {J.~D.}\ \bibnamefont
  {Whitfield}}, \bibinfo {author} {\bibfnamefont {J.}~\bibnamefont {Biamonte}},
  \ and\ \bibinfo {author} {\bibfnamefont {A.}~\bibnamefont {Aspuru-Guzik}},\
  }\href {\doibase 10.1080/00268976.2011.552441} {\bibfield  {journal}
  {\bibinfo  {journal} {Molecular Physics}\ }\textbf {\bibinfo {volume}
  {109}},\ \bibinfo {pages} {735} (\bibinfo {year} {2011})},\ \Eprint
  {http://arxiv.org/abs/https://doi.org/10.1080/00268976.2011.552441}
  {https://doi.org/10.1080/00268976.2011.552441} \BibitemShut {NoStop}%
\bibitem [{\citenamefont {Hastings}\ \emph {et~al.}(2015)\citenamefont
  {Hastings}, \citenamefont {Wecker}, \citenamefont {Bauer},\ and\
  \citenamefont {Troyer}}]{hastingsImprovingQuantumAlgorithms2015}%
  \BibitemOpen
  \bibfield  {author} {\bibinfo {author} {\bibfnamefont {M.~B.}\ \bibnamefont
  {Hastings}}, \bibinfo {author} {\bibfnamefont {D.}~\bibnamefont {Wecker}},
  \bibinfo {author} {\bibfnamefont {B.}~\bibnamefont {Bauer}}, \ and\ \bibinfo
  {author} {\bibfnamefont {M.}~\bibnamefont {Troyer}},\ }\href@noop {}
  {\bibfield  {journal} {\bibinfo  {journal} {Quantum Information \&
  Computation}\ }\textbf {\bibinfo {volume} {15}},\ \bibinfo {pages} {1}
  (\bibinfo {year} {2015})}\BibitemShut {NoStop}%
\bibitem [{\citenamefont {Ryabinkin}\ \emph {et~al.}(2018)\citenamefont
  {Ryabinkin}, \citenamefont {Yen}, \citenamefont {Genin},\ and\ \citenamefont
  {Izmaylov}}]{ryabinkin2018qubit}%
  \BibitemOpen
  \bibfield  {author} {\bibinfo {author} {\bibfnamefont {I.~G.}\ \bibnamefont
  {Ryabinkin}}, \bibinfo {author} {\bibfnamefont {T.-C.}\ \bibnamefont {Yen}},
  \bibinfo {author} {\bibfnamefont {S.~N.}\ \bibnamefont {Genin}}, \ and\
  \bibinfo {author} {\bibfnamefont {A.~F.}\ \bibnamefont {Izmaylov}},\
  }\href@noop {} {\bibfield  {journal} {\bibinfo  {journal} {Journal of
  chemical theory and computation}\ }\textbf {\bibinfo {volume} {14}},\
  \bibinfo {pages} {6317} (\bibinfo {year} {2018})}\BibitemShut {NoStop}%
\end{thebibliography}%

\appendix

\section{Calculational details}

In this section we give some derivations of the algebraic results used in the text.

\subsection{Computation of $ \mathcal{X}$ for the ALCU method}\label{app:comp_X}

We now derive the results that follow \cref{axis_b}. The operator $ \mathcal{X}$ is given by
\be
\begin{split}
 \mathcal{X} &= \frac{i}{2}\left[ H_{n-1}, P_n\right]\\
&=\frac{i}{2}\sum_{k=1}^{n-1}\beta_k\left[ P_k, P_n\right]\\
&=i\sum_{k=1}^{n-1}\beta_k P_k P_n,
\end{split}
\ee
where we wrote $ H_{n-1}=\sum_{k=1}^{n-1}\beta_k P_k$ with $\sum_{k=1}^{n-1}\beta_k^2=1$. Then we can compute:
\be
\begin{split}
 \mathcal{X}^2 &= -\sum_{k=1}^{n-1}\sum_{j=1}^{n-1}\beta_k\beta_j P_k P_n P_j P_n\\
&= -\sum_{k=1}^{n-1}\beta_k^2 P_k P_n P_k P_n-\sum_{k<j}^{n-1}\beta_k\beta_j\{ P_k P_n, P_j P_n\}\\
&= -\sum_{k=1}^{n-1}\beta_k^2 P_k P_n P_k P_n\\
&= \sum_{k=1}^{n-1}\beta_k^2 P_k^2 P_n^2\\
&= {\openone}.
\end{split}
\ee

Now consider the commutator of $ \mathcal{X}$ and $ H_n$. We can use $ \mathcal{X}=i H_{n-1} P_n$ to write
\begin{equation}
\begin{split}
 \mathcal{X}  H_n &= i H_{n-1} P_n  H_n    \\
&=i H_{n-1} P_n (\sin\phi_{n-1} H_{n-1}+\cos\phi_{n-1} P_n)\\
&=i(\sin\phi_{n-1} H_{n-1} P_n  H_{n-1}+\cos\phi_{n-1} H_{n-1} P_n^2).\\
\end{split}
\end{equation}
Using $\{ H_{n-1}, P_n\}=0$ and $ P_n^2= 1$ we have
\begin{equation}
\begin{split}
 \mathcal{X}  H_n &=i(-\sin\phi_{n-1} P_n +\cos\phi_{n-1} H_{n-1}),\\
\end{split}
\end{equation}
so that
\begin{equation}
[ \mathcal{X},  H_n] =2i(-\sin\phi_{n-1} P_n +\cos\phi_{n-1} H_{n-1}).
\end{equation}

This enables us to compute the adjoint action generated by $ \mathcal{X}$ on $H_n$. Using the identity [for any operators $ A$ and $B$, where $ A^2= {\openone}$ so that $e^{-i(\alpha/2)  A}=\cos(\alpha/2)  {\openone} -i\sin(\alpha/2)  A$]
\begin{equation}
\begin{split}
e^{-i(\alpha/2)  A} B e^{i(\alpha/2)  A} %&= (\cos\alpha/2  1 -i\sin\alpha/2  X) Y (\cos\alpha/2  1 +i\sin\alpha/2  X)\\
&= \cos^2(\alpha/2)  B +\sin^2(\alpha/2)  A B A\\ &+i\sin(\alpha/2)\cos(\alpha/2) [ A, B],\\
\end{split}
\end{equation}
we have ($ R = e^{-i(\alpha/2) \mathcal{X}} $)
\begin{equation}
\begin{split}
%e^{-i\alpha/2  X} H_{n}e^{i\alpha/2  X} 
R H_{n} R^\dagger &= \cos^2(\alpha/2)  H_{n} +\sin^2(\alpha/2)  \mathcal{X} H_{n} \mathcal{X}\\ 
&+i\sin(\alpha/2)\cos(\alpha/2) [ \mathcal{X}, H_{n}]\\
&= (\cos^2\alpha/2 -\sin^2\alpha/2)  H_{n}\\ &+(i/2)2\sin(\alpha/2)\cos(\alpha/2) [ \mathcal{X}, H_{n}]\\
&= \cos\alpha   H_{n}\\ &-\sin\alpha(-\sin\phi_{n-1} P_n +\cos\phi_{n-1} H_{n-1})\\
&= \cos\alpha  (\cos\phi_{n-1} P_n +\sin\phi_{n-1} H_{n-1})\\ 
&-\sin\alpha(-\sin\phi_{n-1} P_n +\cos\phi_{n-1} H_{n-1})\\
&= (\cos\alpha\cos\phi_{n-1}+\sin\alpha\sin\phi_{n-1}) P_n\\
& +(\cos\alpha\sin\phi_{n-1}-\sin\alpha\cos\phi_{n-1}) H_{n-1}\\ 
&= \cos(\phi_{n-1}-\alpha) P_n\\
& +\sin(\phi_{n-1}-\alpha) H_{n-1}.\\ 
\end{split}
\end{equation}
Choosing $\alpha=\phi_{n-1}$ gives $ R  H_n R^\dagger =  P_n$. Given this role for $ R$, which is generated by $ \mathcal{X}$, we wish to know the commutation relations among the terms of $ \mathcal{X}$. Because $ \mathcal{X}=2i  P_n H_{n-1}$, the terms of $\mathcal{X}$ have the form $2i  P_n P_j$ for $j<n$. The commutation relations between any pair of terms are
\begin{equation}
\begin{split}
[ P_n P_j, P_n P_k] &=   P_n P_j P_n P_k- P_n P_k P_n P_j   \\
&=  -( P_n P_n P_j P_k- P_n P_n P_k P_j)   \\
&= -[ P_j, P_k]\\
&=2 P_k P_j.\\
\end{split}
\end{equation}

\subsection{Electronic structure Hamiltonian using Majorana operators}\label{sec:H_e_majorana_appendix}

Here we derive the form of the Hamiltonian given in Eq.~\eqref{eq:H_e_majorana}. Since the single-mode Majorana operators are linear combinations of the fermionic ladder operators, we have the identities
\begin{equation}\label{eq:fermion_from_majorana}
		a_p = \frac{\gamma_{2p} + i\gamma_{2p+1}}{2}, \ a_p^\dagger = \frac{ \gamma_{2p} - i\gamma_{2p+1} }{2}.
\end{equation}
Furthermore, recall the permutational symmetries in the coefficients, given by \cref{eq:hpq_sym,eq:hpqrs_sym_1,eq:hpqrs_sym_2}, and the anticommutation relation for arbitrary Majorana operators, Eq.~\eqref{eq:majorana_gen_AR}. These are the only properties we use, but they allow for considerable simplification to the structure of the Hamiltonian terms. For brevity, we shall make use of such properties freely and often without comment.

First, consider the one-body terms, which are quadratic in fermionic operators. Using Majorana operators, they become
\begin{equation}
\begin{split}
\sum_{p,q} h_{pq} a_p^\dagger a_q = \frac{1}{4} \sum_{p,q} &h_{pq} ( \gamma_{2p}\gamma_{2q} + \gamma_{2p+1}\gamma_{2q+1} \\
&+ i\gamma_{2p}\gamma_{2q+1} - i\gamma_{2p+1}\gamma_{2q} ).
\end{split}
\end{equation}
This expression can be simplified by separating the summation into diagonal and off-diagonal terms, a technique which we employ heavily throughout this derivation. The sum over the $ \gamma_{2p}\gamma_{2q} $ and $ \gamma_{2p+1}\gamma_{2q+1} $ terms simply yields a multiple of the identity:
\begin{align}
\sum_{p,q} &h_{pq} ( \gamma_{2p}\gamma_{2q} + \gamma_{2p+1}\gamma_{2q+1} ) \notag \\
%& = \left( \sum_{\substack{p,q\\p=q}} + \sum_{\substack{p,q\\p \neq q}} \right) h_{pq} ( \gamma_{2p}\gamma_{2q} + \gamma_{2p+1}\gamma_{2q+1} ) \notag \\
&= \sum_{p} h_{pp} \left( \gamma_{2p}^2 + \gamma_{2p+1}^2 \right) \notag \\
&\quad + \sum_{\substack{p,q\\p<q}} h_{pq} \left( \{ \gamma_{2p},\gamma_{2q} \} + \{ \gamma_{2p+1},\gamma_{2q+1} \} \right) \notag \\
&= 2 \sum_{p}h_{pp} {\openone}.
\end{align}
%Note the abuse of notation with the grouped summation symbols.
The remaining terms simplify but do not cancel or reduce in order:~by relabeling the indices (another trick which we make frequent use of), we see that $ \sum_{p,q} h_{pq} i \gamma_{2p}\gamma_{2q+1} = \sum_{p,q} h_{pq} i\gamma_{2q}\gamma_{2p+1} $, hence
\begin{equation}\label{eq:1b_terms}
\sum_{p,q} h_{pq} a_p^\dagger a_q = \frac{1}{2} \left(\sum_{p} h_{pp} {\openone} + \sum_{p,q} h_{pq} i\gamma_{2p}\gamma_{2q+1} \right).
\end{equation}

Next, we consider the two-body interaction terms, which feature the quartic order operators. Any such term is written as a linear combination of 16 Majorana operators. To do so, define
\begin{equation}
\Gamma_{pqrs}^{\mathbf{x}} = i^{|\mathbf{x}|} (-1)^{x_1+x_2} \gamma_{2p+x_1} \gamma_{2q+x_2} \gamma_{2r+x_3} \gamma_{2s+x_4},
\end{equation}
where $ \mathbf{x} = x_1x_2x_3x_4 \in \{ 0,1 \}^4 $ is a binary string encoding the parity of each index and $ |\mathbf{x}| $ is its Hamming weight. Then, from Eq.~\eqref{eq:fermion_from_majorana}, a straightforward algebraic expansion gives the following expression for each two-body term:
\begin{equation}
a_p^\dagger a_q^\dagger a_r a_s = \frac{1}{16} \sum_{\mathbf{x} \in \{ 0,1 \}^4} \Gamma_{pqrs}^{\mathbf{x}}. 
\end{equation}

Consider the set $ B_1 = \{ 0011,1100,0101,1010 \} $. These strings correspond to the quartic Majorana operators appearing in Eq.~\eqref{eq:H_e_majorana}, and as we will see, they are the only such terms which do not vanish. Also, note that since $ a_j^2 = (a_j^\dagger)^2 = 0 $, we impose the trivial constraints in the summations that $ p \neq q $ and $ r \neq s $. Specifying these conditions explicitly will be useful once we relabel the indices. We rewrite these terms as
\begin{equation}
\begin{split}
\sum_{p,q,r,s} h_{pqrs} \Gamma_{pqrs}^{1100} &= - \sum_{\substack{p,q,r,s\\p\neq q;r\neq s}} h_{pqrs} \gamma_{2p+1}\gamma_{2q+1}\gamma_{2r}\gamma_{2s} \\
&= - \sum_{\substack{p,q,r,s\\p\neq q;r\neq s}} h_{pqrs} \gamma_{2r}\gamma_{2s}\gamma_{2p+1}\gamma_{2q+1} \\
&= - \sum_{\substack{p,q,r,s\\p\neq q;r\neq s}} h_{pqrs} \gamma_{2p}\gamma_{2q}\gamma_{2r+1}\gamma_{2s+1},
\end{split}
\end{equation}
and, for $ x,y \in \{0,1\} $ such that $ x \neq y $,
\begin{equation}
\begin{split}
\sum_{p,q,r,s} h_{pqrs} \Gamma_{pqrs}^{xyxy} &= \sum_{\substack{p,q,r,s\\p\neq q;r\neq s}} h_{pqrs} \gamma_{2p+x}\gamma_{2q+y}\gamma_{2r+x}\gamma_{2s+y} \\
&= - \sum_{\substack{p,q,r,s\\p\neq q;r\neq s}} h_{pqrs} \gamma_{2p+x}\gamma_{2r+x}\gamma_{2q+y}\gamma_{2s+y} \\
&= - \sum_{\substack{p,q,r,s\\p\neq r;q\neq s}} h_{pqrs} \gamma_{2p+x}\gamma_{2q+x}\gamma_{2r+y}\gamma_{2s+y}.
\end{split}
\end{equation}
Thus we obtain
\begin{align}
&\sum_{p,q,r,s} h_{pqrs} \left( \sum_{\mathbf{x} \in B_1} \Gamma_{pqrs}^{\mathbf{x}} \right) \label{eq:B1} \\
&\quad = -2 \left( \sum_{\substack{p,q,r,s\\p\neq q;r\neq s}} + \sum_{\substack{p,q,r,s\\p\neq r;q\neq s}} \right) h_{pqrs} \gamma_{2p}\gamma_{2q}\gamma_{2r+1}\gamma_{2s+1}. \notag
\end{align}
Since we would like to completely separate the quadratic terms from the quartic terms, we observe that if $ p=q $ or $ r=s $ in the above expression, then those terms reduce to quadratic order (or the identity, if both equalities hold). The first summation automatically excludes such reduction, so we analyze the second one, again separating the diagonal and off-diagonal summands with respect to each pair $ (p,q) $ and $ (r,s) $:
\begin{align}
&\sum_{\substack{p,q,r,s\\p\neq r;q\neq s}} h_{pqrs} \gamma_{2p}\gamma_{2q}\gamma_{2r+1}\gamma_{2s+1} \notag \\
&= \sum_{\substack{p,q,r,s\\p\neq r;q\neq s\\p\neq q;r\neq s}} h_{pqrs} \gamma_{2p}\gamma_{2q}\gamma_{2r+1}\gamma_{2s+1} + \sum_{\substack{p,r\\p\neq r}} h_{pprr} {\openone} \notag \\
&\quad + \sum_{\substack{p,q,r\\p\neq r;q\neq r\\p\neq q}} h_{pqrr} \gamma_{2p}\gamma_{2q} + \!\! \sum_{\substack{p,r,s\\p\neq r;p\neq s\\r\neq s}} h_{pprs} \gamma_{2r+1}\gamma_{2s+1} \notag \\
&= \sum_{\substack{p,q,r,s\\p\neq r;q\neq s\\p\neq q;r\neq s}} h_{pqrs} \gamma_{2p}\gamma_{2q}\gamma_{2r+1}\gamma_{2s+1} + \sum_{\substack{p,r\\p\neq r}} h_{pprr} {\openone} \notag \\
&\quad + \sum_{\substack{p,q,r\\p\neq r;q\neq r\\p<q}} h_{pqrr} \{ \gamma_{2p},\gamma_{2q} \} + \!\!  \sum_{\substack{p,r,s\\p\neq r;p\neq s\\r<s}} h_{pprs} \{ \gamma_{2r+1},\gamma_{2s+1} \} \notag \\
&= \sum_{\substack{p,q,r,s\\p\neq r;q\neq s\\p\neq q;r\neq s}} h_{pqrs} \gamma_{2p}\gamma_{2q}\gamma_{2r+1}\gamma_{2s+1} + \sum_{\substack{p,q\\p\neq q}} h_{ppqq} {\openone}. \label{eq:B1_simp}
\end{align}
So we see that these quadratic terms in fact vanish due to anticommutation.

Now we show that the remaining 12 cases yield the same operators as those already obtained in Eq.~\eqref{eq:1b_terms}. Let $ B_2 = \{ 0000,0110,1001,1111 \} $ and $ x,y \in \{ 0,1 \} $:
\begin{equation}
\begin{split}
\sum_{p,q,r,s} h_{pqrs} \Gamma_{pqrs}^{xyyx} &= \sum_{\substack{p,q,r,s\\p\neq q;r\neq s}} h_{pqrs} \gamma_{2p+x}\gamma_{2q+y}\gamma_{2r+y}\gamma_{2s+x} \\
&= \sum_{\substack{p,q,r,s\\p\neq q;r\neq s\\p\neq s}} h_{pqrs} \gamma_{2p+x}\gamma_{2q+y}\gamma_{2r+y}\gamma_{2s+x} \\
&\quad + \sum_{\substack{p,q,r\\p\neq q;r\neq p}} h_{pqrp} \gamma_{2q+y}\gamma_{2r+y}.
\end{split}
\end{equation}
The second sum simplifies to
\begin{align}
\sum_{\substack{p,q,r\\p\neq q;r\neq p}} h_{pqrp} \gamma_{2q+y}\gamma_{2r+y} &= \sum_{\substack{p,q\\p\neq q;r\neq p\\q<r}} h_{pqrp} \{ \gamma_{2q+y},\gamma_{2r+y} \} \notag \\
&\quad + \sum_{\substack{p,q\\p\neq q}} h_{pqqp} {\openone} \notag \\
&= \sum_{\substack{p,q\\p\neq q}} h_{pqqp} {\openone}.
\end{align}
The first sum depends on whether $ x $ and $ y $ are the same or not. If $ x \neq y $, then
\begin{align}
\sum_{\substack{p,q,r,s\\p\neq q;r\neq s\\p\neq s}} &h_{pqrs} \gamma_{2p+x}\gamma_{2q+y}\gamma_{2r+y}\gamma_{2s+x} \notag \\
&= \sum_{\substack{p,q,r,s\\p\neq q;r\neq s\\p<s}} h_{pqrs} ( \gamma_{2p+x}\gamma_{2q+y}\gamma_{2r+y}\gamma_{2s+x} \notag \\
&\qquad\qquad\ \ + \gamma_{2s+x}\gamma_{2q+y}\gamma_{2r+y}\gamma_{2p+x} ) \notag\\
&= \sum_{\substack{p,q,r,s\\p\neq q;r\neq s\\p<s}} h_{pqrs} ( \gamma_{2p+x}\gamma_{2q+y}\gamma_{2r+y}\gamma_{2s+x} \notag \\
&\qquad\qquad\ \ - \gamma_{2p+x}\gamma_{2q+y}\gamma_{2r+y}\gamma_{2s+x} ) \notag \\
&= 0. \label{eq:quartic_cancellation}
\end{align}
If $ x = y $, we first observe that if $ p \neq r $ and $ q \neq s $, then the sum vanishes, as demonstrated above. Therefore we have the three remaining cases ($ p \neq r $ and $ q = s $, $ p = r $ and $ q \neq s $, and $ p = r $ and $ q = s $):
\begin{align}
\sum_{\substack{p,q,r,s\\p\neq q;r\neq s\\p\neq s}} &h_{pqrs} \gamma_{2p+x}\gamma_{2q+x}\gamma_{2r+x}\gamma_{2s+x} \notag \\
&= - \sum_{\substack{p,q\\p\neq q}} h_{pqpq} {\openone} - \sum_{\substack{p,q,r\\p\neq q;r\neq q\\p\neq r}} h_{pqrq} \gamma_{2p+x}\gamma_{2r+x} \notag \\
&\quad - \sum_{\substack{p,q,s\\p\neq q;p\neq s\\q\neq s}} h_{pqps} \gamma_{2q+x}\gamma_{2s+x} \notag \\
&= - \sum_{\substack{p,q\\p\neq q}} h_{pqpq} {\openone} - \sum_{\substack{p,q,r\\p\neq q;r\neq q\\p<r}} h_{pqrq} \{ \gamma_{2p+x},\gamma_{2r+x} \} \notag \\
&\quad - \sum_{\substack{p,q,s\\p\neq q;p\neq s\\q<s}} h_{pqps} \{ \gamma_{2q+x},\gamma_{2s+x} \} \notag \\
&= - \sum_{\substack{p,q\\p\neq q}} h_{pqpq} {\openone}.
\end{align}
Altogether, the terms corresponding to $ B_2 $ are just the identity operator:
\begin{equation}\label{eq:B2}
\sum_{p,q,r,s} h_{pqrs} \left( \sum_{\mathbf{x} \in B_2} \Gamma_{pqrs}^{\mathbf{x}} \right) = \sum_{\substack{p,q\\p\neq q}} \left( 4h_{pqqp} - 2h_{pqpq} \right) {\openone}.
\end{equation}
Let $ B_3 = \{ 0010, 0100, 1011, 1101 \} $. These strings give rise to the same terms, since for $ x \in \{0,1\} $,
\begin{equation}
\begin{split}
\sum_{p,q,r,s} \Gamma_{pqrs}^{x01x} &= \sum_{p,q,r,s} h_{pqrs} i\gamma_{2p+x}\gamma_{2q}\gamma_{2r+1}\gamma_{2s+x} \\
&= - \sum_{p,q,r,s} h_{pqrs} i\gamma_{2p+x}\gamma_{2q+1}\gamma_{2r}\gamma_{2s+x} \\
&= \sum_{p,q,r,s} \Gamma_{pqrs}^{x10x}. \label{eq:B3_equiv}
\end{split}
\end{equation}
We simplify the sum using the same type of manipulations as in Eq.~\eqref{eq:quartic_cancellation}:
\begin{align}
&\sum_{p,q,r,s} h_{pqrs} i\gamma_{2p}\gamma_{2q}\gamma_{2r+1}\gamma_{2s} \notag \\
&= \sum_{\substack{p,q,r,s\\p\neq q;r\neq s\\p<s;s\neq q}} h_{pqrs} i( \gamma_{2p}\gamma_{2q}\gamma_{2r+1}\gamma_{2s} - \gamma_{2p}\gamma_{2q}\gamma_{2r+1}\gamma_{2s} ) \notag \\
&\quad - \sum_{\substack{p,q,r\\p\neq q;r\neq q}} h_{pqrq} i\gamma_{2p}\gamma_{2r+1} + \sum_{\substack{p,q,r\\p\neq q;r\neq p}} h_{pqrp} i\gamma_{2q}\gamma_{2r+1} \notag \\
&= \sum_{\substack{p,q,r\\p\neq r;q\neq r}} ( h_{prrq} - h_{prqr} ) i\gamma_{2p}\gamma_{2q+1}.
\end{align}
Thus we obtain
\begin{equation}\label{eq:B3}
\begin{split}
\sum_{p,q,r,s} &h_{pqrs} \left( \sum_{\mathbf{x} \in B_3} \Gamma_{pqrs}^{\mathbf{x}} \right) \\
&= 4 \sum_{\substack{p,q,r\\p\neq r;q\neq r}} ( h_{prrq} - h_{pqrr} ) i\gamma_{2p}\gamma_{2q+1}.
\end{split}
\end{equation}
The last set is $ B_4 = \{ 0001,0111,1000,1110 \} $. Again, all four strings correspond to the same terms. We show this by evaluating, for $ w,x,y \in \{0,1\} $ with $ w \neq y $,
\begin{align}
\sum_{p,q,r,s} &\Gamma_{pqrs}^{wxxy} = (-1)^{w+1} \!\! \sum_{p,q,r,s} h_{pqrs} i\gamma_{2p+w}\gamma_{2q+x}\gamma_{2r+x}\gamma_{2s+y} \notag \\
&= (-1)^{w+1} \!\! \sum_{\substack{p,q,r,s\\p\neq q;r\neq s\\q<r}} h_{pqrs} i( \gamma_{2p+w}\gamma_{2q+x}\gamma_{2r+x}\gamma_{2s+y} \notag \\
&\qquad\qquad\qquad\qquad\ \ - \gamma_{2p+w}\gamma_{2q+x}\gamma_{2r+x}\gamma_{2s+y} ) \notag \\
&\quad + (-1)^{w+1} \!\! \sum_{\substack{p,q,s\\p\neq q;q\neq s}} h_{pqqs} i\gamma_{2p+w}\gamma_{2s+y} \notag \\
&= (-1)^{w+1} \!\! \sum_{\substack{p,q,r\\p\neq r;r\neq q}} h_{prrq} i\gamma_{2p+w}\gamma_{2q+y}.
\end{align}
If we order the Majorana product such that the even index appears first, then the sign of $ (-1)^{w+1} $ cancels with that of swapping $ \gamma_{2p+w} $ with $ \gamma_{2q+y} $, and so we have
\begin{equation}\label{eq:B4}
\sum_{p,q,r,s} h_{pqrs} \left( \sum_{\mathbf{x} \in B_4} \Gamma_{pqrs}^{\mathbf{x}} \right) = 4 \sum_{\substack{p,q,r\\p\neq r;q\neq r}} h_{prrq} i\gamma_{2p}\gamma_{2q+1}.
\end{equation}
Finally, we collect all the terms from \cref{eq:B1,eq:B2,eq:B3,eq:B4}, along with the slight simplification in Eq.~\eqref{eq:B1_simp}, to write the two-body terms as
\begin{equation}
\begin{split}
\frac{1}{2} &\sum_{p,q,r,s} h_{pqrs} a_p^\dagger a_q^\dagger a_r a_s = \frac{1}{32} \sum_{p,q,r,s} h_{pqrs} \left( \sum_{\mathbf{x} \in \{0,1\}^4} \Gamma_{pqrs}^{\mathbf{x}} \right) \\
&= \frac{1}{8} \sum_{\substack{p,q\\p\neq q}} \left( h_{pqqp} - h_{pqpq} \right) {\openone} \\
&\quad + \frac{1}{8} \sum_{\substack{p,q,r\\p\neq r;q\neq r}} ( 2h_{prrq} - h_{pqrr} ) i\gamma_{2p}\gamma_{2q+1} \\
&\quad - \frac{1}{16} \left( \sum_{\substack{p,q,r,s\\p\neq q;r\neq s}} + \sum_{\substack{p,q,r,s\\p\neq q;r\neq s\\p\neq r;q\neq s}} \right) h_{pqrs} \gamma_{2p}\gamma_{2q}\gamma_{2r+1}\gamma_{2s+1}.
\end{split}
\end{equation}
Including the one-body terms, Eq.~\eqref{eq:1b_terms}, we express the full electronic structure Hamiltonian in terms of Majorana operators:
\begin{equation}
\begin{split}
H &= \frac{1}{2} \left[\sum_{p} h_{pp} + \frac{1}{4} \sum_{\substack{p,q\\p\neq q}} \left( h_{pqqp} - h_{pqpq} \right) \right] {\openone}  \\
&\quad + \sum_{p,q} \left[ \frac{1}{2} h_{pq} + \!\!\!\! \sum_{\substack{r\\p\neq r;q\neq r}} \!\! \left( \frac{1}{4}h_{prrq} - \frac{1}{8}h_{pqrr} \right) \right] i\gamma_{2p}\gamma_{2q+1} \\
&\quad - \frac{1}{16} \left( \sum_{\substack{p,q,r,s\\p\neq q;r\neq s}} + \sum_{\substack{p,q,r,s\\p\neq q;r\neq s\\p\neq r;q\neq s}} \right) h_{pqrs} \gamma_{2p}\gamma_{2q}\gamma_{2r+1}\gamma_{2s+1}.
\end{split}
\end{equation}
Defining new coefficients as
\begin{equation}\label{eq:new_coeff}
\begin{split}
\tilde{h} &= \frac{1}{2} \sum_{p} h_{pp} + \frac{1}{8} \sum_{\substack{p,q\\p\neq q}} \left( h_{pqqp} - h_{pqpq} \right) , \\
\tilde{h}_{pq} &= \frac{1}{2} h_{pq} + \sum_{\substack{r\\p\neq r;q\neq r}} \left( \frac{1}{4}h_{prrq} - \frac{1}{8}h_{pqrr} \right), \\
\tilde{h}_{pqrs} &= - \frac{1}{8} [ 1 + (1 - \delta_{pr})(1 - \delta_{qs}) ] h_{pqrs}.
\end{split}
\end{equation}
we obtain the Hamiltonian as presented in the main text, Eq.~\eqref{eq:H_e_majorana}.

\section{Proof details for Theorem~\ref{thm:cubicPartitionTheorem}}\label{sec:majorana_proof_details}

To see why \cref{eq:disjointness,eq:covering} hold, we first examine the structure of our anticommuting partition $ \{ S_{(q,r,s)} \} $. Although we have the choice of matching either one or three indices in each term's support, here we only use the condition of three matches. This amounts to matching exactly one even index, since the other two must be odd (or vice versa, by symmetry). In this sense, the problem reduces to finding an anticommuting partition of the set of all quadratic Majorana operators with only even indices in their support. Taking products with the set of all quadratic operators with only odd indices in their support then generates all the relevant quartic operators, $ \mathcal{M} $ (up to phase factors).

One may readily check from the definition of $ S_{(q,r,s)} $ that they do indeed cover $ \mathcal{M} $ and are all pairwise disjoint. However, since we have reduced the problem to considering simply quadratic operators, we may provide a visual argument which clearly demonstrates the partitioning scheme, Figure~\ref{fig:quadPartitions}. Note that for $ N = 2 $, there is only one unique quartic term, and for $ N = 3 $, all the even quadratics already anticommute (i.e., the red bin in the figure). From the figure, we immediately see the disjointness property satisfied, with each set of common index $ 2q $ having size $ q $. The exception, again, is the red bin, which corresponds to the union $ S_{(1,r,s)} \cup S_{(2,r,s)} $ as mentioned in the main text. Hence there are $ N-2 $ anticommuting sets of even-index quadratic operators, and taking products with all $ \binom{N}{2} $ odd-index quadratic operators yields the desired $ O(N^3) $ result.

\begin{figure}
    \centering
    \includegraphics{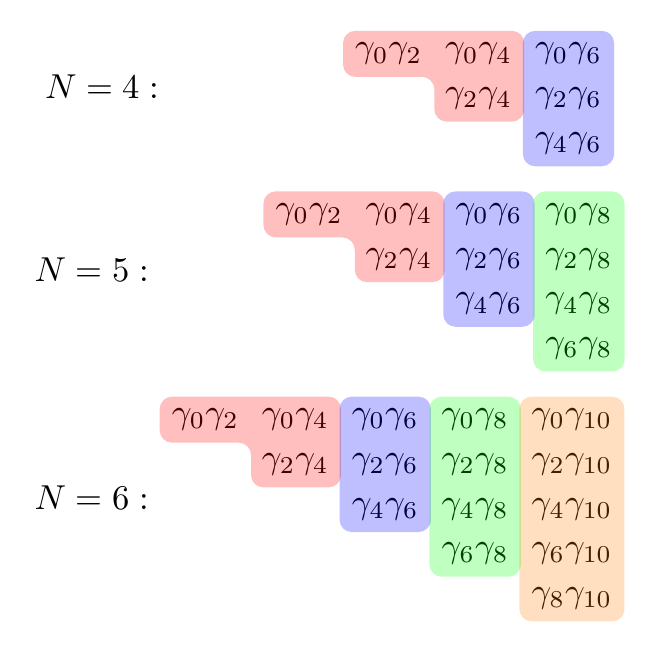}
    \caption{Partitioning of electronic structure terms. Finding an anticommuting partition of the quartic terms can be reduced to finding an anticommuting partition of quadratic terms with exclusively even (equiv.~odd) indices. Each highlighted bin is such an anticommuting set. Excluding the red bin, each set shares one common index $ 2q $ for $ 3 \leq q \leq N-1 $. Although only 3 values of $ N $ are depicted, the induction of this diagram is straightforward for arbitrary $ N $. One thus obtains $ N - 3 $ bins of size $ q $ each and $ 1 $ ``red bin'' of size $ 3 $.}
    \label{fig:quadPartitions}
\end{figure}

%\newpage
\section{Electronic structure systems}
\label{appendix:systems}

Table~\ref{tab:appendixSystems} details the systems from which the electronic structure Hamiltonians studied in Sec.~\ref{subsec:pauliLevelColouring} were generated.

\begin{table}[h]
\centering
\caption{The systems examined in our numerical analysis.  Geometries were obtained from the NIST CCBDB database~\cite{johnsoniiiNISTComputationalChemistry2016}, and molecular orbital integrals in the Hartee--Fock basis obtained from Psi4~\cite{parrishPsi4OpenSourceElectronic2017} and OpenFermion~\cite{mccleanOpenFermionElectronicStructure2017}.}\label{tab:appendixSystems} %\null\hfill
\begin{tabular}[t]{lrrlr}
\toprule
System & Charge & Multiplicity & Basis & Qubits \\
\midrule 
\ce{Ar1} & 0 & 1 & STO-3G & 18 \\
\ce{B1} & 0 & 2 & STO-3G & 10 \\
\ce{Be1} & 0 & 1 & STO-3G & 10 \\
\ce{Br1} & 0 & 2 & STO-3G & 36 \\
\ce{C1O1} & 0 & 1 & STO-3G & 20 \\
\ce{C1O2} & 0 & 1 & STO-3G & 30 \\
\ce{C1} & 0 & 3 & STO-3G & 10 \\
\ce{Cl1} & 0 & 2 & STO-3G & 18 \\
\ce{Cl1} & -1 & 1 & STO-3G & 18 \\
\ce{F1} & 0 & 2 & STO-3G & 10 \\
\ce{F2} & 0 & 1 & STO-3G & 20 \\
\ce{H1Cl1} & 0 & 1 & STO-3G & 20 \\
\ce{H1F1} & 0 & 1 & 3-21G & 22 \\
\ce{H1F1} & 0 & 1 & STO-3G & 12 \\
\ce{H1He1} & 0 & 1 & 3-21G & 8 \\
\ce{H1He1} & 0 & 1 & 6-311G** & 24 \\
\ce{H1He1} & 0 & 1 & 6-311G & 12 \\
\ce{H1He1} & 0 & 1 & 6-31G** & 20 \\
\ce{H1He1} & 0 & 1 & 6-31G & 8 \\
\ce{H1He1} & 0 & 1 & STO-3G & 4 \\
\ce{H1Li1O1} & 0 & 1 & STO-3G & 22 \\
\ce{H1Li1} & 0 & 1 & 3-21G & 22 \\
\ce{H1Li1} & 0 & 1 & STO-3G & 12 \\
\ce{H1Na1} & 0 & 1 & STO-3G & 20 \\
\ce{H1O1} & -1 & 1 & STO-3G & 12 \\
\ce{H1} & 0 & 2 & STO-3G & 2 \\
\ce{H2Be1} & 0 & 1 & STO-3G & 14 \\
\ce{H2C1O1} & 0 & 1 & STO-3G & 24 \\
\ce{H2C1} & 0 & 3 & 3-21G & 26 \\
\ce{H2C1} & 0 & 3 & STO-3G & 14 \\
\ce{H2C1} & 0 & 3 & STO-3G & 14 \\
\ce{H2C2} & 0 & 1 & STO-3G & 24 \\
\ce{H2Mg1} & 0 & 1 & STO-3G & 22 \\
\ce{H2O1} & 0 & 1 & STO-3G & 14 \\
\ce{H2O2} & 0 & 1 & STO-3G & 24 \\
\ce{H2S1} & 0 & 1 & STO-3G & 22 \\
\ce{H2} & 0 & 1 & 3-21G & 8 \\
\ce{H2} & 0 & 1 & 6-311G** & 24 \\
\ce{H2} & 0 & 1 & 6-311G & 12 \\
\ce{H2} & 0 & 1 & 6-31G** & 20 \\
\ce{H2} & 0 & 1 & 6-31G & 8 \\
\ce{H2} & 0 & 1 & STO-3G & 4 \\
\ce{H3N1} & 0 & 1 & STO-3G & 16 \\
\ce{H3} & 0 & 1 & 3-21G & 12 \\
\ce{H3} & 1 & 1 & STO-3G & 6 \\
\midrule
\end{tabular}
%\null\hfill
%\null \\
%\\ \vfill\null \newpage \
%\flushleft
%\null\hfill
%\color{red}
\end{table}
\null\vfill
\begin{table}[h!]
\begin{tabular}[t]{lrrlr}
%\color{red} \toprule 
%\midrule \color{black}
%\null\vspace{0.65ex}
%\null\vspace{0.65ex}
%\null\vspace{0.65ex}
%\null\vspace{0.65ex}
%\null\vspace{0.08em} 
%\null\vspace{0.08em} 
 \phantom{System} & \phantom{Charge} & \phantom{Multiplicity} & \phantom{Basis} & \phantom{Qubits} \\
 %System & Charge & Multiplicity & Basis & Qubits \\
%\\ \\
\midrule 
\ce{H4C1} & 0 & 1 & STO-3G & 18 \\
\ce{H4C2} & 0 & 1 & STO-3G & 28 \\
\ce{H4N1} & 1 & 1 & STO-3G & 18 \\
\ce{He1} & 0 & 1 & STO-3G & 2 \\
\ce{K1} & 0 & 2 & STO-3G & 26 \\
\ce{Li1} & 0 & 2 & STO-3G & 10 \\
\ce{Mg1} & 0 & 1 & STO-3G & 18 \\
\ce{N1} & 0 & 4 & STO-3G & 10 \\
\ce{N2} & 0 & 1 & STO-3G & 20 \\
\ce{Na1} & 0 & 2 & STO-3G & 18 \\
\ce{Ne1} & 0 & 1 & STO-3G & 10 \\
\ce{O1} & 0 & 3 & STO-3G & 10 \\
\ce{O2} & 0 & 1 & STO-3G & 20 \\
\ce{O2} & 0 & 3 & STO-3G & 20 \\
\ce{P1} & 0 & 4 & STO-3G & 18 \\
\ce{S1} & 0 & 3 & STO-3G & 18 \\
\ce{Si1} & 0 & 3 & STO-3G & 18 \\
\bottomrule
\end{tabular}
\null\hfill
\end{table}
\clearpage
\end{document}